\documentclass[11pt]{article}
\usepackage[left=1in, right=1in, top=1in, bottom=1in, margin=1in]{geometry}
\usepackage{booktabs}
\date{}
\usepackage{indentfirst}
\usepackage{amsfonts,enumerate}
\usepackage{amsmath}
\usepackage{amsthm}
\usepackage{verbatim}
\usepackage[makeroom]{cancel}
\usepackage{rotating}
\usepackage{float}
\usepackage{algorithm, caption}
\usepackage[noend]{algpseudocode}
\usepackage{makecell}

\usepackage{multirow}
\usepackage{caption}   
\usepackage{multicol}
\usepackage{enumerate}

\usepackage{tikz}
\usetikzlibrary{quantikz2}

\usetikzlibrary{trees}

\usepackage{amssymb}
\usepackage{mathrsfs}
\usepackage{epsfig}
\usepackage{graphicx}
\usepackage{indentfirst}
\usepackage{framed}
\usepackage[numbers,sort&compress]{natbib}
\usepackage{color}
\usepackage[backref=page]{hyperref}
\usepackage[capitalize]{cleveref}
\usepackage{enumerate}
\usepackage[shortlabels]{enumitem}
\usepackage{bm}

\usepackage[normalem]{ulem}
\usepackage{enumitem}
\newcommand{\ketbra}[2]{\lvert #1 \rangle \! \langle #2 \rvert} %

\usepackage{titlesec}

\usepackage[labelfont=bf]{caption}
\usepackage{soul}
\usepackage{enumerate}

\newcommand{\norm}[1]{\lVert#1\rVert}

\newcommand{\Exp}{\mathbb{E}}

\newcommand{\FF}{\mathbb{F}}

\def\01{\{0,1\}}

\hypersetup{
colorlinks,
linkcolor={blue!100!black},
citecolor={red!100!black},
}

\newtheorem{theorem}{Theorem}
\numberwithin{theorem}{section}
\newtheorem{problem}[theorem]{Problem}
\newtheorem{lemma}[theorem]{Lemma}
\newtheorem{corollary}[theorem]{Corollary}
\newtheorem{proposition}[theorem]{Proposition}

\newtheorem{fact}[theorem]{Fact}
\newtheorem{claim}[theorem]{Claim}

\newtheorem{result}[theorem]{Result}

\theoremstyle{definition}

\usepackage[margin=1in]{geometry}

\DeclareMathOperator{\poly}{poly}

\renewcommand{\Pr}{\mbox{\rm Pr}}

\DeclareMathOperator{\Tr}{Tr}
\newcommand{\ra}{\rangle}
\newcommand{\la}{\langle}

\newcommand{\R}{\mathbb{R}} 
\newcommand{\Cc}{\mathcal{C}} 
\newcommand{\Se}{\mathcal{S}} 

\newcommand{\F}{\mathbb{F}} 
\newcommand{\pmset}[1]{\{-1,1\}^{#1}} 

\newcommand{\Id}{\ensuremath{\mathop{\rm Id}\nolimits}}

\newcommand{\eps}{\varepsilon}

\DeclareMathOperator{\Var}{Var}

\DeclareMathOperator{\spec}{Spec}

\DeclareMathOperator{\op}{op}
\newcommand{\id}{I}

\newcommand{\calH}{\mathcal H}
\newcommand{\calT}{\mathcal T}

\newcommand{\calE}{{\cal E }}

\newcommand{\wt}{\mathsf{wt}}

\newcommand{\commutant}{\mathsf{C}}
\newcommand{\commutantperp}{\mathsf{A}}
\newcommand{\topenergy}{\mathsf{TopEnergy}}

\newcommand{\beq}{\begin{equation}}
\newcommand{\eeq}{\end{equation}}
\newcommand{\beqn}{\begin{equation*}}
\newcommand{\eeqn}{\end{equation*}}
\newcommand{\beqr}{\begin{eqnarray}}
\newcommand{\eeqr}{\end{eqnarray}}
\newcommand{\beqrn}{\begin{eqnarray*}}
\newcommand{\eeqrn}{\end{eqnarray*}}
\newcommand{\bmline}{\begin{multline}}
\newcommand{\emline}{\end{multline}}
\newcommand{\bmlinen}{\begin{multline*}}
\newcommand{\emlinen}{\end{multline*}}

\makeatletter
\newtheorem*{rep@theorem}{\rep@title}
\newcommand{\newreptheorem}[2]{%
\newenvironment{rep#1}[1]{%
\def\rep@title{#2 \ref{##1}}%
\begin{rep@theorem}}%
{\end{rep@theorem}}}
\makeatother

\newreptheorem{lemma}{Lemma}
\newreptheorem{theorem}{Theorem}
\newreptheorem{corollary}{Corollary}
\newreptheorem{proposition}{Proposition}
\newreptheorem{conjecture}{Conjecture}
\usepackage[capitalize]{cleveref}

\definecolor{applegreen}{rgb}{0.0, 0.5, 0.0}

\makeatletter
\def\widebreve{\mathpalette\wide@breve}
\def\wide@breve#1#2{\sbox\z@{$#1#2$}%
     \mathop{\vbox{\m@th\ialign{##\crcr
\kern0.08em\brevefill#1{0.8\wd\z@}\crcr\noalign{\nointerlineskip}%
                    $\hss#1#2\hss$\crcr}}}\nolimits}
\def\brevefill#1#2{$\m@th\sbox\tw@{$#1($}%
  \hss\resizebox{#2}{\wd\tw@}{\rotatebox[origin=c]{90}{\upshape(}}\hss$}
\makeatletter
\title{Testing and learning structured quantum Hamiltonians}
\begin{document}
\author{
Srinivasan Arunachalam\\[2mm]
IBM Quantum\\
\small Almaden Research Center, CA, USA\\
\small \texttt{Srinivasan.Arunachalam@ibm.com}
\and
Arkopal Dutt\\[2mm]
IBM Quantum\\
\small   Cambridge, MA, USA\\
\small \texttt{arkopal@ibm.com}
\and
Francisco Escudero \\ Gutiérrez \\[2mm] 
Qusoft and CWI \\ 
\small Amsterdam, Netherlands \\
\small \texttt{feg@cwi.nl}
}

\maketitle 
\begin{abstract}
    We consider the problems of testing and learning an $n$-qubit  Hamiltonian $H=\sum_x \lambda_x \sigma_x$ expressed in its Pauli basis, from queries to its evolution operator $U=e^{-iHt}$. To this end, we prove the following results.
    \begin{enumerate}
        \item \textbf{Testing}: We give a \emph{tolerant}  testing protocol to decide if a Hamiltonian is $\eps_1$-close to $k$-local or $\eps_2$-far from $k$-local in the $\ell_2$ norm of the coefficients, with $O(1/(\eps_2-\eps_1)^{4})$ queries, thereby solving two open questions posed in a recent work by Bluhm, Caro and Oufkir~\cite{bluhm2024hamiltonianv1}.  We give a protocol for testing whether a Hamiltonian is $\eps_1$-close to being $s$-sparse or $\eps_2$-far from being $s$-sparse in the $\ell_2$ norm of the coefficients, with $O(s^{6}/(\eps_2^2-\eps_1^2)^{6})$ queries. 
        \item \textbf{Learning}: We give a protocol to $\eps$-learn unstructured Hamiltonian in the $\ell_\infty$ norm of the coefficients with $O(1/\eps^4)$ queries. Combining this with the non-commutative Bohnenblust-Hille inequality, we obtain an algorithm for learning $k$-local Hamiltonians in $\ell_2$ norm of the coefficients that only uses $O(\exp(k^2+k\log(1/\eps)))$ queries. For Hamiltonians that are $s$-sparse in the Pauli basis, we can learn them in the $\ell_2$ norm with $O(s^2/\eps^4)$ queries. 
        \item \textbf{Learning without quantum memory}: The learning results stated above have no dependence on the system size $n$, but require $n$-qubit quantum memory. 
        We give subroutines that  allow us to reproduce all the above learning results without quantum memory; squaring the query complexity and paying a $(\log n)$-factor in the local case and an $n$-factor in the sparse case. 
        \item \textbf{Testing without quantum memory}: We give a new subroutine called \emph{Pauli hashing}, which allows one to tolerantly test $s$-sparse Hamiltonians in $\ell_2$ norm using $\tilde{O}(s^{14}/(\varepsilon_2^2-\varepsilon_1^2)^{18})$ query complexity. A key ingredient is showing that $s$-sparse Pauli channels can be tested in a tolerant fashion as being $\varepsilon_1$-close to being $s$-sparse or $\varepsilon_2$-far under the diamond norm, using $\tilde{O}(s^2/(\varepsilon_2-\varepsilon_1)^6)$ queries via Pauli hashing.
    \end{enumerate}

In order to prove these results, we prove new structural theorems for local Hamiltonians, sparse Pauli channels and sparse Hamiltonians. We complement our learning algorithms with lower bounds that are polynomially weaker. Furthermore, our algorithms use short time evolutions and do not assume prior knowledge of the terms on which the Pauli spectrum is supported on, i.e., we do not require prior knowledge about the \emph{support} of the Hamiltonian terms. 
\end{abstract}
\newpage 
{\small \tableofcontents}
\newpage

\section{Introduction}
A fundamental and important challenge with building quantum devices is being able to characterize and calibrate its behavior. One approach to do so is \emph{Hamiltonian learning} which seeks to learn the Hamiltonian governing the dynamics of a quantum system given finite classical and quantum resources. Beyond system characterization, it is also carried out during validation of physical systems and designing control strategies for implementing quantum gates~\cite{innocenti2020supervised}. However, learning an $n$-qubit Hamiltonian is known to be difficult, requiring complexity that scales exponential in the number of qubits unless a coarse metric is used~\cite{caro2023learning}. 

In practice, prior knowledge on the structure of Hamiltonians is available e.g., those of engineered quantum devices~\cite{sheldon2016procedure} where the underlying Hamiltonians primarily involve local interactions with few non-local interactions, and even naturally occurring physical quantum systems such as those with translationally invariant Hamiltonians. To highlight these structural properties, consider an $n$-qubit Hamiltonian $H$ (which is a self-adjoint operator acting on $(\mathbb C^{2})^{\otimes n}$) expanded in terms of the $n$-qubit Pauli operators:
\begin{equation*}
    H=\sum_{x\in \01^{2n}} \lambda_x\sigma_x,
\end{equation*}
where $\lambda_x$ are real-valued coefficients (also called interaction strengths) of the Pauli operators $\sigma_{x}$ denoted by the string $x=(a,b)\in \F_2^{2n}$ with  $\sigma_{(a,b)}= i^{a \cdot b} \otimes_{i=1}^n X^{a_i} Z^{b_i}$. We call the set of Paulis with non-zero coefficients $\lambda_x$ as the Pauli spectrum of the Hamiltonian denoted by $\mathcal{S}=\{x\in \01^{2n}:\ \lambda_x\neq 0\}$. Of particular relevance are \emph{$k$-local} Hamiltonians which involve Pauli operators that act non-trivially on all but at most $k$ qubits and \emph{$s$-sparse} Hamiltonians whose Pauli expansion contains at most $s$ non-zero Pauli operators i.e., $|\mathcal{S}|\leq s$.

There has thus been a growing suite of Hamiltonian learning results that have shown that when the underlying $n$-qubit Hamiltonian $H$ satisfies these structural properties, learning can be performed with only $\poly(n)$ query complexity, either by making ``queries" to the unitary evolution operator $U(t)=\exp(-iHt)$~\cite{Silva2011Practical, holzapfel2015scalable, Zubida2021Optimal, haah2022optimal, yu2023robust, Dutkiewicz.2023, huang2023heisenberg, li2023heisenberglimited, franca2024efficient, Gu2022Practical,zhao2024learning,hu2025ansatz}, or by assuming one has access to Gibbs state~\cite{anshu2021sample, haah2022optimal, rouze2023learning, onorati2023efficient, bakshi2023learning, Gu2022Practical}. Notably, \cite{bakshi2024structure} considered the problem of learning Hamiltonians that are both local and sparse, without prior knowledge of the support. Several of the learning algorithms mentioned above however require assumptions on the support of the Hamiltonian beyond locality or sparsity, such as \cite{huang2023heisenberg} which considers \emph{geometrically-local} Hamiltonians (a subset of local Hamiltonians) and \cite{yu2023robust} which requires assumptions on the~support.

Moreover, before learning, it might be desirable to uncover \emph{what is the structure} of an unknown Hamiltonian in order to choose specialized learning algorithms. Even deciding if a Hamiltonian has a particular structure is a fundamental challenge and constitutes the problem of \emph{testing} if an unknown Hamiltonian satisfies a certain structural property. As far as we know, this line of investigation is nascent with only a few works on Hamiltonian \emph{property} testing \cite{she2022unitary,aharonov2022quantum,laborde2022quantum} with Blum et al.~\cite{bluhm2024hamiltonianv2} having considered the problem of testing local Hamiltonians and the problem of testing sparse Hamiltonians yet to be tackled. This leads us to the motivating question of our~work:
\begin{quote}
\center \emph{{What is the query complexity of learning and testing structured Hamiltonians?}}
\end{quote}

\subsection{Problem statement}
Before we state our results answering the question above, we clearly mention our learning and testing problems first. 
If $H$ is the Hamiltonian describing the dynamics of a certain physical system, then the state of that system evolves according to the \emph{time evolution operator} $U(t)=e^{-iHt}$. This means that if $\rho(0)$ is the state at time $0$, at time $t$ the state would have evolved to $\rho(t)=U(t)\rho(0)  U^{\dagger}(t)$. Hence, to test and learn a Hamiltonian one can do the following: prepare a desired state, apply $U(t)$ or tensor products of $U(t)$ with identity to the state, and finally measure in a chosen basis. From here onwards, this is what we mean by \emph{querying} the unitary $U(t)$. It is usual to impose the normalization condition $\norm{H}_{\mathrm{op}}\leq 1$ (i.e., that the eigenvalues of $H$ are bounded in absolute value by $1$). We will assume this normalization unless otherwise stated, but we will also work out the dependence on $\norm{H}_{\op}$ for our learning algorithms. Throughout this paper, we will consider the normalized Frobenius norm as the distance between Hamiltonians, unless otherwise stated. This distance is 
$$
d(H,H')=\norm{H-H'}_2=\sqrt{\frac{\Tr[(H-H')^2]}{2^n}},
$$and it equals the $\ell_2$-norm of the Pauli spectrum, $d(H,H')=\sqrt{\sum |\lambda_x-\lambda'_{x}|^2}$. A \emph{property} of a Hamiltonian, denoted $\mathcal{H}$ is a class of Hamiltonians that satisfy the property (here we will be interested in sparse and local properties). We say that $H$ is $\eps$-far from having a property $\mathcal{H}$ if $d(H,H')>\eps$ for every $H'\in\mathcal{H}$, and otherwise is $\eps$-close. Now, we are ready to state the testing and learning problems. 
\vspace{2mm}

\fbox{\begin{minipage}{40em}
Let $\mathcal{H}$ be a property and let $H$ be an unknown Hamiltonian with $\norm{H}_{\mathrm{op}}\leq 1$ and $\Tr[H]=0$.
\begin{problem}[Tolerant testing]\label{prob:localitytesting}
    Promised $H$ is either $\eps_1$-close or $\eps_2$-far from satisfying property $\mathcal{H}$, decide which is the case  by making queries to $U(t)$.
\end{problem}
\begin{problem}[Hamiltonian learning]\label{prob:locallearning}
     Promised $H\in \mathcal{H}$, output a classical description of $\widetilde{H}\in\mathcal H$ such that  $\|H-\widetilde{H}\|_2\leq \eps$ by making queries to $U(t)$.
\end{problem}
\end{minipage}
}

\subsection{Summary of results}

The main results of this work are query-efficient algorithms for testing and learning Hamiltonians that are local and/or sparse. We can reproduce these results without using quantum memory by increasing the number of queries.  We summarize our results in the \cref{tab:summary_results_paper}.

{\renewcommand{\arraystretch}{1.3}\begin{table}[h]
\footnotesize
\centering
\begin{tabular}{c | c  | c  |c|c|} 
\cline{2-5}
\multirow{2}{*}{}& \multicolumn{2}{c|}{\textbf{Testing}} & \multicolumn{2}{c|}{\textbf{Learning}} \\
\cline{2-5}
& \emph{with memory} &\emph{w/o memory} & \emph{with memory} & \emph{w/o memory} \\ [0.3ex] 
\hline
\multicolumn{1}{|c|}{ $s$-sparse} & \makecell{$\poly(s)$ \\ \cref{theo:sparsitytesting}}  & \makecell{$\poly(s)$\\\cref{thm:sparsity_testing_pauli_hashing_ham}}   & \makecell{$\poly(s)$\\ \cref{theo:sparselearning}} & \makecell{$n\cdot \poly(s)$\\ \cref{theo:learnnomemory}}\\ 
\hline
\multicolumn{1}{|c|}{$k$-local}   & \makecell{$O(1)$\\ \cref{theo:localitytesting}} & \makecell{$O(1)$\\ \cite{bluhm2024hamiltonianv2}} & \makecell{$\exp(k^2)$\\ \cref{theo:locallearning}} &\makecell{$(\log n)\cdot \exp(k^2)$\\ \cref{theo:learnnomemory} } \\ \hline
\multicolumn{1}{|c|}{$k$-local \& $s$-sparse}   & $\poly(s)$ & $\poly(s)$ & $\min\{\exp(k^2),\poly(sk)\}$ & $(\log n)\cdot \min\{\exp(k^2),\poly(sk)\}$ \\ \hline
\end{tabular}
\caption{Query complexity for learning and testing $n$-qubit structured Hamiltonians. Dependence on $n$ and the structural property is shown for constant accuracy. Results are indicated with quantum memory (i.e.,  an $n$-qubit ancillary system is available) and without quantum memory.} 
\label{tab:summary_results_paper}
\end{table}
}

Before we discuss our results in more detail, we make a few remarks about our main results.
\begin{enumerate}
    \item [$(i)$] As far as we know, this is the first work: $(a)$  with complexities that are \emph{independent} of~$n$ (with memory)\footnote{There are a few works that achieve $n$-independent complexities for learning local Hamiltonians in the $\infty$-norm of the Pauli coefficients, but when transformed into $2$-norm learners they yield complexities depending on $n^k$.}, and
    $(b)$ that does not assume knowledge of the support.\footnote{Soon after the third-author's work~\cite{gutierrez2024simple}, Bakshi et al.~\cite{bakshi2024structure} presented a learning algorithm that does not require prior knowledge of the support, achieving Heisenberg scaling using heavy machinery.}
    \item [$(ii)$] We give the first learning algorithm for Hamiltonians that are only promised to be sparse, and not necessarily local. Similarly, our local Hamiltonian learning problem doesn't assume geometric locality which was assumed in several prior works.
    \item [$(iii)$] Our  testing algorithms are tolerant, i.e., they can handle the setting where $\varepsilon_1\neq 0$. As far as we know, there are only a handful of polynomial-time tolerant~testers for quantum objects.
    \item [$(iv)$] We show that all our learning protocols with quantum memory can be translated to ones which require no quantum memory. In the case of learning structured Hamiltonians, we obtain a protocol with only a factor $\log n$ overhead for local Hamiltonians and a protocol with a factor $n$ overhead for sparse Hamiltonians. 
    \item [$(v)$] We also give a tolerant testing algorithm for $s$-sparse Hamiltonians that requires no quantum memory based on a new subroutine called Pauli hashing. The query complexity is $O(\poly(s))$ and is notably independent of dimension $n$. 
    \item [$(vi)$] Our learning algorithms are based on a subroutine that estimates arbitrary $n$-qubit Hamilotmians with $O(1/\eps^4)$ queries, albeit in the coarser metric of the $\ell_\infty$-norm of the Pauli coefficients. As far as we know, this is the best result for unstructured Hamiltonians.  Notably, it is also the first time-efficient proposal for this problem.
\end{enumerate}
We remark that most  previous work on Hamiltonian learning (that we highlighted earlier) are done under the distance induced by the supremum norm of the Pauli spectrum and with extra constraints apart from locality \cite{Silva2011Practical, holzapfel2015scalable, Zubida2021Optimal, haah2022optimal, wilde2022scalably, yu2023robust, caro2023learning, Dutkiewicz.2023, huang2023heisenberg, li2023heisenberglimited, möbus2023dissipationenabled, franca2024efficient, Gu2022Practical}. When transformed into learning algorithms under the finer distance induced by the  $\ell_2$-norm of the Pauli spectrum, these proposals yield complexities that depend polynomially on $n^k$ and only work for a restricted family of $k$-local Hamiltonians. The works that explicitly consider the problem of learning under the $\ell_2$-norm have complexities depending on $n$ and assume a stronger access model \cite{castaneda2023hamiltonian,bakshi2024structure}.

\subsection{Results}
\paragraph{Testing.} Recently, Bluhm, Caro and Oufkir proposed a non-tolerant testing algorithm, meaning that it only works for the case $\eps_1=0,$ whose query complexity is $O(n^{2k+2}/(\eps_2-\eps_1)^4)$ and with total evolution time $O(n^{k+1}/(\eps_2-\eps_1)^3)$. They posed as open questions whether the dependence on $n$ could be removed and whether an efficient tolerant-tester was possible \cite[Section 1.5]{bluhm2024hamiltonianv1}. Our first result gives positive answer to both questions. 

\begin{result}
     There is an algorithm that solves~\cref{prob:localitytesting} for $k$-local Hamiltonians by making $\poly(1/(\eps_2-\eps_1))$ queries to the evolution operator and with $\poly(1/(\eps_2-\eps_1))$ total evolution time. 
\end{result}

See Theorem~\ref{theo:localitytesting} for a formal statement of this result. Our algorithm to test for locality is simple. It consists of repeating the following process $1/(\eps_2-\eps_1)^4$ times: prepare $n$ EPR pairs, apply $U(\eps_2-\eps_1)\otimes \Id_{2^n}$ to them and measure in the Bell basis. Each time that we repeat this process, we sample from the  Pauli sprectrum of $U(\eps_2-\eps_1)$.\footnote{The Pauli spectrum of a unitary $U=\sum_x \widehat{U}_x\sigma_x$ determines a probability distribution because $\sum_x |\widehat{U}_x|^2=1$.} As $\eps_2-\eps_1$ is  small, Taylor expansion ensures that $U(\eps_2-\eps_1)\approx \Id_{2^n}-i(\eps_2-\eps_1) H$, so sampling from the Pauli spectrum of $U(\eps_2-\eps_1)$  allows us to estimate the weight of the non-local terms of $H$. If that weight is big, we output that $H$ is far from $k$-local, and otherwise we conclude that $H$ is close to $k$-local.

Classically there have been many papers studying the problem of testing and learning sparse Boolean functions \cite{gopalan2011sparsity,negahban2012learning,yaroslavtsev2020fast,eskenazis2022low}, however there are not many results on learning \emph{sparse Hamiltonians} (and not necessarily local). The only testing result that we are aware of requires $O(sn)$ queries and time $O(4^{ns})$~\cite[Remark B.2]{bluhm2024hamiltonianv2}. Here, we present  the first sparsity time-efficient testing algorithm whose query complexity does not depend on $n$.

\begin{result}
     There is an algorithm that solves~\cref{prob:localitytesting} for $s$-sparse Hamiltonians by making $\poly(s/(\eps_2-\eps_1))$ queries to the evolution operator and with $\poly(s/(\eps_2-\eps_1))$ total evolution time. 
\end{result}
See Theorem~\ref{theo:sparsitytesting} for a formal statement. This testing algorithm consists on performing Pauli sampling of $U(\sqrt{(\eps_2^2-\eps_1^2)/s})$ a total of $O(s^4/(\eps_2^2-\eps_1^2)^{4})$ times. From these samples one can estimate the sum of the squares of the top $s$ Pauli coefficients of $U$. If this quantity is big enough, we output that the Hamiltonian is close to $s$-sparse, and otherwise that is far. Although from this high-level description the algorithm seems similar to the locality testing one, the analysis is more involved and requires taking the second order Taylor expansion, which is the reason why the dependence on $(\eps_2-\eps_1)$ is worse in this case. 

Additionally, we provide a sparsity tester (\cref{theo:sparsitytestnottol}) that  only makes $O(s^2/\eps_2^4)$ queries with $O(s^{1.5}/\eps_2^3)$ total evolution time, but only works in the regime $\eps_1=O(\eps_2/\sqrt{s})$. 

\paragraph{Learning.} We first propose a protocol to learn unstructured Hamiltonians efficiently in the coarser $\ell_\infty$ norm of the Pauli coefficients. Then, we turn it into a learner in the $\ell_2$ norm for local and sparse Hamiltonians. In particular, we propose the first learning algorithm for sparse~Hamiltonians which does not make any assumptions regarding the support of the Hamiltonian beyond sparsity.\footnote{A concurrent work also dealt with the problem of learning sparse Hamiltonians \cite{zhao2024learning}. See \cref{tab:comparison} for a comparison.} 

\begin{result}\label{res:unslearn}
    There is an algorithm that outputs estimates $\widetilde \lambda_x$ such that $|\lambda_x-\widetilde\lambda_x|\leq \eps$ for every $x\in\{0,1\}^{2n}$ by making $O(1/\eps^4)$ queries to the evolution operator with $O(1/\eps^3)$  total evolution time. 
\end{result}

See~\cref{theo:unstructuredlearning} for a formal result. The learning algorithm  has two stages. In the first stage one samples from the Pauli distribution of $U(\eps)$, as in the testing algorithm, and from that one can detect which are the big Pauli coefficients of $H$. In the second stage we learn the large Pauli coefficients via a subroutine based on Clifford Shadows (see \cref{lem:ShadowHamEstimation}). This subroutine allows to learn a set of $m$ Pauli coefficients of a Hamiltonian with $\log(m)$ queries to its time evolution operator, which may be of independent interest. For Hamiltonians that are $k$-local, we have the following learning result in the $\ell_2$-norm. 

\begin{result}\label{res:loclearn}
    There is an algorithm that solves~\cref{prob:locallearning}  for $k$-local Hamiltonians by making $\exp(k^2+k\log (1/\eps))$ queries to the evolution operator with $\exp(k^2+k\log (1/\eps))$  total evolution time. 
\end{result}

See Theorem~\ref{theo:locallearning} for a formal statement of this result. In the case that the Hamiltonian is $k$-local, one can ensure that the coefficients not detected as big in the first stage of the algorithm of \cref{res:unslearn} have a neglectable contribution to the $\ell_2$-norm, from which \cref{res:loclearn} follows. To argue this formally, we use the non-commutative Bohnenblust-Hille inequality, which has been used recently for various quantum learning algorithms~\cite{huang2023learning,volberg2023noncommutative}.  For Hamiltonians that are $s$-sparse, we have the following learning result in the $\ell_2$-norm. 

\begin{result}\label{res:spalearn}
    There is an algorithm that solves~\cref{prob:locallearning}  for $s$-sparse Hamiltonians by making $\poly(s/\varepsilon)$ queries to the evolution operator with $\poly(s/\varepsilon)$  total evolution time. 
\end{result}

See Theorem~\ref{theo:sparselearning} for a formal statement. \cref{res:spalearn} follows by adding a rounding step to the algorithm of \cref{res:unslearn} that ensures that all zero coefficients of the Hamiltonians are also zero for the approximating Hamiltonian.

\vspace{-1em}
\paragraph{Learning and testing without quantum memory.}
Motivated by the limitations of current devices, there has been a series of recent works to understand the power of quantum memory in testing and learning tasks, exhibiting exponential separations in some cases~\cite{chen2022exponential,chen2024optimal,chen2024optimalent}. A natural question is, if the problems that we mentioned above become harder without quantum memory? 

\textbf{Learning without memory.} We surprisingly show that, the learning protocols that we mention above, can be implemented efficiently when one  has \emph{ no quantum memory}.  To this end, we provide two crucial subroutines for $(i)$ estimating the Pauli spectrum of a unitary, $(ii)$ estimating a single Pauli coefficient to make our protocols work in the memory-less setting. Subroutine $(ii)$ incurs in no extra query-cost, and subroutine $(i)$ only incurs in a factor-$n$ overhead in the case of learning $s$-sparse Hamiltonians and a factor $\log(n)$ in the case of learning $k$-local Hamiltonians. These subroutines can also be useful in other contexts. In particular, we propose tolerant tester to decide if an unknown unitary is a $k$-junta  which uses $O(4^k)$ queries (see \cref{prop:testkjuntaunitaries}), making progress on a question of Chen et al.\ \cite[Section 1.3]{chen2023testing}, and then we use subroutine $(i)$ to turn it into a memory-less tester that only makes $O(4^kn)$ queries. 

\paragraph{Testing sparse Pauli channels via Pauli hashing.} In order to test for sparsity of Hamiltonian without memory we reduce to the problem of testing sparsity of a Pauli channel $\Phi:\rho \mapsto\sum_x p(x) \sigma_x \rho \sigma_x$, which is of independent interest. To do that, we introduce a new technique called \emph{Pauli Hashing} which allows to construct random partitions of Pauli operators. 
The high-level idea is to bucket the error rates $p(x)$ and thereby the corresponding Pauli operators: for this, we choose a random subgroup $G$ of the $n$-qubit Pauli group with dimension $t = O(\log s)$. Pauli hashing allows us to partition all the Pauli operators into cosets of the centralizer of $G$,  which is the set of all the Paulis that commute with the elements in $G$. The buckets are then the $O(s)$ cosets of the centralizer of $G$.
The main work then goes into arguing that the sum of the weights of the top $s$ buckets is a good estimate of the top $s$ error rates, and then a structural lemma we prove shows this is a good proxy for indicating whether the Pauli channel is close to being $s$-sparse or not. Putting everything together, with some careful analysis, we get an efficient tolerant tester for $s$-sparse Pauli~channels. 
\begin{result}
    There is an algorithm with no quantum memory that tests if a Pauli channel is $\eps_1$-close to or $\eps_2$-far-from being $s$-sparse in diamond norm by making $\widetilde{O}(s^2/(\varepsilon_2-\varepsilon_1)^6)$  queries to the channel.
\end{result}
See \cref{thm:testing_sparse_PC} for a formal statement. We remark that Pauli Hashing only requires the preparation of Pauli eigenstates and Pauli measurements, making it suitable for the near-term.

\paragraph{Testing sparse Hamiltonians without memory.} We provide a memory-less testing algorithm for $s$-sparse Hamiltonians that uses Pauli hashing, that is completely independent of our tester with memory and only requires $\poly(s/\eps)$ queries and total evolution time, notably avoiding any dependence on $n$. To do this, we reduce the problem of testing Hamiltonian sparsity to testing the sparsity of an associated Pauli channel. To be precise, given the time evolution channel $\mathcal{H}_t:\rho\to U(t)\rho U^\dagger (t)$, we define its \emph{Pauli-twirled channel} via $$\mathcal{H}_t^\calT(\rho)=\mathbb E_{x}[\sigma_x\mathcal{H}_t(\sigma_x\rho\sigma_x)\sigma_x],$$ 
and we prove that $\mathcal{H}_t$ is sparse if and only if $\mathcal{H}_t^\calT$ is sparse. Our result is then as follows. See Theorem~\ref{thm:sparsity_testing_pauli_hashing_ham} for a formal statement. 
\begin{result}
    There is an algorithm with no quantum memory for~\cref{prob:locallearning} for $s$-sparse Hamiltonians by making $\poly(s/\varepsilon)$ queries to the evolution operator with $\poly(s/\varepsilon)$ total evolution~time. 
\end{result}

\paragraph{Lower bounds.} One drawback of our learning and testing algorithms is the exponent of the sparsity parameter $s$, locality parameter $k$ and the tolerance $(\varepsilon_2-\varepsilon_1)$. Reducing to classical Boolean functions, we show lower bounds without memory that certify that the dependence on these parameters cannot be completely avoided, but an interesting and important future direction is to obtain the optimal results for these near-term relevant problems.\footnote{We remark that Bakshi et al.~\cite{bakshi2024structure} used highly non-trivial ideas to get Heisenberg scaling for their learning algorithm, and potentially similar ideas could be useful here.}

\textbf{Direct comparison to previous work.} 
Comparing the plethora of Hamiltonian learning algorithms can be challenging due to the different assumptions on the structure of the Hamiltonians (local, sparse, geometrical structures, etc.), the different distances to measure the error ($\ell_\infty$ norm of the coefficients, $\ell_2$ norm, etc.), the different complexity measures (queries, total evolution time, number of experiments, etc.), the different access models (coherent/non-coherent queries, with/without memory, etc.) and the different goals of the algorithm (minimizing the dependence on the dimension parameters like $n, s, k$, achieving the Heisenberg scaling $1/\eps$, etc.). Thus, we only include a direct comparison in \cref{tab:comparison} with the works that explicitly consider the same structure and the same error metric as us. As a summary, one can say that for constant $\eps$ our results achieve better dependence on the parameters $n,s,k$ than previous work, while also using the weaker model of incoherent queries, where one can perform only one query before measuring, as opposed to the coherent query model. We also want to remark that our result for learning unstructured Hamiltonian is time efficient, while the, to the best of our knowledge, only previous one is not~\cite{caro2023learning}.

\begin{table}[h!]
    \centering
    \makebox[\linewidth][c]{%
    \begin{tabular}{|c| c c c c |}
        \toprule
         Hamiltonians & Reference & $t_{total}$ & Queries    & Access model \\[0.5mm]
        \hline
        \multirow{2}{*}{Unstructured, $\ell_\infty$ error} &\cite{caro2023learning}& $n/\varepsilon^{4}$ & $n/\varepsilon^{4}$  & Coherent queries 
        \\[0.5mm]
         & \cref{theo:unstructuredlearning}  &   $1/\eps^3$ & $1/\varepsilon^4$     &   Incoherent queries \\ [0.5mm]
        \hline 
        \multirow{3}{*}{$s$-sparse, $\ell_\infty$ error} 
        &\cite{zhao2024learning}$^*$  & $1/\varepsilon^4$ & $1/\varepsilon^8$  &  Coherent queries\\[0.5mm]
        &\cite{hu2025ansatz}$^\dagger$  & $s^2/\varepsilon$ & $s^2/\varepsilon$  &  Coherent queries\\[0.5mm]      
        & \cref{theo:sparselearning}  &   $1/\eps^3$ & $1/\varepsilon^4$     &   Incoherent queries \\[0.5mm] 
        \hline 
         \multirow{3}{*}{$k$-local, $\ell_2$ error} 
         &\cite{castaneda2023hamiltonian}  & $n^k/\varepsilon^2$ & $n^k/\varepsilon^2$   &  Controlled and inverse queries \\[0.5mm]
        &\cite{ma2024learning}$^\circ$  & $(9n)^k/\varepsilon$ & $(27n^3)^{k}/\varepsilon^2$  &  Coherent queries \\[0.5mm]
        & \cref{theo:locallearning}  &   $\exp(k^2)/\eps^k$ & $\exp(k^2)/\eps^k$   &   Incoherent queries       \\ 
        \bottomrule
    \end{tabular}
    }
    \captionsetup{justification=raggedright, singlelinecheck=false} 
    \caption{Comparison of  algorithms for learning Hamiltonians with $\norm{H}_{\op}\leq 1$.\newline
    \footnotesize \textsuperscript{*} It can be improved to $O(1/\eps^{2+o(1)})$ total evolution time and $O(1/\eps^{6+o(1)})$ queries by paying huge constant factors.\\
    \textsuperscript{$\dagger$} This algorithm works for Hamiltonians with $\sup_x|\lambda_x|\leq 1,$ a weaker constraint than $\norm{H}_{\op}\leq 1$.\\
    \textsuperscript{$\circ$} This algorithm is the only one in the table that uses no quantum memory. We provide an algorithm with no quantum memory for $k$-local learning that performs as the one in the last row, but with an extra factor $\log n.$}
    \label{tab:comparison}
\end{table}

\subsection{Discussion and open questions}
Our work opens up several interesting directions which we state here and leave for future work.
\begin{enumerate}
    \item \textbf{Dependence on parameters $\varepsilon_1,\varepsilon_2$.} Our main objective in this work was to obtain query complexities for testing and learning with good dependence on the structural parameters. It is natural to ask if we could improve the dependence on the error parameters and perhaps achieve Heisenberg limited scaling as has been shown to be possible in some particular cases for Hamiltonian learning~\cite{huang2023heisenberg,bakshi2024structure}.
    \item \textbf{Robustness to SPAM noise.} It would be desirable to make the protocols introduced in this work to be robust to SPAM noise. A potential approach is to adapt strategies in~\cite{flammia2020efficient}.
    \item \textbf{Adaptivity.} For learning structured Hamiltonians, adaptive strategies~\cite{granade2012robust,dutt2023active} can improve query complexity by shedding constant factors over baseline learning algorithms, thereby improving performance in practice. Another direction is to then explore adaptive protocols for testing structured Hamiltonians and the performance gains they may bring.
    \item \textbf{Testing and learning with limited quantum memory.} For estimating properties of quantum states, Chen et al. \cite{chen2022exponential} showcased the utility of the resource of quantum memory or a $k$-qubit ancillary system ($k < n$). Large separations in query complexity when learning with memory (even for $k \ll n$) and without memory have been reported for learning Pauli channels~\cite{chen2022channel,chen2023paulichannels} and shadow tomography~\cite{chen2024optimal}. We could thus imagine having access to only limited quantum memory during learning or testing structured Hamiltonians as well. However, it should be noted that given the separation between the query complexities (see Table~\ref{tab:summary_results_paper}) with access to $n$-qubit quantum memory and without any, only marginal gains in complexity are expected from having access to limited quantum memory.
    \item \textbf{Testing and learning Hamiltonians from Gibbs states.} Another natural learning model is that of having access to copies of the Gibbs state of a quantum Hamiltonian at a certain inverse temperature. There has been a suite of work investigating learning local Hamiltonians from Gibbs states~\cite{anshu2021sample,bakshi2023learning} but answering the question of testing structured Hamiltonians given access to copies of the Gibbs state remains wide open.
\end{enumerate}
\vspace{-1em}
\paragraph{Note added.} After sharing \cref{theo:localitytesting} with Bluhm et al., they independently improved the analysis of their testing algorithm and showed that it only requires $O(1/(\eps_2-\eps_1)^3\eps_2)$ queries and $O(1/(\eps_2-\eps_1)^{2.5}\eps_2^{0.5})$ total evolution time, which is very similar to our 
\cref{theo:localitytesting} \cite{bluhm2024hamiltonianv2}. In addition, for a wide range of $k=O(n)$, their algorithm does not require the use of auxiliary qubits.
\vspace{-2em}
\paragraph{Acknowledgements.} S.A. and A.D. thank the Institute for Pure and Applied Mathematics (IPAM) for its hospitality throughout the long program “Mathematical and Computational Challenges in Quantum Computing” in Fall 2023 during which part of this work was initiated. This work was done in part while S.A. was visiting the Simons Institute for the Theory of Computing, supported by DOE QSA grant \#FP00010905. This research was supported by the Europea union’s Horizon 2020 research and innovation programme under the Marie Sk{\l}odowska-Curie grant agreement no. 945045, and by the NWO Gravitation project NETWORKS under grant no. 024.002.003. We thank Amira Abbas, Francesco Anna Mele, Andreas Bluhm, Jop Briët, Matthias Caro, Nunzia Cerrato, Aadil Oufkir, and Daniel Liang for useful comments and discussions. A.D. thanks Patrick Rall for multiple conversations on stabilizer subgroups and Pauli twirling. A.D. thanks Isaac Chuang for discussions on the problem of testing Hamiltonians. F.E.G. is funded by the Deutsche Forschungsgemeinschaft (DFG, German Research Foundation) under Germany's Excellence Strategy – EXC-2047/1 – 390685813. We thank the referees for their detailed comments.
\section{Preliminaries}
\subsection{Notation}
In this section we collect a few well-known facts that we will repeatedly use in our proofs. 
The $1$-qubit Pauli matrices matrices are defined as follows
$$\id=\begin{pmatrix}
1 & 0\\
0 & 1
\end{pmatrix},\ X=\begin{pmatrix}
0 & 1\\
1 & 0
\end{pmatrix},\ Y=\begin{pmatrix}
0 & -i\\
i & 0
\end{pmatrix}\ \mathrm{and}\  Z=\begin{pmatrix}
1 & 0\\
0 & -1
\end{pmatrix}.
$$
It is well-known that the $n$-qubit Pauli matrices $\{\id,X,Y,Z\}^{\otimes n}$ form an  {orthonormal basis} for $\mathbb{C}^{2^n}$.    In particular, for every $x=(a,b)\in \mathbb{F}_2^{2n}$, one can define the \emph{Pauli operator}
$$
\sigma_x=i^{a\cdot b} (X^{a_1}Z^{b_1}\otimes X^{a_2}Z^{b_2} \otimes \cdots \otimes X^{a_n}Z^{b_n}).
$$  
and these operators $\{\sigma_x\}_{x \in \FF_2^{2n}}$ are orthonormal.  Hence, every $n$-qubit operator $H$ can be written down in its Pauli decomposition  as
$$
H=\sum_{x\in \01^{2n}}\lambda_x\sigma_x,
$$
where the real-valued coefficients $\lambda_x$ are given by $\lambda_x=\frac{1}{2^n}\Tr(H\sigma_x)$. Parseval's identity states that the normalized Frobenius norm of $H$, denoted as $\norm{H}_2$, is the  $\ell_2$-norm of its Pauli spectrum,~i.e.,
$$
\norm{H}_2=\sqrt{\frac{\Tr [H^\dagger H]}{2^n}}=\sqrt{\sum_{x\in\01^{2n}}|\lambda_x|^2}.
$$
We will repeatedly use that $\norm{H}_2\leq \norm{H}_{\mathrm{op}}$, which holds because $\norm{H}_2^2$ is the average of the squares of the eigenvalues of $H$. We will also consider the $\ell_{\infty}$ norm of the Pauli coefficients of an operator, which is given by 
$$ \norm{H}_{\ell_\infty}=\sup_x|h_x|.$$
Additionally, we will use $\norm{H}:=\max\{\norm{H}_{\mathrm{op}},1\}$. 

Throughout this work, we will also use the following correspondence between $\01^{2n}$ and $\{0,1,2,3\}^n$: i.e., for $a,b\in \F_2^n$ and $(a,b)\in \01^{2n}$, consider the string $\big((a_1,b_1),\ldots,(a_n,b_n)\big)\subseteq (\F_2^2)^n$ and one can write out $(a_i,b_i)$ as an element in $\{0,1,2,3\}$.

Given $x\in\{0,1,2,3\}^{n}$, define $|x|$ as the number of indices $i\in [n]$ where $x_i\neq 0$,~define 
$$
H_{>k}=\sum_{|x|>k}\lambda_x\sigma_x
$$ and $H_{\leq k}$ as $\sum_{|x|\leq k}\lambda_x\sigma_x$. From the formulation of the 2-norm in terms of the Pauli coefficients it follows that $\norm{H_{>k}}_2\leq \norm{H}_2$.
We note that the distance of a Hamiltonian $H$ from the space of $k$-local Hamiltonians is given by $\norm{H_{>k}}_2$, as $H_{\leq k}$ is the  $k$-local Hamiltonian closest to $H$. The $\ell_2$-distance of $H$ to being $s$-sparse also has a nice expression. Assign labels from $[4^n]$ to $x\in \01^{2n}$ in a way that  and $|\lambda_{x_1}|\geq |\lambda_{x_2}|\dots \geq |\lambda_{x_{4^n}}|$. Then, $\sum_{i\in[s]}\lambda_{x_i}\sigma_{x_i}$ is the closest $s$-sparse Hamiltonian to $H$, so the $\ell_2$-distance of $H$ to the space of $s$-sparse Hamiltonians is $\sqrt{\sum_{i>s}|\lambda_{x_i}|^2}$.

\subsection{Necessary subroutines}
 Suppose $U$ is a unitary and we write out its Pauli decomposition as $U=\sum_x \widehat{U}_x \sigma_x$, then by Parseval's identity  $\sum|\widehat{U}_{x}|^2=\Tr[U^\dagger U]/2^n=1$, i.e.,  $\{|\widehat{U}_{x}|^2\}_x$ is a  \emph{probability distribution}. We will be using the fact below extensively.
 \begin{fact}
 \label{fact:bellsamplingU}
     Given access to a unitary $U$, one can sample from the distribution  $\{|\widehat{U}_x|^2\}_x$.
 \end{fact}
 \begin{proof}
     The proof simply follows by applying $U\otimes \Id_{2^n}$ to $n$ EPR pairs (i.e., preparing the Choi-Jamiolkowski state of $U$) and measuring in the Bell basis, because $$
 U\otimes\Id_{2^n}\ket{\mathrm{EPR}_n}=\sum_{x\in\01^{2n}}\widehat{U}_x \mathop{\bigotimes}_{i\in [n]}(\sigma_{x_i}\otimes \Id_2\ket{\mathrm{EPR}}),
 $$
and the Bell states can be written as $\sigma_x\otimes\Id_{2}\ket{\mathrm{EPR}}$ for $x\in\{0,1,2,3\}$.
 \end{proof}

 We will also use that given a Hamiltonian $H$, the Taylor expansion of the exponential allows us to approximate the time evolution operator as 

\begin{equation}\label{eq:TaylorOrder1}
    U(t)=e^{-itH}=\Id_{2^n}-itH+ct^2 R_1(t)\norm{H}_{\mathrm{op}}^2
\end{equation}
for $t\leq 1/2,$ where the first order remainder $R_1(t)$ is bounded $\norm{R_1(t)}_{\mathrm{op}}\leq 1$ and $c>1$ is a constant. 

\begin{fact}[{\cite{canonne2020short}}]
    Let  $p:\mathcal{X}\to [0,1]$ be a  distribution for a finite set $\mathcal{X}$.  Then, with  prob.~$\geq 1-\delta$, $$O\left(\frac{\log(1/\delta)}{\eps^2}\right)$$
 samples from $p$, the empirical distribution $\tilde p$ satisfies $\max_{x\in\mathcal{X}}|p(x)-\tilde p(x)|\leq~\eps$
\label{thm:canoneproof}
\end{fact}

\begin{theorem}[Clifford shadows \cite{huang2020predicting}]\label{theo:CliffShadows}
    Let $\rho$ be an $n$-qubit state and let $\{O_i\}_{i\in [M]}$ be $n$-qubit traceless observables. Assume that $\sup_i\Tr[O_i^2]=O(1)$. Then, \cref{algo:CliffShad} obtains estimates  $\widetilde O_{i,\rho}$ such that, with probability $1-\delta,$ satisfy $$|\Tr[O_i\rho]-\widetilde O_{i,\rho}]|\leq \eps$$ for every $i\in [M]$. The algorithm uses $O
    \left(\frac{\log(M/\delta)}{\eps^2}\right)$ copies of $\rho$.
\end{theorem}

\begin{algorithm}
\textbf{Input:} Copies of a quantum state $\rho$, target set of observables $\{O_i\}_{i\in [M]}$, error parameter $\eps\in (0,1)$, and failure parameter $\delta\in (0,1)$
\begin{algorithmic}[1]
    \State Set $T= O(\log(M/\delta)/\eps^2)$ and $J=O(\log(M/\delta))$
    \For{$j\in [J]$}
        \For{$k\in [T/J]$}
            \State Apply a uniformly random Clifford gate $C$ to a copy of $\rho$
            \State Measure in the computational basis. Let $\ket{b_{j,k}}$ be the outcome
            \For{$i\in [M]$}
                \State Let $ \widetilde{O}_{i,j,k}=(2^n+1)\bra{b_{j,k}}C^{-1}O_iC\ket{b_{j,k}}$
            \EndFor
        \EndFor
        \For{$i\in [M]$}
            \State Let $\widetilde O_{i,j}=\mathrm{Mean}((\widetilde O_{i,j,k})_k)$
        \EndFor
    \EndFor
    
    \For{$i\in [M]$}
            \State Set $\widetilde O_i:=\mathrm{Median}(( { O}_{i,j})_j)$ 
    \EndFor
\end{algorithmic}
\textbf{Output}: $( \widetilde O_{i})_{i\in\mathcal [M]}$
\caption{Clifford shadows}\label{algo:CliffShad}
\end{algorithm}

\subsection{Concentration inequalities}

We state a few concentration inequalities that we use often. 

\begin{lemma}[Hoeffding bound]\label{lem:hoeffding}
Let~$X_1,\dots,X_m$ be independent-random variables that satisfy $-a_i\leq |X_i|\leq a_i$ for some $a_i>0$.
Then, for any $\tau > 0$, we have
$$
\Pr\Big[\Big|\sum_{i\in [m]} X_i-\sum_{i\in [m]}\mathbb E[X_i]\Big| > \tau\big]
\leq
2\exp\left(-\frac{\tau^2}{2(a_1^2 + \cdots + a_m^2)}\right).
$$
\end{lemma}

\begin{lemma}[Bernstein inequality]\label{lem:Bernstein} Let $X_1,\dots,X_m$ be independent-random variables with $|X_i|\leq M$ for some $M>0$. Then, 
$$
\Pr\Big[\Big|\sum_{i\in [m]} X_i-\sum_{i\in [m]}\mathbb E[X_i]\Big| > \tau\big]
\leq
2\exp\left(-\frac{\tau^2/2}{\sum_{i\in [m]}\mathrm{Var}[X_i]+\tau M/3}\right).
$$  
\end{lemma}

\subsection{Symplectic Fourier analysis}
Consider $x,y \in \FF_2^{2n}$ with $x=(x_1,x_2)$ and $y=(y_1,y_2)$. Define the symplectic inner product~as 
\begin{equation}
    [x,y] = \la x_1, y_2 \ra + \la x_2, y_1 \ra \mod 2.
    \label{eq:symplectic_inner_product}
\end{equation}
Under this notation, observe that
\begin{equation}
    \sigma_{x}\sigma_y=(-1)^{[x,y]}\sigma_y\sigma_x.
\end{equation}
The symplectic Fourier decomposition  for functions $f: \FF_2^{2n} \rightarrow \R$ is defined as
\begin{align*}
    f(x) = \sum_{a \in \FF_2^{2n}} (-1)^{[a,x]} \widebreve{f}(a),
\end{align*}    
where the symplectic Fourier coefficients are defined as
$$
    \widebreve{f}(a) = \frac{1}{4^n} \sum_{x \in \FF_2^{2n}} (-1)^{[a,x]} f(x).
$$

\subsection{Subspaces, stabilizer groups and stabilizer states}
Let $V \subseteq \{0,1\}^{2n}$ be a set of strings corresponding to Paulis. We will denote $\commutant(V)$ as the commutant of the set $V$ (also called centralizer of $V$ when $V$ is a group), defined as $\commutant(V) = \{c : [c,h] =0 \ \forall\, h \in V\}$. We define $\commutantperp(V)$ to be the quotient space $\FF^{2n} / \commutant(V)$. 

\begin{fact}\label{fact:sum_over_symplectic_subspace}
Suppose $V$ is a subspace of $\FF^{2n}$. Then,
$$
\frac{1}{|V|} \sum_{x \in V} (-1)^{[a,x]} = [a \in \commutant(V)],
$$
where $[\cdot]$ is the indicator of the event in the parenthesis.
\end{fact}

We will denote a  stabilizer state corresponding to a stabilizer subgroup $G$ of dimension $k \leq n$~as 
\begin{equation}
    \label{eq:stab_state}
    \rho_G = \frac{1}{2^n} \sum_{g \in G} \sigma_g,
\end{equation}
which will be a pure stabilizer state when dimension of $G$ is $k=n$ i.e., $G$ is an $n$-qubit stabilizer group. When $k < n$, then this will be a mixed state. Another useful state will be the state obtained by action of a Pauli $\sigma_x$ on $\rho_G$ which we will denote by $\rho_{G,x}$ and is given by
\begin{equation}
    \label{eq:stab_state_x}
    \rho_{G,x} = \sigma_x \rho_G \sigma_x = \frac{1}{2^n} \sum_{g \in G} (-1)^{[x,g]} \sigma_g.
\end{equation}
Note that if $x \in \commutant(G)$, then $\rho_{G,x}$ equals $\rho_G$. 

\subsection{Pauli Channels}
An $n$-qubit channel $\cal E$ is called a \emph{Pauli channel} if it can be written as
\begin{equation}
\label{eq:pauli_channel}
\calE(\rho) = \sum_{x \in \{0,1\}^{2n}} p(x) \sigma_x \rho \sigma_x,
\end{equation}
where $p(x)$ are referred to as the \emph{error rates} corresponding to $x \in \{0,1\}^{2n}$. The eigenvalues of a Pauli channel, also called \emph{Pauli fidelities}, are easy to compute as follows
\begin{equation}
\label{eq:pauli_fidelity}
\lambda(y) = \frac{1}{2^{n}} \Tr\left(\sigma_y  \cdot \calE(\sigma_y)\right),
\end{equation}
where we have denoted the Pauli fidelity corresponding to $y \in \{0,1\}^{2n}$ as $\lambda(y)$. Of relevance to us is that the Pauli fidelities and error rates are related via the symplectic Fourier transform as $\widebreve{\lambda}(\alpha) = p(\alpha)$  for all $\alpha \in \{0,1\}^{2n}$. We define $\spec(\mathcal E)=\{x\in\{0,1\}^{2n}:p(x)\neq 0\}.$

We say that a Pauli channel is $s$-sparse if the corresponding set of Pauli coefficients (or error rates) $\{p(x)\}_{x \in \FF_2^{2n}}$ contains at most $s$ non-zero values. We define the \emph{energy} of the top $s$ Pauli coefficients of a given Pauli channel $\calE$ as $\mathsf{Energy}(\calE;s)$, expressed as
\begin{equation}
    \mathsf{Energy}(\calE;s):=\max_{\substack{T\subseteq \FF_2^{2n}:\\|T|=s}}\Big\{\sum_{x \in T} p(x) \Big\}.
    \label{eq:energy}
\end{equation}
We denote the distance between two Pauli channels $\calE_1$ (with error rates $\{p_1(x)\}$) and $\calE_2$ (with error rates $\{p_2(x)\}$) as the total variation distance between the error rates, namely
$$
\mathrm{dist}(\calE_1,\calE_2) =  \frac{1}{2}\norm{p_1 - p_2}_1.
\label{eq:distance_PC}
$$
This distance is equivalent to the diamond distance up a factor 2~\cite[Section A]{magesan2012character}.
We now give a simple formula for the distance of Pauli channels to the set of sparse Pauli~channels.
\begin{claim} \label{lemma:sparse_paulis_struct}
Let $\calE $ be an $n$-qubit Pauli channel. Then, distance $d(\mathcal E,s\text{-sparse})$ of $\calE$ to the set of  $s$-sparse Pauli channels satisfies 
\begin{enumerate}
    \item [$(a)$] $d(\mathcal E,s\text{-sparse})\leq \eps_1\implies \mathsf{Energy}(\mathcal{E};s)\geq 1-2\eps_1$,
    \item [$(b)$] $d(\mathcal E,s\text{-sparse})\geq \eps_2\implies \mathsf{Energy}(\mathcal{E};s)\leq 1-\eps_2.$
\end{enumerate}

\end{claim}
\begin{proof} We begin with $(a)$. If $d(\mathcal E,s\text{-sparse})\leq \eps_1,$ then there is a $s$-sparse probability distribution $\{q(x)\}_x$ such that $$\sum_{x\in\{0,1\}^{2n}}|p(x)-q(x)| \leq 2\eps_1.$$
In particular, if $S$ is the support of $q(x)$, then 
$\sum_{x\notin S}|p(x)| \leq 2\eps_1$,  so $\sum_{x\in S}|p(x)| \geq 1-2\eps_1$, as desired.  Now, we prove $(b).$ Assume $d(\mathcal E,s\text{-sparse})\geq \eps_2.$ Let $S\subseteq \{0,1\}^{2n}$ be the set of the $s$ largest Pauli coefficients of $\mathcal{ E}$. Then, $\mathcal{E}'(\cdot)=\sum_{x\in S}\frac{p(x)}{\mathsf{Energy}(\mathcal E,s)}\sigma_x\cdot\sigma_x$ is a $s$-sparse Pauli channel. As such, it satisfies that 
\begin{align*}
    2\eps_2&\leq \sum_{x\in S}\left|p(x)-\frac{p(x)}{\mathsf{Energy}(\mathcal E,s)}\right|+\sum_{x\notin S}p(x)\\
    &=\frac{1-\mathsf{Energy}(\mathcal E,s)}{\mathsf{Energy}(\mathcal E,s)}\sum_{x\in S}p(x)+\left(1-\sum_{x\in S}p(x)\right)\\
    &=2(1-\mathsf{Energy}(\mathcal E,s)),
\end{align*}
where in the first line we have used that $\mathsf{Energy}(\mathcal E,s)\leq 1$ and that $\sum_{x\in\{0,1\}^{2n}}p(x)=1.$ Now, $(b)$ follows by dividing by $2$ in both sides.
\end{proof}

\section{Technical results}
In this section, we will first prove our main structural theorems for Hamiltonians and provide subroutines which will be used later for testing and learning these structured~Hamiltonians. 

\subsection{Structural lemma for local Hamiltonians}
First, we prove a lemma regarding the discrepancy on the weights of non-local terms of the short-time evolution operator for close-to-local and far-from-local Hamiltonians. 
\begin{lemma}\label{lem:testingdicotomy}
    Let $0\leq \eps_1<\eps_2$. Let $\alpha=(\eps_2-\eps_1)/(3c)$ and $H$ be an $n$-qubit Hamiltonian with $\norm{H}_{\mathrm{op}}\leq 1$, where $c$ is the constant appearing in~\cref{eq:TaylorOrder1}. If $H$ is $\eps_1$-close $k$-local, then 
    $$
    \norm{U(\alpha)_{>k}}_2\leq (\eps_2-\eps_1)\frac{2\eps_1+\eps_2}{9c},
    $$ and if $H$ is $\eps_2$-far from being $k$-local, then  $$\norm{U(\alpha)_{>k}}_2\geq (\eps_2-\eps_1)\frac{\eps_1+2\eps_2}{9c}.$$
\end{lemma}
\begin{proof}
Recall that 
$  U(\alpha)=\Id_{2^n}-i\alpha H+c\alpha^2 R(\alpha)$ by Eq~\eqref{eq:TaylorOrder1} where $\|R\|_{\mathrm{op}}\leq 1$. For simplicity, we set $U=U(\alpha)$ and $R=R_1(\alpha)$. First, assume that $H$ is $\eps_1$-close $k$-local, then by definition we have that $\|H_{>k }\|_2\leq \varepsilon_1$. Then  
$$
\norm{U_{>k}}_2\leq \alpha\norm{H_{>k}}_2+ c\alpha^2\norm{R_{>k}}_2\leq \frac{\eps_2-\eps_1}{3c}\eps_1+c\left(\frac{\eps_2-\eps_1}{3c}\right)^2 =(\eps_2-\eps_1)\frac{2\eps_1+\eps_2}{9c},
$$
    where in the first inequality we have used the triangle inequality, and in the second that $H$ is $\eps_1$-close to $k$-local and that $\norm{R_{>k}}_2\leq \norm{R}_2\leq  \norm{R}_{\mathrm{op}}\leq 1$. Now, assume that $H$ is $\eps_2$-far from being $k$-local (i.e., $\|H_{>k}\|_2\geq \varepsilon_2$). Then 
    $$
    \norm{U_{>k}}_2\geq \alpha\norm{H_{>k}}_2-c\alpha^2\norm{R_{>k}}_2\geq \frac{\eps_2-\eps_1}{3c}\eps_2-c\left(\frac{\eps_2-\eps_1}{3c}\right)^2 \geq (\eps_2-\eps_1)\frac{\eps_1+2\eps_2}{9c},
    $$
    where in first inequality we have used  triangle inequality on $i\alpha H=ct^2 R(\alpha)-U(\alpha)$ to conclude $\alpha\norm{H_{>k}}_2\leq \norm{U_{>k}}_2 + c\alpha^2\norm{R_{>k}}_2$, and in the second the fact that $H$ is $\eps_2$-far from $k$-local.
\end{proof}

\subsection{Structural lemma for sparse Hamiltonians}
Similar to local Hamiltonians, we show a discrepancy in the  sum of the top Pauli coefficients of the short-time evolution operator for close-to-sparse and far-from-sparse Hamiltonians. To formally state this result we need to introduce the concept of \emph{top energy}. Let $U(t)$ the time evolution operator at time $t$ and let $\{\widehat{U}(t)\}_x$ be its Pauli coefficients. We assign labels from $\{x_0,\dots,x_{4^n-1}\}$ to $x\in \01^{2n}$ in a way that $\widehat{U}_{x_0}=\widehat{U}_{0^n}$ and $|\widehat{U}_{x_1}|\geq |\widehat{U}_{x_2}|\geq \cdots \geq |\widehat{U}_{x_{4^n-1}}|$. Now, we define the top energy at time $t$ as 
$$\topenergy(t;s):=|\widehat{U}_{x_0}(t)|^2+\sum_{i\in [s]}|\widehat{U}_{x_i}(t)|^2,$$
\begin{lemma}\label{lem:sparsitydiscrepancy}
    Let $H$ be a $n$-qubit Hamiltonian with $\norm{H}_{\mathrm{op}}\leq 1$ and $\Tr[H]=0$. Let $t\in (0,1)$. On the one hand, if $H$ is $\eps_1$-close to $s$-sparse, then
    $$\topenergy(t;s)\geq 1-\eps_1^2t^2-O(t^3s).$$
    On the other hand, if $H$ is $\eps_2$-far from $s$-sparse, then
    $$\topenergy(t;s)\leq 1-\eps_2^2t^2+O(t^3s).$$
\end{lemma}
\begin{proof}
    For this proof we need to consider the $2$nd order Taylor expansion of $U(t)$,
    \begin{equation*}
    U(t)=\Id-itH-t^2H^2/2+O(t^3)R_2,
    \end{equation*}
    where $R_2$ is the remainder of the series of order 2 that satisfies $\norm{R_2}_{\mathrm{op}}\leq 1,$ because $\norm{H}_{\mathrm{op}}\leq 1$. Since $\Tr[H]=0$ (so $\lambda_{0^n}=0$), we have 
    $$
    \widehat{U}_0(t)=1-\frac{t^2}{2}\cdot \sum_{x\in\01^{2n}}\lambda_x^2+O(t^3),
    $$
    so, using that $|a^2-b^2|= |a-b||a+b|$, we have that
    \begin{equation}\label{eq:U0Expression}
        \Big||\widehat{U}_0(t)|^2-\big(1-t^2\sum_{x\in\01^{2n}}\lambda_x^2\big)\Big|=O(t^3).
    \end{equation}
    To control $|U_x(t)|$ for $x\neq 0^{2n}$, we use the first order Taylor expansion  of $  U(t)=\Id_{2^n}-it H+ct^2 R_1(t)$ and get 
    \begin{align}\label{eq:Uxapprox1}
        \big||\widehat{U}_x(t)|-|t \lambda_x|\big|\leq |\widehat U_x(t)-(-it\lambda_x)|\leq \norm{U(t)-(-itH)}_2\leq O(t^2) \norm{R_1}_2\leq O(t^2),
    \end{align}
    where we again used that $\|R_1\|_2\leq 1$. From this it follows that 
    \begin{align}
        \big||\widehat{U}_x(t)|^2-t^2 \lambda_x^2\big|=\Big|\big(|\widehat{U}_x(t)|-|t \lambda_x|\big)\cdot \big(|\widehat{U}_x(t)|+|t \lambda_x|\big)\Big|\nonumber&= O(t^2)(|U_x|+|t \lambda_x|)\nonumber\\
        &= O(t^2)(2|t \lambda_x|+O(t^2))=O(t^3),
    \end{align}
    where  the second and third equality both used \cref{eq:Uxapprox1}; and in the last line used $| \lambda_x|\leq \norm{H}_{\mathrm{op}}\leq 1$. In particular, the above implies that
    \begin{align}
    \label{eq:Uxapprox2}
         |\widehat{U}_x(t)|^2\geq t^2 |\lambda_x|^2-O(t^3)
    \end{align}

    Now we will define a quantity similar to the top energy, but now we will define the top coefficients as the top coefficients of $H$. To be precise, we assign labels to $\{y_0,\dots,y_{4^n-1}\}$ to the elements of $\{0,1\}^{2n}$ in a way such that $y_0=0^{2n}$ and $|\lambda_{y_1}|\geq \dots\geq |\lambda_{y_{4^n-1}}|$.  We now define  
    \begin{equation*}
        \topenergy_H(t;s):=\Big(1-t^2\sum_{x\in\01^{2n}}\lambda_x^2\Big)+\sum_{i\in [s]}(t\lambda_{y_i})^2.
    \end{equation*}
    If the top $s$ Pauli coefficients of $H$ coincided with the ones of $U(t)$ and there was no error in the Taylor expansion, then $\topenergy_H(t;s)(t)=\topenergy(t;s)$. However, this may not be true in general. Nevertheless, we show that both quantities are close to each other. To this end, 
    \begin{align*}
        \topenergy(t;s)&=|\widehat{U}_{x_0}(t)|^2+\sum_{i\in [s]}|\widehat{U}_{x_i}(t)|^2\\
        &\geq |\widehat{U}_{y_0}(t)|^2+\sum_{i\in [s]}|\widehat{U}_{y_i}(t)|^2\\
        &\geq \big(1-t^2\sum_{x\in\01^{2n}}\lambda_x^2\big)+\sum_{i\in [s]}(t\lambda_{y_i})^2-(s+1)O(t^3)\\
        &=\topenergy_H(t;s)-(s+1)O(t^3),
    \end{align*}
    where in the first inequality we used that $x_1,\dots,x_s$ correspond to the $s$ largest coefficients of $U(t)$, so $\sum_{i\in [s]}|\widehat{U}_{x_i}(t)|^2$ is larger than the sum of the squares of any other $s$ coefficients of $U$; in the second inequality  we used \cref{eq:U0Expression,eq:Uxapprox2}. Similarly, one can check that $\topenergy_H(t;s) \geq \topenergy(t;s)-(s+1)O(t^3),$ so
    $$
    |\topenergy_H(t;s) - \topenergy(t;s)|\leq O(st^3).
    $$
    Now, the claimed result follows by noticing that $$\topenergy_H(t;s)=1-t^2\sum_{i>s}|\lambda_{y_i}|^2,$$ and that $\sum_{i>s}|\lambda_{y_i}|^2$ is the square of the $\ell_2$-distance of $H$ to the space of $s$-sparse Hamiltonians, because $\sum_{i\in [s]}\lambda_{y_i}\sigma_{y_i}$ is the $s$-sparse Hamiltonian closest to $H$.
\end{proof}

\subsection{Subroutines for learning without memory}
\label{sec:tecnicalsubroutinesmemory}
Motivated by the difficulty of accessing quantum memory in the NISQ era \cite{preskill2018quantum}, we propose two subroutines that serve the purpose of substituting Pauli sampling (i.e., sampling from $\{|\widehat{U}_x|^2\}_x$ by creating the Choi state corresponding to $U$ as in Fact~\ref{fact:bellsamplingU}) and the \textsf{SWAP} test in our learning~algorithms, by protocols which do not require memory.
\subsubsection{Estimating Pauli distribution}
Our first lemma constructs an algorithm to estimate the Pauli distribution determined by a $n$-qubit unitary in $\ell_\infty$-error with just $O(n)$ queries to the unitary. These queries are performed on a random stabilizer state and the measurement are also random stabilizer measurements, as in other quantum protocols used for testing and learning \cite{flammia2020efficient,yu2021sample,fawzi2023lower,bluhm2024hamiltonianv2}. Here we propose a novel classical post-processing that allows us to emulate Pauli sampling in our context and others. In particular, we propose an algorithm to tolerantly test if a unitary is a $k$-junta, making progress in a question by Chen, Nadimpalli and Yuen \cite[Section 1.3]{chen2023testing}, and show that can be implemented without quantum memory.

\begin{lemma}\label{lem:memorylessPaulisampling}
    Let $U$ be a $n$-qubit unitary, let $\mathcal{S}\subseteq \01^{2n}$. There is an memory-less algorithm that makes $O\big(\log(|\mathcal{S}|/\delta)/\eps^2\big)$ queries to the unitary on stabilizer states and performing Clifford measurements can provides estimates $|\alpha_x|^2$ such that 
    $$
    \big||\widehat{U}_x|^2-|\alpha_x|^2\big|\leq\eps
    $$
    for every $x\in\mathcal{S}$.
\end{lemma}
To this end, we will look at its Pauli expansion from the $\{0,1\}^{2n}$ point of view. Namely, we will consider the expansion
\begin{equation*}
    U=\sum_{x\in \{0,1\}^{2n}}\widehat{U}_x \sigma_x,
\end{equation*}
where $\sigma_{(a,b)}=i^{a\cdot b}X^aZ^b$. Before presenting our algorithm, we introduce a few facts. Let $N=2^n$. There exists subspaces $G_1,\dots, G_{N+1} \subseteq \{0,1\}^{2n}$ such~that 
\begin{itemize}
    \item $G_i=\{x\in\{0,1\}^{2n}:\ [ x,y]=0,\ \forall y\in G_i\},$ $\forall i\in [N+1]$,
    \item $G_i\cap G_j=\{ \sigma_{0^n}\}$ if $i\neq j,$
    \item $|G_i|=N$. 
\end{itemize}
Let $i\in [N]$ \cite{bandyopadhyay2002new}. Then, $|\{0,1\}^{2n}/G_i|=N$. Let $r_j^i\in \{0,1\}^{2n}$ for $j\in [N]$ be representatives of the different equivalence classes of  $\{0,1\}^{2n}/G_i$. Then, for every $i,j\in [N]$ the following matrix determines a pure state 
\begin{equation}
    \ket{\phi_{i,j}}\bra{\phi_{i,j}}=\frac{1}{N}\sum_{x\in G_i}(-1)^{[r_j^i,x]}\sigma_x.
\end{equation}
An important property is that $\mathcal{B}_i=\{\ket{\phi_{i,j}}\}_{j\in [N]}$ is an orthonormal basis
for every $i\in [N+1]$. Also, $\mathcal B_i$ are mutually unbiased bases and form a $2$-design, meaning that 
\begin{align}\label{eq:2design}
    \frac{1}{N(N+1)}\sum_{i\in [N+1]}\sum_{j\in [N]}\ket{\phi_{i,j}}\bra{\phi_{i,j}}\otimes\ket{\phi_{i,j}}\bra{\phi_{i,j}}=\frac{\mathbb{I}+F}{N(N+1)},
\end{align}
where $F$ is the swap operator, i.e., $F=\sum \ket{xy}\bra{yx}$ \cite[Proposition 2.3]{bluhm2024hamiltonianv2}.  We make an extra remark: given $i\in [N],\ j\in [N+1]$ and $x\in\{0,1\}^{2n}$, there exists a unique $\ell(i,j,x)$ such that
\begin{equation}\label{eq:lijx}
    |\langle \phi_{i,\ell(i,j,x)}|\sigma_x|\phi_{i,j}\rangle|=1.
\end{equation}
Indeed, 
\begin{equation*}
    \sigma_x\ket{\phi_{i,j}}\bra{\phi_{i,j}}\sigma_x=\sum_{y\in G_i}(-1)^{[r_{j}^i,y]}\sigma_x\sigma_y\sigma_x=\sum_{y\in G_i}(-1)^{[r_{j}^i+x,y]}\sigma_y,
\end{equation*}
so $\sigma_x\ket{\phi_{i,j}}\bra{\phi_{i,j}}\sigma_x=\ket{\phi_{i,\ell(i,j,x)}}\bra{\phi_{i,\ell(i,j,x)}}$ for $\ell(i,j,x)$ such that $r_{\ell(i,j,x)}^i\sim_{G_i}r_j^i+x$. 
Now, we are ready to introduce \cref{algo:memory-lessPauliSampling} and prove \cref{lem:memorylessPaulisampling}. 
\begin{algorithm}
\textbf{Input:} Query access to a unitary $U$, error parameter $\eps\in (0,1)$, a set  $\mathcal S\subseteq\{0,1\}^{2n}$, failure parameter $\delta\in (0,1)$
\begin{algorithmic}[1]
    \State Set $T=O(\log(|\mathcal S|/\delta)/\eps^2)$
    \State Initialize $|\alpha_x|^2=|\beta_x|^2=0$ for $x\in \mathcal S$
    \For{$t=1,...,T$}
        \State Sample $i\in [N+1]$ and $j\in [N]$ uniformly at random 
        \State Prepare $U\ket{\phi_{i,j}}$
        \State Measure in the basis $\mathcal B_i$ and let $l$ be the outcome
        \For{$x\in\mathcal S$}
            \If{$l=l(i,j,x)$}
            \State $|\beta_x|^2\leftarrow |\beta_x|^2+1/T$
            \EndIf
        \EndFor
    \EndFor
    \For{$x\in\mathcal S$}
        $|\alpha_x|^2=(N+1)/N|\beta_x|^2-1/N$
    \EndFor
\end{algorithmic}
\textbf{Output}: $(|\alpha_x|^2)_{x\in\mathcal{ S}}$
\caption{Memory-less Pauli sampling}\label{algo:memory-lessPauliSampling}
\end{algorithm}

\begin{proof}[Proof of \cref{lem:memorylessPaulisampling}]
One iteration of \cref{algo:memory-lessPauliSampling} consists of the following: 
\begin{itemize}
    \item Pick uniformly at random $i\in [N+1]$ and $j\in [N]$:
    \item Prepare $U\ket{\phi_{i,j}}$:
    \item Measure the state in the basis according $\mathcal{B}_i$. Suppose we obtain outcome $\ell$.
\end{itemize}
For every $x\in \{0,1\}^{2n}$, we define the following random variable that takes different values depending on the outcome of the random iteration we just have described: 
\begin{equation}
    M_{x}=\left\{\begin{array}{ll}
         0 & \text{if }\ell\neq \ell(i,j,x), \\
         1 & \text{if }\ell=\ell(i,j,x).
    \end{array}\right. 
\end{equation}
Note that given that $M_x$ takes values within a bounded interval, by the Hoeffding bound and a union bound, we can estimate the expectation of $M_{x}$ for every $x\in\mathcal S$ within error $\eps$ with probability $\geq 1-\delta$ from $O(\log(|\mathcal S|/\delta)/\eps^2)$ repetitions of the iteration above, which requires 1 query and no quantum memory. Thus, it only remains to show that these expectations are closely related to $|\widehat{U}_x|^2$.  To this end, observe that
  \begin{align*}
        \mathbb E [M_x]&=\mathbb P[\ell=\ell(i,j,x)]\\
        &= \frac{1}{N(N+1)}\sum_{i,j} |\langle \phi_{i,\ell(i,j,x)}|U|\phi_{i,j}\rangle |^2\\
        &=\frac{1}{N(N+1)}\sum_{i,j} |\langle \phi_{i,j}|\sigma_xU|\phi_{i,j}\rangle |^2\\
        &=\frac{1}{N(N+1)} \Tr\Big((\sigma_x U)^* \otimes \sigma_x U \cdot \sum_{i,j} \ketbra{\phi_{i,j}}{\phi_{i,j}}^{\otimes 2}\Big)\\
         &=\frac{1}{N(N+1)} \Tr\Big((\sigma_x U)^* \otimes \sigma_x U \cdot (\id+F)\Big)\\
        &=\frac{1}{N(N+1)}(\Tr[U^*\sigma_x\sigma_xU]+|\Tr[\sigma_xU]|^2)\\
        &=\frac{1}{N+1}+\frac{N}{N+1}|\widehat{U}_x|^2,
    \end{align*}
    where in the third line we used the definition of $\ell(i,j,x)$ (see \cref{eq:lijx}), in the fourth line we have used \cref{eq:2design}.
 The equation above implies that, $$|\widehat U_x|^2=\frac{N+1}{N}\mathbb E[M_x]-\frac{1}{N},$$ so $((N/N+1)\eps)$-estimates of $M_x$ yield $\eps$-estimates of~$|\widehat{U}_x|^2$. 
\end{proof}

\paragraph{Bonus application: Testing $k$-junta quantum unitaries.}
Now, we use Pauli sampling to show that $O(\exp(k)\poly(1/(\eps_2-\eps_1)))$ queries are sufficient to test whether an $n$-qubit unitary is $\eps_2$-far or $\eps_1$-close to $k$-junta. In addition, using \cref{lem:memorylessPaulisampling} shows that $O(n\exp(k)\poly(1/(\eps_2-\eps_1)))$ queries are sufficient for the same task in the absence of memory. To the best of our knowledge, this is the first tolerant $k$-junta tester for unitaries. Despite $\tilde O(\sqrt k)$ queries being sufficient for non-tolerant testing~\cite{chen2023testing}, the worse $k$-dependence of our tolerant tester should not come as a surprise, as classical tolerant testing of $k$-junta Boolean functions requires $\Omega(2^{\sqrt k})$ samples, but the non-tolerant tester requires only $O(k)$ samples~\cite{blais2009testing,pallavoor2022approximating}. We will make use of the following lemma by Wang \cite{wang2011testing}.

\begin{lemma}[\cite{wang2011testing}]\label{lem:juntalemmawang}
    Let $U$ be an $n$-qubit unitary. If $U$ is $\eps_1$-close to $k$-junta in $2$-norm, then there exists $K\subseteq [n]$ of size $k$ such that $$\sum_{\mathbf{supp}(x)\subseteq K}|\widehat U_x|^2\geq 1-\eps_1^2.$$
    On the other hand, if $U$ is $\eps_2$-far from $k$-local in 2-norm, then for every $K\subseteq [n]$ of size $k$
    $$\sum_{\mathbf{supp}(x)\subseteq K}|\widehat U_x|^2\leq 1-\frac{\eps^2_2}{4}.$$
\end{lemma}

\begin{proposition}\label{prop:testkjuntaunitaries}
    Let $U$ be an $n$-qubit unitary and $0<2\eps_1<\eps_2$. One can test whether $U$ is $\eps_2$-far or $\eps_1$-close to $k$-junta in $2$-norm by making $O(16^k/(\eps_2-\eps_1)^4)$ queries to $U$ with $n$ qubits of memory. In absence of memory, $O(16^kn/(\eps_2-\eps_1)^4)$ queries are sufficient.
\end{proposition}
\begin{proof}
    Assume for the moment that we have a distribution $(\alpha_x)_{x\in\{0,1,2,3\}^{n}}$ such that 
    \begin{equation}\label{eq:approxtestjunta}
        |\alpha_x-|\widehat U_x|^2|\leq \frac{\eps_2^2/4-\eps_1^2}{2\cdot 4^k} 
    \end{equation}
    for every $x\in\{0,1,2,3\}^n$. Then, for every $K\subseteq [n]$ of size $k$ we can approximate $\sum_{\text{supp}(x)\subseteq K}|\widehat U_x|^2$ via $\sum_{\text{supp}(x)\subseteq K}\alpha_x$ up to error $(\eps_2^2/4-\eps_1^2)/2.$ Thanks to \cref{lem:juntalemmawang}, this is enough for testing if $U$ is $\eps_2$-close or $\eps_1$-far from local. 

    If we have access to $n$-qubits of quantum memory, then we can sample $O(16^k/(\eps_2^2-\eps_1^2)^2)$ times from $(|\widehat U_x|^2)_x$ with $O(16^k/(\eps_2^2-\eps_1^2)^2)$ queries via \cref{fact:bellsamplingU}, and the empirical distribution will satisfy \cref{eq:approxtestjunta} thanks to \cref{thm:canoneproof}. If we have no access to quantum memory, then $O(16^kn/(\eps_2^2-\eps_1^2)^2)$ queries are sufficient thanks to \cref{lem:memorylessPaulisampling}.
\end{proof}

\subsubsection{Estimating Pauli coefficients}
Below we give a protocol that allows one to estimate the Pauli coefficients  of the unknown unitary without quantum memory. It is based on the fact that any non-identity Pauli operator can be written as the product of two other anti-commuting Pauli operators, a fact which has been previously used in \cite[Lemma 6.1]{caro2023learning}.

\begin{lemma}\label{lem:memorylesscoefestimation}
    Let $x\in\01^{2n}$ and let $H$ be an $n$-qubit traceless Hamiltonian. There is a memory-less algorithm (Algorithm~\ref{algo:memory_less_pauli_est}) that by making $O(\norm{H}^4/\eps^4)$ queries to $U(\eps/\norm{H}^2)$ on a product state and making Pauli measurements can output an estimate $\tilde \lambda_x$ such that $|\lambda_x-\tilde \lambda_x|\leq\eps$.  
\end{lemma}

\begin{algorithm}
\textbf{Input:} Query access to the time evolution of $U(t)=e^{-itH}$, error parameter $\eps\in (0,1)$,  $x\in\{0,1\}^{2n}$, failure parameter $\delta\in (0,1)$
\begin{algorithmic}[1]
    \State Set $T=O((\norm{H}^4/\eps^4)\log(1/\delta))$
    \State Pick $x',x''\in\{0,1\}^{2n},$ such that $\sigma_{x'}\sigma_{x''}=a \sigma_x$ for some $a\in\{\pm i\}$, and $\sigma_{x'}\sigma_{x''}=-\sigma_{x''}\sigma_{x'}.$
    Set $\alpha=0$
    \For{$j=1,...,T$}
        \State Prepare $\rho=(\Id-\sigma_{x'})/2^n$
        \State Apply $U(\Theta(\eps/\norm{H}^2))$ to $\rho$
        \State Measure in a eigenbasis of $\sigma_{x''}$. Let $\alpha_j$ be the $\pm 1$ eigenvalue of the measured eigenvector of $\sigma_{x''}$
        \State $\alpha\leftarrow\alpha+\alpha_j/T$
    \EndFor
\end{algorithmic}
\textbf{Output}: $\alpha/(2i\eps a)$
\caption{Memory-less Pauli coefficient estimation}\label{algo:memory-lessCoefficientEstimation}
\label{algo:memory_less_pauli_est}
\end{algorithm}

\begin{proof} Let $x',x''\in \{0,1\}^{2n}$ and $a\in\{-i,i\}$ such that $\sigma_{x'}\sigma_{x''}=a \sigma_{x}$ and $\sigma_{x'}\sigma_{x''}=-\sigma_{x''}\sigma_{x'}$. Then,
\begin{align}
    \nonumber&\Tr\Big[\sigma_{x''}\cdot e^{-itH}\cdot \frac{\Id-\sigma_{x'}}{2^n}\cdot e^{itH}\Big]\\
    \nonumber&=\Tr\Big[\sigma_{x''}\cdot \big(\Id_{2^n}-itH+ct^2 R_1(t)\big)\cdot \frac{\Id-\sigma_{x'}}{2^n}\cdot \big(\Id_{2^n}+itH+ct^2 R_1(t)\big)\Big]\\
    \nonumber&=\frac{1}{2^n}\Big(\Tr[\sigma_{x''}(\Id-\sigma_{x'})]-it\Tr[\sigma_{x''}H(\Id-\sigma_{x'})]+it\Tr[\sigma_{x''}(\Id-\sigma_{x'})H]+O(t^2)\Tr[\sigma_{x''}R]\Big)\\
    \nonumber&=\frac{1}{2^n}\Big(-it\Tr[\sigma_{x''}H(\Id-\sigma_{x'})]+it\Tr[\sigma_{x''}(\Id-\sigma_{x'})H]\Big)+O(t^2)\norm{H}^2\\
    \nonumber&=\frac{it}{2^n}(-\Tr[\sigma_{x''}H]+\Tr[\sigma_{x'}\sigma_{x''}H]+\Tr[\sigma_{x''}H]-\Tr[\sigma_{x''}\sigma_{x'}H])+O(t^2)\norm{H}^2\\
    \nonumber&=\frac{2ita}{2^n}\Tr[\sigma_xH]+O(t^2)\norm{H}^2\\
    &=2ita \lambda_x+O(t^2)\norm{H}^2\label{eq:approxnomemo}
\end{align}
where in the first line used the Taylor expansion, third inequality used $\norm{R}_{\mathrm{op}}\leq \norm{H}$, $\Tr[\sigma_{x''}]=\Tr[\sigma_x]=0$ and $\Tr[\sigma_{x''}R]\leq \norm{\sigma_{x''}}_1\norm{R}_{\mathrm{op}}\leq 2^n\norm{H}$, and in the fifth inequality we used that
$$
\sigma_{x'}\sigma_{x''}=-\sigma_{x''}\sigma_{x'}=a \sigma_x.
$$ Thus, by taking $t=\eps/\norm{H}^2$ and dividing by $2i\eps/\norm{H}^2 a$ on both sides of \cref{eq:approxnomemo}, we have that 
\begin{align*}
    \left|\lambda_x-\frac{\norm{H}^2}{{2i\eps a}}\cdot {\Tr\big[\sigma_{x''}e^{-i\eps H}\frac{\Id-\sigma_{x'}}{2^n}e^{i\eps H}\big]}\right|=O(\eps). 
\end{align*}
Hence, if we estimate $\Tr[\sigma_{x''}e^{-i\eps/\norm{H}^2 H}\frac{\Id-\sigma_{x'}}{2^n}e^{i\eps/\norm{H}^2 H}]$ up to error $\eps^2/\norm{H}^2$, which can be done with $O(\norm{H}^4/\eps^4)$ queries to $e^{-i\eps/\norm{H}^2 H}$,
we obtain an $\eps$-estimate of $\lambda_x$. 
\end{proof}

\subsection{Pauli Hashing}\label{sec:pauli_hashing}
In this section, we describe how the $n$-qubit Pauli operators may be \emph{hashed} or isolated into separate sets (which we will call \emph{buckets}) based on their commutation relations with a subgroup of Pauli operators that will be defined shortly. We will then observe that given a function $f:\{0,1\}^{2n} \rightarrow [-1,1]$ taking $2n$-bit strings corresponding to the Weyl operators to values in $[-1,1]$, the hashing process gives us access to projections of these functions onto each of these buckets described by its symplectic Fourier transform. This entire hashing process can thus be viewed as \emph{symplectic} Fourier hashing, analogous  to Fourier hashing used by Gopalan et al.~\cite{gopalan2011sparsity}. 

\subsubsection{Hashing to cosets of a random subgroup} \label{sec:hashing_random_cosets}
We now describe our approach for pairwise independently hashing the Weyl operators (or the symplectic Fourier characters of $f$ defined earlier). Consider a set of $2n$-bit strings $\beta_1,\ldots,\beta_t$ sampled independently and uniformly at random from $\FF_2^{2n}$. We define $H$ to be the subspace $H = \text{span}\{\beta_1,\ldots,\beta_t\}$ spanned by these $t$ vectors. The dimension of $H$ is thus $t$ and $|H|=2^t$.\footnote{We remark that the probability these $t$ vectors are linearly independent is $\geq 1-t2^{t-2n}$ and since the $t$ we eventually pick will be $\ll 2n$, this probability is negligible and from here onwards for simplicity we will assume that these vectors are linearly independent.}
The set of Weyl operators corresponding to elements in $H$ is a subgroup $S$ of the $n$-qubit Pauli group and the Weyl operators corresponding to $\{\beta_j\}_{j \in [t]}$ are the generators of $S$. For each $b \in \FF_2^t$,  define the \emph{bucket}
\begin{equation}
    C(b) := \{ \alpha \in \FF_2^{2n} : [\alpha, \beta_j] = b_j \, \forall j \in [t] \}.
    \label{eq:bucket_b_of_H}
\end{equation}
Let $\Cc = \{C(b)\}_{b \in \FF_2^t}$ be the set of all buckets. Moreover, observe that $\{C(b)\}$ are  the cosets in $\commutantperp(H)=\FF_2^{2n} / \commutant(H)$, because $z\in \commutant (H)$ if and only if $[z,\beta_j]=0$ for all $j\in [t]$. Also, note that $|\commutant(H)|=2^{2n-t}$ and $|\commutantperp(H)|=2^t$. 
We now argue that this indeed corresponds to a random hashing process, similar to~\cite[Proposition~2.9]{gopalan2011sparsity}.
\begin{proposition}\label{prop:symplectic_hashing}
The following facts are true for the random coset structure described so far.
\begin{enumerate}[$(i)$]
    \item For each $\alpha \in \FF_2^{2n} \setminus 0^{2n}$ and each $b$, we have $\Pr_H[\alpha \in C(b)] = 2^{-t}$.
    \item Consider distinct $\alpha, \alpha' \in \FF_2^{2n}$. Then, $\Pr_H[\alpha,\alpha' \in C(b) \text{ for some } b] = 2^{-t}$.
    \item Suppose $S \subset \FF_2^{2n}$ with $|S| \leq s + 1$. Then, $t \geq 2 \log s + \log(3)$ ensures all elements in $S$ fall into different buckets with probability at least $2/3$.
    \item For each $b \in \FF_2^t$ and distinct $\alpha,\alpha' \in \FF_2^{2n}$, we have $\Pr_H[\alpha,\alpha' \in C(b)] = \Pr_H[\alpha \in C(b)] \cdot \Pr_H[\alpha' \in C(b)]$ i.e., the  random variables $[\alpha \in C(b)]$ and $[\alpha' \in C(b)]$ are independent.\footnote{We remark that $0^{2n}$ will always lie in bucket corresponding to $\commutant(0^t)$. We can randomize this by further adding a random permutation after the construction of the cosets as done in~\cite{gopalan2011sparsity}.}
\end{enumerate}
\end{proposition}
\begin{proof}
    $(i)$ is true because as $\beta_i$ is chosen uniformly at random, so for a fixed $\alpha$, $\Pr[[\alpha,\beta_i]=1]=1/2$. $(ii)$ follows by applying $(i)$ to $\alpha-\alpha'$. $(iii)$ follows from $(ii)$ and a union bound over the $\leq s^2$ pairs of elements of $S$. To prove $(iv)$ we divide it in two cases. In the first case, $\alpha=0^n$. As $0^n$ always belongs to $C(0^t)$, the statement reads as $\Pr_{H}[\alpha'\in C(0^t)]=\Pr_{H}[\alpha'\in C(0^t)]$, which is true. If both $\alpha,\alpha'\neq 0^n$, then by $(i)$ the RHS of $(iv)$ is $4^{-t}$.  The LHS is $4^{-t}$ too because 
    \begin{align*}
        2^{-t}\underbrace{=}_{(ii)}\Pr_H[\alpha,\alpha' \in C(b)\text{ for some }b]=\sum_{b\in\{0,1\}^t} \Pr_H[\alpha,\alpha' \in C(b)]
    \end{align*}
    and $\Pr_H[\alpha,\alpha'\in C(b)]$ does not depend on $b.$
\end{proof}

\subsubsection{Projection of functions onto cosets}
We now describe the functions corresponding to the different cosets obtained during hashing. Let $V \subseteq \FF_2^{2n}$ be a subspace. Consider the symplectic complement of $V$, denoted by $\commutant(V) = \{z \in \FF_2^{2n} : \forall v \in V, \, [z,v] = 0\}$. For $a \in \FF_2^{2n}$,  define the coset $a + V := \{a + v : v \in V\}$. Given a function $f:\FF_2^{2n} \rightarrow [-1,1]$, we define the ``projection function" $f|_{a+V}$ as 
\begin{align*}
    f|_{a+V} = \sum_{\beta \in a+V} \widebreve{f}(\beta) \chi_\beta(z),
\end{align*}
where $\chi_\beta(x) = (-1)^{[\beta,x]}$. Notably, the symplectic Fourier coefficients of the projected function are
\begin{align*}
\widebreve{f}|_{a+V}(\alpha) = \begin{cases}
    \widebreve{f}(\alpha) & \text{ if } \alpha \in a + V, \\
    0 & \text{ otherwise.}
\end{cases}
\end{align*}
\begin{fact}
\label{fact:projection_symp}
The projected function $f|_{a+V}$ satisfies $f|_{a+V}(z) = \Exp_{x \in \commutant(V)}\left[f(x+z) \chi_a(x) \right]$.
\end{fact}
\begin{proof}
Expanding the right hand side, we have
\begin{align*}
    \Exp_{x \in \commutant(V)}\left[f(x+z) \chi_a(x) \right]
     &= \frac{1}{|\commutant(V)|} \sum_{x \in \commutant(V)}\left[ \sum_{b \in \FF_2^{2n}} \widebreve{f}(b) (-1)^{[b,x+z]}  \right] (-1)^{[a,x]} \\
     &= \frac{1}{|\commutant(V)|} \sum_{b \in \FF_2^{2n}} \widebreve{f}(b) (-1)^{[b,z]} \left[ \sum_{x \in \commutant(V)}  (-1)^{[a+b,x]}  \right]  \\
     &= \frac{1}{|\commutant(V)|} \sum_{b \in \FF_2^{2n}} \widebreve{f}(b) (-1)^{[b,z]} |\commutant(V)| [a+b \in V]  \\
     &= \sum_{b \in \FF_2^{2n}} \widebreve{f}(b) (-1)^{[b,z]} [b \in a + V]  \\
     &= \sum_{b \in a + V} \widebreve{f}(b) (-1)^{[b,z]} \\
     & = f|_{a+V}(z),
\end{align*}
where we expanded $f(x+z)$ by considering its symplectic Fourier expansion in the first equality, used Fact~\ref{fact:sum_over_symplectic_subspace} in the fourth equality and the last equality follows from the definition of the projected function. This completes the proof.
\end{proof}

We define the weight $\wt_f(a+V) = \sum_{\alpha \in a + V} \widebreve{f}(\alpha)^2$.

\begin{fact}
The weight of the function $f$ on a coset $a+V$ can be evaluated as 
$$
\wt_f(a+V) = \mathop{\Exp}_{\substack{x \in \FF_2^{2n},\\ z \in \commutant(V)}} \left[\chi_a(z) f(x) f(x+z) \right].
$$
This can be estimated up to $\varepsilon$-error with probability $1-\delta$ using $O(1/\varepsilon^2 \log(1/\delta))$ queries to $f$.
\end{fact}
\begin{proof}
Using Parseval's theorem and Fact~\ref{fact:projection_symp}, we have that
\begin{align*}
    \wt_f(a+V) = \Exp_{w \in \FF_2^{2n}} \left[f|_{a+V}(w)^2 \right] &= \Exp_{w \in \FF_2^{2n}} \Exp_{y_1,y_2 \in \commutant(V)} \left[f(y_1+w)f(y_2+w) \chi_a(y_1+y_2) \right] \\
    &=  \Exp_{y_1,y_2 \in \commutant(V)} \Exp_{x \in \FF_2^{2n}}\left[f(x)f(y_1+y_2+x) \chi_a(y_1+y_2) \right] \\
    &= \Exp_{z \in \commutant(V)} \Exp_{x \in \FF_2^{2n}} \left[f(x)f(x+z) \chi_a(z) \right],
\end{align*}
where in the second line we rewrote $x = y_1 + w$, and the third line we rewrote $z = y_1 + y_2$ and used that $\commutant(V)$ is a subspace.
Estimating the weight on $a+V$ requires queries to $f$. The query complexity follows from the Hoeffding bound (\cref{lem:hoeffding}) and that the term inside the expectation lies in $[-1,1]$.
\end{proof}

\section{Testing and learning with quantum memory}
In this section, we  give our testing and learning algorithm for local Hamiltonians.
\subsection{Testing local Hamiltonians}
We now state our locality testing algorithm and prove its guarantees.

\begin{algorithm}
\textbf{Input:} Query access to the time evolution of $U(t)=e^{-itH}$, closeness and farness parameters $\eps_1,\ \eps_2\in (0,1)$, locality parameter $k\in\mathbb N$ and failure parameter $\delta\in (0,1)$
\begin{algorithmic}[1]
    \State Set $T=O(\log(1/\delta)/(\eps_2-\eps_1)^4)$
    \State Let $t=(\eps_2-\eps_1)/(3c)$ and $U=U(t)$
    \State Initialize $\alpha_k'=0$
    \For{$i=1,...,T$}
        \State Perform Pauli sampling from $U$. Let $x\in\{0,1,2,3\}^n$ be the outcome.
        \If{ $|x|>k$}
            \State $\alpha_k'\leftarrow\alpha_k'+1/T$
        \EndIf
    \EndFor
    Set $\alpha_{k}''=0$
    \For{$i=1,...,T$}
        \State Perform Pauli sampling from $U$. Let $x\in\{0,1,2,3\}^n$ be the outcome.
        \State If $|x|>k$, $\alpha_k''\leftarrow\alpha_k''+1/T$
    \EndFor
\end{algorithmic}
\textbf{Output}: If $\alpha_k'\geq (3/4)(\eps_2-\eps_1)^2$ or $\alpha_k''\geq (\eps_2-\eps_1)(\eps_1+2\eps_2)/(9c)-(\eps_2-\eps_1)^2/(18c)$ output that $H$ is far from local, and close to local otherwise
\caption{Locality tester}\label{algo:localitytesting}
\end{algorithm}

\begin{theorem}\label{theo:localitytesting}
    \cref{algo:localitytesting} solves the locality testing problem (\cref{prob:localitytesting} with the property of being $k$-local) with probability $\geq 1-\delta$, by making $O(1/(\eps_2-\eps_1)^4\cdot\log(1/\delta))$ queries to the evolution operator and with $O(1/(\eps_2-\eps_1)^3\cdot\log(1/\delta))$ total evolution time.  
\end{theorem}
\begin{proof}
    Let $t=(\eps_2-\eps_1)/(3c)$ and let $U=U(t)$. For notational simplicity, let $\alpha_k:= \norm{U_{>k}}^2_2$. We will first estimate $\alpha_k$ upto error $(\eps_2-\eps_1)^2/4$. To do that we sample from $\{|\widehat{U}_x|^2\}_x$ using Fact~\ref{fact:bellsamplingU} a total of $T=O(1/(\eps_2-\eps_1)^4\log(1/\delta))$ times, which can be done with $T$ queries. 
    If $x_1,\dots, x_{T}$ are the outcomes of those samples, we define our estimate as 
    $$
    \alpha'_k:=\frac{1}{T}\sum_{i\in [T]}[|x_i|>k].
    $$
    By the Hoeffding bound, we have that indeed $|\alpha'_k-\alpha_k|\leq (\eps_2-\eps_1)^2/4$ with probability $\geq 1-\delta/2$. 
    
    If $\alpha'_k\geq (3/4)(\eps_2-\eps_1)^2,$ then $\alpha_k\geq (\eps_2-\eps_1)^2/2$, so by~\cref{lem:testingdicotomy} we  conclude that $H$ is far from $k$-local. Otherwise, if $\alpha'_k\leq (3/4)(\eps_2-\eps_1)^2,$ then $\alpha_k\leq (\eps_2-\eps_1)^2$. Now we take again $T$ samples from $y_1,\dots,y_T$ from $\{|\widehat{U}_x|^2\}_x$ and define a new estimate 
    $$
    \alpha''_k=\frac{1}{T}\sum_{i\in [T]} [|y_i|>k].
    $$
    By definition $\alpha''_k$ equals $\alpha_k$ in expectation. Furthermore, $\alpha_k$ is the empirical average of random variables whose variance is considerably small, because 
    \begin{align*}
        \mathbb E[[|y|>k]^2]=\mathbb E[[|y|>k]]=\norm{U_{>k}}_2^2\leq (\eps_2-\eps_1)^2. 
    \end{align*}
    Then, an application of Bernstein's inequality (\cref{lem:Bernstein}) shows that $\alpha''_k$ approximates ${\norm{U_{>k}}_2^2}$ up to error $((\eps_2-\eps_1)^2/(18c))^2$ with success probability $1-\delta/2$. At this point, using our structure \cref{lem:testingdicotomy}, this is sufficient for testing $k$-locality.
\end{proof}

We remark that the algorithm for testing locality can be used in more generality for testing if the support of the Hamiltonians is a given $\mathcal{S}\subseteq \{0,1,2,3\}^n$. Also, by a union bound one can test for $M$ supports $\mathcal S_1,\dots,\mathcal{S}_M$ by paying a factor $\log(M).$ 
\begin{theorem}\label{theo:supporttesting}
    Let $H$ be a $n$-qubit Hamiltonian and let $\mathcal S_1,\dots, \mathcal{S}_M\subseteq\{0,1,2,3\}^n$. Then, with $O(1/(\eps_2-\eps_1)^4\log(M/\delta))$ queries and $O(1/(\eps_2-\eps_1)^3\log(M/\delta))$ total evolution time one can simultaneously, for every $i\in[M]$, test if $H$ is $\eps_1$-close or or $\eps_2$-far from being supported on $\mathcal{S}_i$.
\end{theorem}

\noindent \cref{theo:localitytesting} is one case of \cref{theo:supporttesting} where $M=1$ and $\mathcal{S}_1=\{x\in\{0,1,2,3\}^n:|x|\leq k\}.$

\subsection{Testing sparse Hamiltonians}
Now we state our sparsity testing algorithm and prove its guarantees. 
\begin{algorithm}
\textbf{Input:} Query access to the time evolution of $U(t)=e^{-itH}$, closeness and farness parameters $\eps_1,\ \eps_2\in (0,1)$, sparsity parameter $s\in\mathbb N$ and failure parameter $\delta\in (0,1)$
\begin{algorithmic}[1]
    \State Set $T=O(s^6/(\eps_2^2-\eps_1^2)^{6}\cdot\log(1/\delta))$
    \State Let $t=O((\eps_2^2-\eps_1^2)/s)$ and $U=U(t)$
    \State Perform Pauli sampling from $U$ a total of $T$ times. Let $(|\alpha_x|^2)_{x\in\mathcal \{0,1,2,3\}^n}$ the empirical estimate of $(|\widehat U_x|^2)_{x}$ obtained this way. 
    \State Let $|\alpha_{x_1}|^2,\dots,|\alpha_{x_s}|^2$ the $s$-biggest elements of $(|\alpha_x|^2)_{x\in\{0,1,2,3\}^n-\{0^n\}}$
    \State Set $\Gamma=|\alpha_{0^n}|^2+\sum_{i\in [s]}|\alpha_{x_i}|^2.$
\end{algorithmic}
\textbf{Output}: If $\Gamma\geq 1-\eps_1^2\frac{(\eps_2^2-\eps_1^2)^2}{s^2}-\frac{1}{2}\frac{(\eps^2_2-\eps_1^2)^3}{s^2}$ output that $H$ is close to sparse, and far from sparse otherwise
\caption{Fully tolerant sparsity tester}\label{algo:sparsitytesting}
\end{algorithm}

\begin{theorem}\label{theo:sparsitytesting}
    \cref{algo:sparsitytesting} solves the $s$-sparsity testing problem with probability $\geq 1-\delta$, by making $O(s^6/(\eps_2^2-\eps_1^2)^{6}\cdot\log(1/\delta))$ queries to the evolution operator and with $O(s^{5}/(\eps_2^2-\eps_1^2)^{5}\cdot\log(1/\delta))$ total evolution time.  
\end{theorem}
\begin{proof}
Let $t=O((\eps_2^2-\eps_1^2)/s)$. By \cref{lem:sparsitydiscrepancy} we have that if $H$ is $\eps_1$-close to being sparse, then 
    \begin{equation*}
        \topenergy(t;s) \geq 1-\eps_1^2\frac{(\eps_2^2-\eps_1^2)^2}{s^2}-\frac{1}{3}\frac{(\eps^2_2-\eps_1^2)^3}{s^2},
    \end{equation*}
    while if $H$ is $\eps_2$-far from $s$-sparse, then 
    \begin{equation*}
        \topenergy(t;s) \leq 1-\eps_2^2\frac{(\eps_2^2-\eps_1^2)^2}{s^2}+\frac{1}{3}\frac{(\eps^2_2-\eps_1^2)^3}{s^2}.
    \end{equation*}
    From here, it follows that to test it suffices to estimate $\topenergy(t;s)$ up to error 
    \begin{equation*}
        \eps=\frac{1}{2}\left(1-\eps_1^2\frac{(\eps_2^2-\eps_1^2)^2}{s^2}-\frac{1}{3}\frac{(\eps^2_2-\eps_1^2)^3}{s^2}-\left\{1-\eps_2^2\frac{(\eps_2^2-\eps_1^2)^2}{s^2}+\frac{1}{3}\frac{(\eps^2_2-\eps_1^2)^3}{s^2}\right\}\right)= \frac{(\eps_2^2-\eps_1^2)^3}{6s^2}.
    \end{equation*}
    To do that we will obtain an estimate $(\{|\alpha_x|^2\}_x$ of $\{|\widehat{U}_x|^2\}_x$ and use it to approximate $\topenergy(t;s)$. Using Fact~\ref{thm:canoneproof}, we  obtain an empirical distribution $\{|\alpha_x|^2\}_x$ that is obtained after $T=O(s^2\log(1/\delta)/\eps^2)$ samples from $\{|\widehat{U}_x|^2\}_x$ (which can be performed with $T$ queries to $U(t)$ thanks to \cref{fact:bellsamplingU})  satisfies that 
    \begin{equation}\label{eq:goodapproxsparsitytesting}
        \big||\alpha_x|^2-|\widehat{U}_x|^2\big|\leq \frac{\eps}{2s+1}
    \end{equation}
    for all $x\in\01^{2n}$ with probability  $\geq 1-\delta$. We assign new labels $y_0,y_1,\dots, y_{4^n-1}$ to $\01^{2n}$ in a way such that $|\alpha_{y_0}|^2=|\alpha_{0^n}|^2$ and $|\alpha_{y_1}|^2\geq \dots\geq |\alpha_{y_{4^n-1}}|^2$. Now, we define our estimate for $\topenergy(t;s)$ as
    $$
    {\topenergy'(t;s)}=|\alpha_{y_0}(t)|^2+2\sum_{i\in [s]}|\alpha_{y_i}(t)|^2.
    $$
    It only remains to show that ${\topenergy'(t;s)}$ $\eps$-approximates $\topenergy(t;s)$. We will see that in two steps. First, 
    \begin{align*}
        {\topenergy'(t;s)}&=|\alpha_{y_0}(t)|^2+2\sum_{i\in [s]}|\alpha_{y_i}(t)|^2\\
                &\geq |\alpha_{x_0}(t)|^2+2\sum_{i\in [s]}|\alpha_{x_i}(t)|^2\\
                &\geq | u_{x_0}(t)|^2+2\sum_{i\in [s]}| u_{x_i}(t)|^2-\eps\\
                &=\topenergy(t;s)-\eps,
    \end{align*}
    where the second line is true by definition of $y_0,\dots,y_{4^n-1}$ and the third line is true because \cref{eq:goodapproxsparsitytesting}. Switching the roles of ${\topenergy'(t;s)}$ and $\topenergy(t;s)$, one can prove that $\topenergy(t;s) \geq {\topenergy'(t;s)}-\eps$.

    \textbf{Complexity analysis.} We have queried $U(t)$ a total of $T=O(s^2\log(1/\delta)/\eps^2)$ times with $\eps=(\eps_2^2-\eps_1^2)^3/6s^2$ and $t=O((\eps_2^2-\eps_1^2)/s)$, so the number of queries is $$O\left(\frac{s^6}{(\eps_2^2-\eps_1^2)^{6}}\log(1/\delta)\right)$$
    and the total evolution time $$O\left(\frac{s^{5}}{(\eps_2^2-\eps_1^2)^{5}}\log(1/\delta)\right).$$    
\end{proof}
Furthermore, for the regime where $\eps_1=O(\eps_2/s^{0.5})$ we propose a more efficient testing algorithm.

\begin{algorithm}
\textbf{Input:} Query access to the time evolution of $U(t)=e^{-itH}$, sparsity parameter $s\in\mathbb N$, closeness and farness parameters $\eps_1,\ \eps_2\in (0,1)$ satisfying $\eps_1=O(\eps_2/\sqrt{s})$ and failure parameter $\delta\in (0,1)$
\begin{algorithmic}[1]
    \State Set $T=O(s^2/\eps_2^4\cdot\log(1/\delta))$ 
    \State Let $t=\Omega(\eps_2/\sqrt{s})$ and $U=U(t)$
    \State Perform Pauli sampling from $U$ a total of $T$ times. Let $\mathcal X$ the set of sampled Paulis.
\end{algorithmic}
\textbf{Output}: If $|\mathcal X-\{0^{2n}\}|\leq s$ output that $H$ is close to sparse, and far from sparse otherwise
\caption{Not that tolerant sparsity tester}\label{algo:nottolsparsitytesting}
\end{algorithm}
\begin{theorem}\label{theo:sparsitytestnottol}
    Let $H$ be a traceless Hamiltonian with $\norm{H}_{\op}\leq 1$. \cref{algo:nottolsparsitytesting} solves the $s$-sparsity testing problem with probability $\geq 1-\delta$ for $\eps_1=O(\eps_2/s^{0.5})$. The algorithm makes $O(s^2/\eps_2^4\cdot\log(1/\delta))$ queries to the evolution operator and uses $O(s^{1.5}/\eps_2^{3}\cdot\log(1/\delta))$ total evolution~time.  
\end{theorem}

    \begin{proof}
        Let $C>1$ be a constant that appears in the first-order Taylor expansion, $$U(t)=\Id-itH+Ct^2R_1(t)$$ with $\norm{R_1}_{\op}\leq 1$ for $t\in (0,1).$ We will assume that $\delta=1/3,$ as the case $\delta\in (0,1/3)$ follows by a standard majority voting argument. 
        
        \cref{algo:nottolsparsitytesting} is simple. One just performs Pauli sampling of $U=U(t)$ a number of  $T$ times, for some $t$ and $T$ to be determined later. Let $\mathcal{X}$ be the labels of the Pauli strings sampled in this process. If $|\mathcal X-\{0^{2n}\}|\leq s$ we output that $H$ is sparse, and otherwise we output that is far from sparse. It remains to analyze the correctness. 

        \textbf{Correctness.} In the case that $H$ is $\eps_1$-close $s$-sparse, there exists $\mathcal{S}\subset \{0,1\}^{2n}$ of size $s$ where $H$ is $\eps_1$-concentrated. Then, by Taylor expansion, $$\sqrt{\sum_{x\notin (\mathcal{S}\cup \{0^{2n}\})}|\widehat U_x|^2}\leq t\sqrt{\sum_{x\notin (\mathcal{S}\cup \{0^{2n}\})}|\lambda_x|^2}+ Ct^2\leq t\eps_1+ Ct^2\leq 2Ct^2,$$
        where in the last inequality we have assumed that 
        \begin{equation}\label{condition:teps1}
            \eps_1\leq Ct.
        \end{equation}  
        Hence, the probability of sampling an element outside $\mathcal{S}\cup \{0^{2n}\}$ in one sample is at most $4C^2t^4$. Thus, the probability of not sampling an element outside $\mathcal{S}\cup \{0^{2n}\}$ in $T$ samples is at least 
        \begin{align*}
            (1-4C^2t^2)^T\geq 1-4C^2t^4T.
        \end{align*}
        In particular, if 
        \begin{align}\label{eq:conditionTtclose}
            T\leq \frac{1}{3}\frac{1}{4C^2t^4}
        \end{align}
        it will be satisfied that $|\mathcal X-\{0^{2^n}|\leq s$ with probability $\geq 2/3$, as desired. 
        
        In the case that $H$ is $\eps_2$-far from $s$-sparse, we will perform an analysis similar to the coupon collector problem. By Taylor expansion we have that for every set $\mathcal S$ of size $s$, 
        \begin{equation}\label{eq:farfromsparse}
            \sqrt{\sum_{x\notin(\mathcal S-\{0^{2n}\})} |\widehat U_x|^2}\geq \eps_2t-Ct^2\geq \frac{\eps_2t}{2},
        \end{equation}
        where we have assumed that 
        \begin{equation}\label{eq:teps2}
            Ct\leq \eps_2/2.
        \end{equation}
        Let $X_i$ the random variable that accounts for the number of samples between the $(i-1)$-th sampled non-$0^{2n}$-Pauli and the $i$-th sampled non-$0^{2n}$-Pauli. Applying \cref{eq:farfromsparse} to every $\mathcal{X}_i$, it follows that $\mathbb E[X_i]\leq 4/\eps_2^2t^2$ for every $i\in [s+1]$, so
        \begin{equation*}
            \mathbb E[X_1+\dots +X_{s+1}]\leq \frac{4(s+1)}{\eps_2^2t^2}.
        \end{equation*}
        Hence, by Markov's inequality, if 
        \begin{equation}\label{eq:Ttfar}
            T\geq \frac{\sqrt{3}4(s+1)}{\eps_2^2t^2}
        \end{equation}
        it will be satisfied that $|\mathcal X-\{0^{2n}\}|\geq s+1$ with probability $\geq 2/3$, as desired. 

        Finally, we note that we have assumed conditions \cref{condition:teps1,eq:conditionTtclose,eq:teps2,eq:Ttfar} to ensure the correctness of the algorithm. All these equations are satisfied provided that 
        \begin{align*}
            t&= \frac{\eps_2}{\sqrt{50C^2(s+1)}}=\Omega\left(\frac{\eps_2}{\sqrt{s}}\right),\\
            T&=\frac{1}{12C^2t^4}=O\left(\frac{s^2}{\eps_2^4}\right),\\
            \eps_1&\leq \frac{\eps_2}{\sqrt{50(s+1)}}=O\left(\frac{\eps_2}{\sqrt{s}}\right).
        \end{align*}
    \end{proof}

\subsection{Learning unstructured and structured Hamiltonians}
We start by showing how to efficiently learn an arbitrary $n$-qubit Hamiltonian in $\ell_\infty$ error. To do that, we propose a protocol to estimate a given set of Pauli coefficients $\mathcal X$ of a Hamiltonian via Shadow tomography. To describe the protocol, we introduce the following $2n$-qubit observables. Given $x\in\{0,1\}^{2n}$, we define 
\begin{align*}
    \mathcal R_x&:=\frac{1}{2}(\ket{\text{Bell}_{0^{2n}}}\bra{\text{Bell}_{x}}+\ket{\text{Bell}_x}\bra{\text{Bell}_{0^{2n}}}),\\
    \mathcal I_x&:=\frac{1}{2}(-i\ket{\text{Bell}_{0^{2n}}}\bra{\text{Bell}_{x}}+i\ket{\text{Bell}_x}\bra{\text{Bell}_{0^{2n}}}).
\end{align*}

\begin{algorithm}
\textbf{Input:} Query access to the time evolution of $U(t)=e^{-itH}$, target set of Pauli coefficients $\mathcal X\subseteq \{0,1\}^{2n}-\{0^n\}$, error parameter $\eps\in (0,1)$, and failure parameter $\delta\in (0,1)$
\begin{algorithmic}[1]
    \State Set $T= O(\norm{H}^4/\eps^4\cdot \log(|\mathcal X|/\delta))$ and $t_0=\Theta(\eps/\norm{H}^2)$
    \State Set $U=U(t_0)$
    \For{$j\in [T]$}
        \State Prepare $\ket{J(U)}=(U\otimes \Id_{2^n})\ket{\mathrm{Bell}_n}$
        \State Apply a uniformly random Clifford gate $C$
        \State Measure in the computational basis. Let $\ket{b_j}$ be the outcome
        \For{$x\in\mathcal X$}
            \State Let $ {\mathcal R}_{x,j}=(2^n+1)\bra{b_j}C^{-1}\mathcal R_xC\ket{b_j}$ and $ {\mathcal I}_{x,j}=(2^n+1)\bra{b_j}C^{-1}\mathcal I_xC\ket{b_j}$
        \EndFor
    \EndFor
    \For{$x\in\mathcal X$}
            \State Set $\widetilde R_x:=\mathrm{MedianOfMeans}( {\mathcal R}_{x,j})_j$ and $\widetilde I_x:=\mathrm{MedianOfMeans}( {\mathcal I}_{x,j})_j$
        \EndFor
\end{algorithmic}
\textbf{Output}: $( (\widetilde R_x+i\widetilde I_x)/(-it))_{x\in\mathcal X}$
\caption{Estimating a given set of Pauli coefficients of a Hamiltonian}\label{algo:CoeffEstimation}
\end{algorithm}

\begin{lemma}\label{lem:ShadowHamEstimation}
    Let $H$ be an $n$-qubit traceless Hamiltonian and $\mathcal{X}\subseteq \{0,1\}^{2n}$. Then, \cref{algo:CoeffEstimation} allows to estimate the Pauli coefficients corresponding to $\mathcal X$ with success probability $\geq 1-\delta$. It uses $O((\log |\mathcal X|/\delta)\norm{H}^4/\eps^4)$ queries and $O(\log(|\mathcal X|/\delta)\norm{H}^2/\eps^3)$ total evolution time. 
    
    The minimum evolution time is $\eps/\norm{H}^2$, the number of ancillas is $n$, and the time complexity is $O(\poly(n)|\mathcal X|\norm{H}^4/\eps^4\cdot \log(|\mathcal X|/\delta))$.
\end{lemma}
\begin{proof}
    \textbf{Correctness of the algorithm:} Let $t_0=\Theta(\eps/\norm{H}^2)$ and $U=U(t_0).$ As $\Tr[\mathcal R_x^2]=\Tr[\mathcal I_x^2]=2$, by \cref{theo:CliffShadows}, the numbers $\widetilde R_x$ and $\widetilde I_x$ that \cref{algo:CoeffEstimation} outputs satisfy 
    \begin{equation}\label{eq:classShad}
        |\Tr[\mathcal R_x\ket{J(U)}\bra{J(U)}]-\widetilde{ R}_x|\leq\frac{\eps^2}{\norm{H}^2},\  \ |\Tr[\mathcal I_x\ket{J(U)}\bra{J(U)}]-\widetilde{I}_x|\leq \frac{\eps^2}{\norm{H}^2},
    \end{equation}
    for every $x\in\mathcal X$ with probability $\geq 1-\delta$. By Taylor expansion, as $\lambda_{0^{2n}}=0$, we have that 
$|\widehat U_{0^{2n}}-1|\leq O(t_0^2\norm{H}^2).$  Thus,
\begin{equation}\label{eq:TaylorRxIx}
    \Tr[\mathcal R_x\ket{J(U)}\bra{J(U)}]=\frac{1}{2}(\widehat U_x\widehat U_{0^{2n}}^*+\widehat U_{0^{2n}}\widehat U_x^*)=\text{Re}(\widehat U_x\widehat U_0^*)=\text{Re}(\widehat U_x)\pm O(t_0^2\norm{H}^2),
\end{equation}
 and similarly $\Tr[\mathcal I_x\ket{J(U)}\bra{J(U)}]=\text{Im}(\widehat U_x)\pm O(t_0^2\norm{H}^2).$ Hence, combining \cref{eq:classShad,eq:TaylorRxIx} we have that 
 \begin{equation*}
     |\widehat U_x-(\widetilde{ R}_x+i\widetilde{ I}_x)|\leq \frac{\eps^2}{\norm{H}^2}+O(t_0^2\norm{H}^2)\leq O\left(\frac{\eps^2}{\norm{H}^2}\right),
 \end{equation*}
 for every $x\in\mathcal X$. Finally, by Taylor expansion we have that $|\widehat U_x/(-it_0)-\lambda_x|\leq O(t_0\norm{H}^2),$ so
 \begin{equation*}
     \left|\lambda_x-\frac{\widetilde{ R}_x+i\widetilde{ I}_x}{-it_0}\right|\leq O\left(\frac{\eps^2}{t_0\norm{H}^2}\right)+O(t_0\norm{H}^2)=O(\eps),
 \end{equation*}
for every $x\in \mathcal X$, as claimed.

\textbf{Time complexity:} The time complexity is dominated by the first loop in \cref{algo:CoeffEstimation}, whose time complexity is $O(|\mathcal X|\cdot T\cdot (t_{est}+\poly(n))$, where the $\poly(n)$ comes from applying a random Clifford gate and $t_{est}$ is the time taken to compute $\bra{b}C^{-1}\mathcal R_xC\ket{b}$ for an $n$-qubit Clifford gate $C$ and a computational basis state $\ket{b}.$ Now, expanding $R_x$ one can write $\bra{b}C^{-1}\mathcal R_xC\ket{b}$ as an algebraic expression of a finite number of terms of the kind $\bra{y}D\ket{z}$, where $\ket{y}$ and $\ket{z}$ are computational basis states and $D$ a Clifford gate. Hence, via Gottesman-Knill theorem \cite{gottesman1998heisenberg,aaronson2004improved} follows that $t_{est}=O(n^2)$, so the total time complexity is $O(\poly(n)|\mathcal X|\norm{H}^4/\eps^4\cdot \log(|\mathcal X|/\delta))$.
\end{proof} 

Now, we are ready to present our learning algorithm for arbitrary Hamiltonians with no promise about its structure.
\begin{algorithm}
\textbf{Input:} Query access to the time evolution of $U(t)=e^{-itH}$, error parameter $\eps\in (0,1)$, and failure parameter $\delta\in (0,1)$
\begin{algorithmic}[1]
    \State Set $T= O(\norm{H}^4/\eps^4\cdot \log(\norm{H}^2/\eps^2\delta))$ and $t_0=\Theta(\eps/\norm{H}^2)$
    \State Set $U=U(t_0)$
    \State Set $\mathcal X=\emptyset$
    \For{$j\in [T]$}
        \State Prepare $\ket{J(U)}=(U\otimes \Id_{2^n})\ket{\mathrm{Bell}_n}$
        \State Measure in the Bell basis and add the outcome $x\in\{0,1\}^{2n}$ to $\mathcal X$ if $x\neq 0^{2^n}$
    \EndFor
    \State     Run \cref{algo:CoeffEstimation} run with $U(t)$, $\mathcal X$, $\eps$ and $\delta$ as inputs. Let $(\widetilde \lambda_x)_{x\in\mathcal X}$ the output.
\end{algorithmic}
\textbf{Output}: $\widetilde H=\sum_{x\in\mathcal X}\widetilde \lambda_x\sigma_x$
\caption{Learning unstructured Hamiltonians}\label{algo:unstructuredlearn}
\end{algorithm}
    \begin{theorem}[Learning unstructured Hamiltonians]\label{theo:unstructuredlearning}
        Let $H$ be an $n$-qubit and traceless Hamiltonian. Then, \cref{algo:unstructuredlearn} $\eps$-learns $H$ in the $\ell_\infty$ norm with success probability $\geq 1-\delta.$  It uses $\widetilde O((\norm{H}/\eps)^4)$ queries to the evolution operator and $\widetilde O(\norm{H}^2/\eps^3)$ total evolution time. The minimum evolution time is $\Theta(\eps/\norm{H}^2)$, the algorithm uses $n$ ancilla qubits and only one round of adaptivity, and the time complexity is $\poly(n,1/\eps,\norm{H})$.
    \end{theorem}
    \begin{proof} 
     Let $t_0=\Theta(\eps/\norm{H}^2)$, $U=U(t_0)$ and $T= O(\norm{H}^4/\eps^4\cdot \log(\norm{H}^2/\eps^2\delta))$, as in \cref{algo:unstructuredlearn}.
     
     \textbf{Correctness}: We claim that with probability $\geq 1-\delta$ the set $\mathcal X$ generated in \cref{algo:unstructuredlearn} contains all $x$ such that 
    \begin{equation}\label{eq:XcontainsallBig}
        |\lambda_x|\geq \eps,
    \end{equation}
    and that 
    \begin{equation}\label{eq:Xnottoolarge}
        |\mathcal X|\leq \widetilde O\left(\frac{\norm{H}^4}{\eps^4}\right).
    \end{equation}
    To show \cref{eq:XcontainsallBig} we note that by Taylor expansion, if $|\lambda_x|\geq \eps$, then $|\widehat U_x|=\Omega((\eps^2/\norm{H}^2))$, so $|\widehat U_x|^2=\Omega((\eps^4/\norm{H}^4))$. Hence, the probability that such an $x$ does not belong to $\mathcal X$, which stores the non-$0^{2n}$ outcomes of sampling from $(|\widehat U_x|^2)_x$, is at most $$\left(1-|\widehat U_x|^2\right)^T\leq e^{-T|\widehat U_x|^2}\leq \frac{\eps^2\delta}{\norm{H}^2}.$$ 
    Hence, as there is at most $\norm{H}^2/\eps^2$ coefficients with $|\lambda_x|\geq \eps$, because $\sum_x|\lambda_x|^2\leq \norm{H}^2$, \cref{eq:XcontainsallBig} follows from a union bound. \cref{eq:Xnottoolarge} holds because $|\mathcal X|\leq T.$
    
    Now, if \cref{eq:XcontainsallBig,eq:Xnottoolarge} are satisfied, \cref{algo:CoeffEstimation} provides estimates of the coefficients of $\mathcal X,$ which contains all labels $x$ of coefficients $|\lambda_x|\geq \eps.$

    \textbf{Complexities:} The query complexity is $2T=\tilde O(\norm{H}^4/\eps^4)$, the minimum evolution time $t_0=\Theta(\eps/\norm{H}^2)$ and the total time evolution $2Tt_0=\widetilde O(\norm{H}^2/\eps^3)$. Additionally, the time complexity of \cref{algo:unstructuredlearn} is dominated by the call to \cref{algo:CoeffEstimation}, which runs in time $O(\poly(n)|\mathcal X|\norm{H}^2/\eps^2)$, which thanks to \cref{eq:Xnottoolarge} is $\poly(n,1/\eps,\norm{H})$.
    \end{proof}

\subsubsection{Learning local Hamiltonians}
We now introduce our local Hamiltonian learner and prove its guarantees. 
\begin{algorithm}
\textbf{Input:} Query access to the time evolution of $U(t)=e^{-itH}$, error parameter $\eps\in (0,1)$, locality parameter $k\in\mathbb N$ and failure parameter $\delta\in (0,1)$
\begin{algorithmic}[1]
    \State Set $T=\exp(O(k^2+k\log (1/\eps))\log(1/\delta)$
    \State Let $t=\eps^{k+1}\exp(-k(k+1)/2)$ and $U=U(t)$
    \State Set $\gamma=(\eps/\norm{H}^2)^{k+1}\exp(-k(k+1)/2)$ and $\beta=\gamma\eps/\norm{H}$
    \State Learn $\beta$-estimates $ \lambda_x'$ of $\lambda_x$ via \cref{algo:unstructuredlearn}
    \For{$|x|\leq k$}
        \If{$|\lambda_x'|\leq \gamma$}
        \State $\widetilde \lambda_x=0$
        \Else
        \State $\widetilde \lambda_x=\lambda_x'$
        \EndIf
    \EndFor
\end{algorithmic}
\textbf{Output}: $\sum_{x\leq k}\widetilde \lambda_x\sigma_x$
\caption{Local Hamiltonian learner}\label{algo:localitylearning}
\end{algorithm}

\begin{theorem}\label{theo:locallearning}
     Given a $n$-qubit $k$-local Hamiltonian $H$, \cref{algo:localitylearning} outputs $\widetilde H$ such that with probability $\geq 1-\delta$ satisfies $\norm{H-\widetilde H}_{\ell_2}\leq \eps$. The algorithm makes $\exp(O(k^2+k\log (\norm{H}^2/\eps))\log(1/\delta)$ queries to the evolution operator with $\exp(O(k^2+k\log (\norm{H}^2/\eps))\log(1/\delta)$  total evolution time. 
\end{theorem}
 To prove this theorem, we use the non-commutative Bohnenblust-Hille inequality~\cite{volberg2023noncommutative}.
\begin{theorem}[Non-Commutative Bohnenblust-Hille inequality]\label{theo:NCBH}
    Let $H=\sum_x \lambda_x\sigma_x$ be a $k$-local Hamiltonian. Then, there is a universal constant $C$ such that $$\widetilde H=\sum_{x\in\01^{2n}}|\lambda_x|^{\frac{2k}{k+1}}\leq C^{k}\norm{H}.$$
\end{theorem}

\begin{proof}[Proof of \cref{theo:locallearning}]
    We only analyze the correctness of \cref{algo:localitylearning}, as the complexity quickly follows from \cref{theo:unstructuredlearning}. In this proof we also use the notation of \cref{algo:localitylearning}. The $\ell_2$-error of approximating $H$ with $\widetilde H$ is 
    \begin{equation}
        \norm{\widetilde H-H}_{\ell_2}^2=\sum_{|\lambda_x'|\leq \gamma} |\lambda_x|^2+\sum_{|\lambda_x'|\geq \gamma, |x|\leq k} |\lambda_x-\lambda_x'|^2.
    \end{equation}
    We show separately that the two terms are at most $O(\eps^2)$. To bound the contribution of the small Pauli coefficients, we first note that by \cref{theo:unstructuredlearning} we have that 
    \begin{equation}
        |\lambda_x'|\leq \gamma\quad \implies\quad |\lambda_x|\leq \gamma+\beta=O(\gamma).\label{eq:lam'smallimplieslamsmall}
    \end{equation} Hence, 
    \begin{align}
        \sum_{|\lambda_x'|\leq \gamma} |\lambda_x|^2\leq \sum_{|\lambda_x|\leq O(\gamma)} |\lambda_x|^2\leq O(\gamma^{\frac{2}{k+1}})\sum_{x\in\{0,1\}^{2n}} |\lambda_x|^{\frac{2k}{k+1}}\leq \gamma^{\frac{2}{k+1}}(C^{k}\norm{H}^2)^{\frac{2k}{k+1}}=O(\eps),
    \end{align}
    where in the first inequality we used \cref{eq:lam'smallimplieslamsmall}, in the third inequality we used \cref{theo:NCBH} and in the last inequality that $\gamma=(\eps/\norm{H}^2)^{k+1}\exp(-k(k+1)/2).$
    To bound the contribution of the coefficients $|\lambda_x|\geq \gamma$ we notice that there is at most $\norm{H}^2/\gamma^2$ of them, because $\sum_{x}|\lambda_x|^2\leq \norm{H}^2$.~Thus,  
    \begin{equation*}
        \sum_{|\lambda_x'|\geq \gamma, |x|\leq k} |\lambda_x-\lambda_x'|^2\leq \frac{\norm{H}^2}{\gamma^2} \sup_x |\lambda_x-\lambda_x'|^2\leq \frac{\norm{H}^2\beta^2}{\gamma^2}=\eps^2,
    \end{equation*}
    where in the second inequality we use the $\lambda_x'$ are $\beta$-estimates of $\lambda_x$ and in the last equality we use that $\beta=\gamma\eps/\norm{H}$.
\end{proof}

\subsubsection{Learning sparse Hamiltonians}
In this section we introduce our sparse Hamiltonian learner and prove its guarantees.

\begin{algorithm}
\textbf{Input:} Query access to the time evolution of $U(t)=e^{-itH}$, error parameter $\eps\in (0,1)$, sparsity parameter $s\in\mathbb N$ and failure parameter $\delta\in (0,1)$
\begin{algorithmic}[1]
     \State Learn $(\eps/2)$-estimates $ \lambda_x'$ of $\lambda_x$ via \cref{algo:unstructuredlearn}
    \For{$x\in \{x:\lambda_x\neq 0\}$}
        \If{$\lambda_x'\leq \eps/2$}
            \State $\widetilde \lambda_x=0$
        \Else{$\lambda_x> \eps/2$}
            \State $\widetilde\lambda_x=\lambda_x'$
        \EndIf
    \EndFor
\end{algorithmic}
\textbf{Output}: $\widetilde H=\sum_{x}\widetilde \lambda_x\sigma_x$
\caption{Sparse Hamiltonian learner}\label{algo:sparsitylearning}
\end{algorithm}

\begin{theorem}[Sparse Hamiltonian learning]\label{theo:sparselearning}
    Given an $n$-qubit, $s$-sparse Hamiltonian $H$, \cref{algo:sparsitylearning} outputs another Hamiltonian $\widetilde H=\sum\widetilde \lambda_x\sigma_x$ such that with probability $\geq 1-\delta$ satisfies $\norm{H-\widetilde H}_{\ell_\infty}\leq \eps$, The algorithms uses $\widetilde O(\norm{H}^4/\eps^4\cdot \log(1/\delta))$ queries and $\widetilde O(\norm{H}^2/\eps^3\cdot \log(1/\delta))$ total evolution time.
    
    Furthermore, if $\lambda_x=0$, then $\widetilde{\lambda}_x=0$. This implies that running \cref{algo:sparsitylearning} with $\eps=\eps'/\sqrt{s}$ outputs $\widetilde H$ such that $\norm{H-\widetilde H}_{\ell_2}\leq \eps'$. In this case, the algorithm uses $\widetilde O(\norm{H}^4s^2/\eps'^4\cdot \log(1/\delta))$ queries and $\widetilde O(\norm{H}^2s^{1.5}/\eps'^3\cdot \log(1/\delta))$ total evolution time.
\end{theorem}

\begin{proof}
    The first part, concerning learning in the $\ell_\infty$ error follows from \cref{theo:unstructuredlearning}. The fact that $\lambda_x=0$, then $\widetilde{\lambda}_x=0$ follows from Line 3 of \cref{algo:sparsitylearning}. Finally, we note that having $\lambda_x=0\implies \widetilde{\lambda}_x=0$ and $|\lambda_x-\lambda_x|\leq \eps'/\sqrt{s},$ implies $\norm{H-\widetilde H}_{\ell_2}\leq \eps'.$ Indeed, 
    \begin{align*}
        \norm{H-\widetilde H}_{\ell_2}=\sum_{\lambda_x\neq 0}|\lambda_x-\widetilde \lambda_x|^2\leq s \sup_x |\lambda_x-\widetilde \lambda_x|^2=\eps'^2,
    \end{align*}
    where in the first step we have used that $\lambda_x=0\implies \widetilde{\lambda}_x=0$, in the second that $|\lambda_x-\lambda_x|\leq \eps'/\sqrt{s}$ and in the third that $H$ is $s$-sparse.
\end{proof}

We remark here that, our testing complexity is larger than our learning complexity, which might make it seem as if our testing result is trivial. But crucially, our testing algorithm here is in the much harder framework of \emph{tolerant} testing: in this harder framework, it is unclear if proper learning algorithms imply tolerant testing algorithms.\footnote{The seminal result of Goldreich et al.~\cite{goldreich1998property} shows that \emph{proper} learning algorithms implied standard property testing algorithms.}

\section{Testing and learning without quantum memory}
\subsection{Learning without memory}
In this section, we use the subroutines that we established in Section~\ref{sec:tecnicalsubroutinesmemory} to remove the need for quantum memory in our learning algorithms.
\begin{theorem}\label{theo:learnnomemory}
    Let $H$ be a $n$-qubit traceless Hamiltonian. There are memory-less algorithms with probability $\geq 1-\delta$ such that 
    \begin{itemize}
        \item  they learn $H$ up to error $\eps$ in $\ell_\infty$ error using $\widetilde{O}(\norm{H}^8n/\eps^8)$ queries and $\widetilde{O}(\norm{H}^7n/\eps^6 )$ total evolution time,
        \item if $H$ is $k$-local, they learn $H$ in $\eps$ error in $\ell_2$ norm using just $\exp(O(k^2+k\log (\norm{H}^2/\eps)))\log(n/\delta)$ queries and total evolution time. 
        \item if $H$ is $s$-sparse, they learn $H$ in $\eps$ error in $\ell_2$ norm using just $\widetilde{O}(\norm{H}^8s^4n/\eps^8)$ queries and $\widetilde{O}(\norm{H}^7s^3n/\eps^6 )$ total evolution time, 
    \end{itemize}
\end{theorem}
\begin{proof}
    The proof of the first bullet point results follow by mimicking the proofs of \cref{algo:unstructuredlearn}, but substituting Pauli sampling by \cref{lem:memorylesscoefestimation} and \cref{lem:ShadowHamEstimation} by \cref{lem:memorylesscoefestimation}. The second and third bullet points follow from the first as \cref{theo:locallearning,theo:sparselearning} followed from \cref{theo:unstructuredlearning}. The memory-less subroutines, \cref{lem:memorylesscoefestimation,lem:memorylessPaulisampling}, incur in quadratically more queries and an extra factor due to the union bound over the set of potentially non-zero Pauli coefficients of $H$ that appears in the proof of \cref{lem:memorylessPaulisampling}. Note that in the case where the Hamiltonian is promised to be local, the Pauli support is of size at most $O(n)^k$, so in this case we pay just a factor $k\log(n)$ when applying \cref{lem:memorylessPaulisampling}. In case that there is no promise or we are only promised sparsity, any of the $4^n$ Pauli can be non-zero, so we have to pay a factor $n.$
\end{proof}

\subsection{Testing sparse Pauli channels via Pauli hashing}
We will test the sparsity of a Hamiltonian by reducing the task to testing sparsity of a relevant Pauli channel. In this section, we thus describe how to test sparse Pauli channels via Pauli hashing, which was introduced in Section~\ref{sec:pauli_hashing}. Recall that a Pauli channel can be written as $$\calE(\rho) = \sum_{x \in \{0,1\}^{2n}} p(x) \sigma_x \rho \sigma_x,
$$
and its Pauli fidelities $\lambda(y) = \frac{1}{2^{n}} \Tr\left(\sigma_y  \cdot \calE(\sigma_y)\right)$ satisfy that $\widebreve{\lambda}(x)=p(x).$ We can thus test the sparsity of $\calE$, that is it having $s$ many non-zero error rates by testing the sparsity of the symplectic Fourier spectrum of the Pauli fidelity~$\lambda: \FF_2^{2n} \rightarrow [-1,1]$.

\subsubsection{Testing algorithm}
\paragraph{Random coset structure.} 
As in Section~\ref{sec:hashing_random_cosets}, we define the pairwise independent hashing process corresponding to a random subgroup $H$ of dimension $t$, generated by uniformly sampling $\{h_j\}_{j \in [t]}$
from $\FF_2^{2n}$ randomly. Abusing notation, we will also call the subgroup of Weyl operators corresponding to the strings in $H$ as $H$ as well and similarly for the generators. We then define the buckets corresponding to $b \in \FF_2^t$ accordingly as
\begin{equation} \label{eq:buckets_PC}    
C(b) := \{ \alpha \in \FF_2^{2n} : [\alpha, h_j] = b_j \, \forall j \in [t] \}.
\end{equation}
Let us denote the energy of bucket $C(b)$ for each $b\in \FF_2^t$ as the sum of the error rates in the corresponding coset $a + \commutant(H)$, expressed as
\begin{equation} 
E(b) = \sum_{x \in a + \commutant(H)} p(x).
\end{equation}

\paragraph{Energy estimation.}
Let the bucket $C(b)$ be the coset $a + \commutant(H)$. The energy of the bucket $E(b)$~is 
\begin{align}
E(b) = \sum_{x \in a + \commutant(H)} \widebreve{\lambda}(x) &= \sum_{x \in a + \commutant(H)} \Exp_{z \in \FF_2^{2n}} \left[ (-1)^{[z,x]} \lambda(z) \right] \nonumber \\
&= \frac{1}{4^n} \sum_{z \in \FF_2^{2n}} \lambda(z) \sum_{x \in a + \commutant(H)}  (-1)^{[z,x]} \nonumber \\
&= \frac{1}{4^n} \sum_{z \in \FF_2^{2n}} \lambda(z) (-1)^{[z,a]} \sum_{h \in  \commutant(H)}  (-1)^{[z,h]} \nonumber \\
&= \frac{|\commutant(H)|}{4^n} \sum_{z \in \FF_2^{2n}} \lambda(z) (-1)^{[z,a]} [z \in \commutant(\commutant(H))] \nonumber \\
&= \frac{1}{|\commutantperp(H)|} \sum_{z \in H} \lambda(z) (-1)^{[z,a]},
\label{eq:energy_bucket_PC}
\end{align}
where we used Fact~\ref{fact:sum_over_symplectic_subspace} in the second to last equality and $\commutant(\commutant(H))=H$ in the last equality. We now discuss how the energy of each bucket is computed efficiently using Pauli eigenstates and Pauli measurements.

\begin{claim}
Let $H$ be a subspace of dimension $t$. Then, the energy $E(b)$ of each bucket $b \in \FF_2^t$ can be estimated upto error $\varepsilon$ with probability $1-\delta$ using $O(2^t/\varepsilon^2 \cdot \log(2^t/\delta))$ queries. This only requires preparation of Pauli eigenstates and Pauli measurements.
\label{claim:energy_est_PC}
\end{claim}

\begin{proof}
By definition, $\lambda(z) = \frac{1}{2^n} \Tr(\sigma_z \calE(\sigma_z))$ and thereby can be estimated by repeated queries involving preparation of the eigenstate of Pauli $\sigma_z$, application of the Pauli channel $\calE$, and then measuring with respect to the Pauli $\sigma_z$. For a given $z$, we can thus obtain an estimate $\widetilde{\lambda}(z)$ such that $|\widetilde{\lambda}(z) - \lambda(z)| \leq \varepsilon$ with probability $1-\delta$, using $O(1/\varepsilon^2 \log(1/\delta))$ queries as described.

For a given bucket $C(b)$, we can obtain an estimate $\widetilde{E}(b)$ of its energy $E(b)$ (Eq.~\eqref{eq:energy_bucket_PC})
by taking the empirical mean of the each of the Pauli fidelity estimates $\{\widetilde{\lambda}(z)\}_{z \in H}$
$$
\widetilde{E}(b) = \frac{1}{|\commutantperp(H)|} \sum_{z \in H} \widetilde{\lambda}(z) (-1)^{[z,a]},
$$
where $a \in \FF_2^{2n}$ is such that the bucket $C(b)$ is equivalent to the coset $a+\commutant(H)$.
Note that with an overall query complexity $O(2^t/\varepsilon^2 \log(2^t/\delta))$, we can ensure that $|\widetilde{\lambda}(z) - \lambda(z)| \leq \varepsilon$ for all $z \in H$ and thereby have
$$
|\widetilde{E}(b) - E(b)| \leq \frac{1}{|\commutantperp(H)|} \sum_{z \in H} | \widetilde{\lambda}(z) (-1)^{[z,a]} - \lambda(z) (-1)^{[z,a]}| \leq \frac{1}{|H|} \sum_{z \in H} \varepsilon = \varepsilon, 
$$
where we used that $|\commutantperp(H)|=|H|$.
\end{proof}

We are now in a position to state Algorithm~\ref{algo:testing_sparsity_PC} which tests $s$-sparsity of Pauli channels.
\begin{algorithm}[H]
\caption{Testing $s$-sparsity of Pauli channels} \label{algo:testing_sparsity_PC}
\textbf{Input}: A Pauli channel $\calE$, sparsity parameter $s$, error parameters $0 \leq 2\varepsilon_1 < \varepsilon_2 < 1$
\begin{algorithmic}[1]
    \State Randomly sample a subgroup $H$ of dimension $t=O(\log s)$
    \State Set bucket indices $\{b_1,\ldots,b_T\}=\FF_2^t$
    \State Set $\varepsilon = (\varepsilon_2 - 2\varepsilon_1)/3$ and query budget $T= O(s/\varepsilon^2 \log(s/\delta))$
    \State Obtain energy estimates $\{\widetilde{E}(b)\}_{b \in \FF_2^t}$ using Claim~\ref{claim:energy_est_PC} upon inputs of $H$ and budget $T$
    \State Compute $\Gamma \leftarrow \max \limits_{S \subseteq \FF_2^t: |S| = s} \sum_{j \in S} \widetilde{E}(j)$
    \State If $\Gamma \geq 1 - 8/3\varepsilon_1 + 1/3\varepsilon_2$, accept.
    \State If $\Gamma \leq 1 - 4/3\varepsilon_2 + 2/3\varepsilon_1$, reject.
\end{algorithmic}
\textbf{Output}: \textsf{FLAG} for accepting/rejecting $\calE$
\end{algorithm}
In the next section, we will argue that the query complexity of Algorithm~\ref{algo:testing_sparsity_PC} is indeed $\poly(s)$.

\subsubsection{Analysis}\label{sec:analysis_sparse_testing_PC}
In this section, we prove the following theorem. Our analysis  follows the approach of~\cite{yaroslavtsev2020fast}.
\begin{theorem}\label{thm:testing_sparse_PC}
Let $\calE$ be an $n$-qubit Pauli channel. We can test if $\calE$ is $\varepsilon_1$-close or $\varepsilon_2$-far from being $s$-sparse in diamond norm with probability $\geq 0.9$ using 
$$
O\left(\frac{s^2}{(\varepsilon_2-\varepsilon_1)^6} \left(\log\frac{1}{(\varepsilon_2 - \varepsilon_1)} + \log s \right)\right)
$$ 
queries to $\calE$. These queries involve no quantum memory and only Pauli eigenstate preparations or Pauli measurements.
\end{theorem}

Our goal will be to show that $\Gamma$ in \cref{algo:testing_sparsity_PC} is an $((\eps_2-2\eps_1)/3)$-approximation of $\mathsf{Energy}(\mathcal E;s)$, which is enough to determine if $\mathcal E$ is close or far from being sparse thanks to \cref{lemma:sparse_paulis_struct}. Note that $\Gamma$ has two sources of error. The first appears because we approximate $\mathsf{Energy}(\mathcal E;s)$ via $\max_{|S| = s} \sum_{j \in S} {E}(j),$ which we will refer to it as the \emph{hashing error}. The hashing error will be small because thanks to \cref{prop:symplectic_hashing} all the top $p(x)$ will be in different buckets. The second source of error comes from approximating $E(b)$ via $\widetilde E(b)$. We will refer to it as \emph{estimation error} and will be controlled thanks to \cref{claim:energy_est_PC}. 

\paragraph{Hashing error.}
Let the energies of the buckets be indexed in non-increasing order by energy as $E_1 \geq E_2 \geq \cdots \geq E_{2^t}$ and label the $2^t$ buckets accordingly, $B(1),\dots,B(2^t)$. Let  ${E}_j' = \max_{x \in B(j)} p(x)$. 
We also assign labels $y_1,\dots,y_{4^n}$ to $\{0,1\}^{2n}$ in a way such that  the error rates satisfy as $p_{y_1} \geq p_{y_2} \geq \cdots \geq p_{y_{4^n}}$. With this notation, the hashing error is given by
\begin{equation*}
\mathsf{err}(\calE;H,s) = \sum_{j=1}^s \left( E_j - p_{y_j} \right).    
\end{equation*}
We will prove the following upper bound to the hashing error. 

\begin{lemma} \label{lemma:hashing_error_PC}
Let $\varepsilon \in (0,1/2]$. Suppose $H$ is a random subgroup of dimension $t \geq \log (2s/\varepsilon^2)$, then
$$
\Pr_{H}\left[\mathsf{err}(\calE;H,s) \leq 6\varepsilon \right] \geq 0.96 \,.
$$
\end{lemma}
To show \cref{lemma:hashing_error_PC}, we first note that since $p_{y_1},\dots,p_{y_s}$ are the  $s$ largest error rates, it follows that 
\begin{equation}\label{eq:ub_hashing_error_PC}
    \mathsf{err}(\calE;H,s)= \sum_{j=1}^s \left( E_j - {E}_j' \right) + \sum_{j=1}^s \left( {E}_j' - p_{y_j} \right) \leq \sum_{j=1}^s \left( E_j - {E}_j' \right).
\end{equation}
Hence, to upper bound the hashing error we just have to upper bound the \emph{collision error}  $\sum_{j=1}^s \left( E_j - {E}_j' \right),$ which occurs when two or more coefficients collide in the same bucket. With that purpose, we prove the following claim regarding $\mathsf{err}_j:=E_j-E_{j}'.$

\begin{claim} \label{claim:exp_var_collision_errors_PC}
Suppose $H$ is a subgroup of dimension $t$ drawn uniformly at random. Then,
$$
\Exp_{H} \left[ \sum_{j=1}^s \mathsf{err}_j \right] \leq \sqrt{ \frac{2s}{2^t}}, \quad \mathsf{Var}_{H} \left[ \sum_{j=1}^s \mathsf{err}_j \right] \leq \frac{2s}{2^t}.
$$    
\end{claim}
\begin{proof}
We define $A_{j,i}$ as $\01$-valued random variable for the event that $y_j \in B(i)$ or $A_{j,i} = [y_j \in B(i)]$, and $D_{j,i}$ as the indicator variable for the event that $p_{y_j}$ is not the largest error rate in $B(i)$ or $D_{j,i} = [p_{y_j} \neq {E}_{i}']$.
We then have that
\begin{align}
    E_i = \sum_{j \in [2^{2n}]} p_{y_j} A_{j,i}, \quad {E}_i' = \sum_{j \in [2^{2n}]} p_{y_j} A_{j,i} (1 - D_{j,i})
    \label{eq:expressions_rvs_buckets}
\end{align}
By the Cauchy-Schwarz inequality follows that
\begin{align*}
    \sum_{i=1}^s \mathsf{err}_i \leq \sqrt{s} \sqrt{\sum_{i=1}^s \mathsf{err}_i^2}.
\end{align*}
We can bound $\sum_{i=1}^s \mathsf{err}_i^2$ as follows:
\begin{align*}
    \sum_{i=1}^s \mathsf{err}_i^2 \leq \sum_{i=1}^{2^t} \left(E_i - {E}_i' \right)^2 &= \sum_{i=1}^{2^t} \left( \sum_{j \in [2^{2n}]} p_{y_j} A_{j,i} D_{j,i} \right)^2 \\
    &= \sum_{i=1}^{2^t} \sum_{j,k \in [2^{2n}]} p_{y_j} p_{y_k} A_{j,i} A_{k,i} D_{j,i} D_{k,i} \\
    &\leq \sum_{j,k \in [2^{2n}]} p_{y_j} p_{y_k} \sum_{i=1}^{2^t} A_{j,i} A_{k,i} D_{j,i} \\
    &= \sum_{j,k \in [2^{2n}]} p_{y_j} p_{y_k} \sum_{i=1}^{2^t} [y_j,y_k \in B(i)] [y_j \neq {E}'_{B(i)}] \\
    &= \sum_{j,k \in [2^{2n}]} p_{y_j} p_{y_k} [y_j,y_k \in B(i) \text{ for some } i] [p_{y_j} \neq {E}'_{i} \text{ when } p_{y_j} \in B(i)]
\end{align*}
where in the first line we used that the top $s$ buckets is a subset of the total $2^t$ buckets and  Eq.~\eqref{eq:expressions_rvs_buckets}; and in the third line we used that $D_{k,i} \in \{0,1\}$. Let  $F_{j,k}:=[y_j,y_k \in B(i) \text{ for some } i]$, i.e, $F_{j,k}$ indicates if $y_j$ and $y_k$ collide in the same bucket. Let $D_j:=[p_{y_j} \neq {E}'_{i} \text{ when } p_{y_j} \in B(i)]$, i.e., $D_j$ indicates that $y_j$ does not correspond to the biggest error rate in its bucket. Taking expectation over $H$ gives us
\begin{align}
    \nonumber\Exp_H\left[ \sum_{i=1}^s \mathsf{err}_i^2 \right] \leq \Exp_H\left[ \sum_{j,k \in [2^{2n}]} p_{y_j} p_{y_k} F_{jk} D_j \right] &= \Exp_H\left[ \sum_{j \in [2^{2n}]} p_{y_j}^2 D_j \right] + \Exp_H\left[ \sum_{j\neq k \in [2^{2n}]} p_{y_j} p_{y_k} F_{jk} D_j \right] \\
    &\leq \underbrace{\sum_{j \in [2^{2n}]} p_{y_j}^2 \Exp_H\left[ D_j \right]}_{(*)} + \underbrace{\sum_{j\neq k \in [2^{2n}]} p_{y_j} p_{y_k} \Exp_H\left[ F_{jk} \right]}_{(**)}\label{eq:ub_sum_collision_error_2norm}
\end{align}
where in the second line we used that $D_j \in \{0,1\}$. We analyze the terms $(*)$ and $(**)$ separately. First, we deal with $(*)$, 
\begin{align}\nonumber
    (*)&=\sum_{j \in [2^{2n}]} p_{y_j}^2 \Exp_H\left[ D_j \right]=\sum_{j \in [2^{2n}]} p_{y_j}^2 \Exp_H\left[ \max_{i<j}\{F_{j,i}\} \right]\leq \sum_{j \in [2^{2n}]} p_{y_j}^2 \Exp_H\left[ \sum_{i<j}F_{j,i} \right]\\
    &\leq \sum_{j =1}^{2^{2n}} p_{y_j}^2 \frac{j-1}{2^t} \leq \frac{1}{2^t} \sum_{j=2}^{2^{2n}} \sum_{k=1}^{j-1} p_{y_j} p_{y_k} \leq \frac{1}{2^t} \left( \sum_{j=1}^{2^{2n}} p_{y_j} \right)^2 
    = \frac{1}{2^t}, \label{eq:first_sum_ub}
\end{align}
where the first line we have used that $y_j$ does not correspond to the biggest error rate of its bucket if and only if there is an $i<j$ such that $y_i$ is in the same bucket as $y_j$; and in the  line we used \cref{prop:symplectic_hashing} (b) and that $p_{y_j}<p_{y_k}$ if $j>k$. 

The term $(**)$ can be bounded as
\begin{align} \label{eq:second_sum_ub}
    (**)=\sum_{j\neq k \in [2^{2n}]} p_{y_j} p_{y_k} \Exp_H\left[ F_{jk} \right] = \frac{1}{2^t} \sum_{j\neq k \in [2^{2n}]} p_{y_j} p_{y_k} \leq \frac{1}{2^t} \left(\sum_{j \in [2^{2n}]} p_{y_j} \right)^2 = \frac{1}{2^t},
\end{align}
where in the first equality we used Proposition~\ref{prop:symplectic_hashing} (b). 
We can now substitute Eq.~\eqref{eq:first_sum_ub} and Eq.~\eqref{eq:second_sum_ub} into Eq.~\eqref{eq:ub_sum_collision_error_2norm} to obtain
\begin{align}
    \Exp_H\left[ \sum_{i=1}^s \mathsf{err}_i^2 \right] \leq \frac{2}{2^t}.
    \label{eq:ub2_sum_collision_error_2norm}
\end{align}
In order to prove our expectation upper bound now, we use
\begin{align}
    \Exp_H \left[ \sum_{i=1}^s \mathsf{err}_i \right] \leq \sqrt{s} \, \Exp_H\left[ \sqrt{\sum_{i=1}^s \mathsf{err}_i^2} \right] \leq \sqrt{s} \sqrt{\Exp_H\left[ \sum_{i=1}^s \mathsf{err}_i^2 \right]} \leq \sqrt{\frac{2s}{2^t}},
\end{align}
where we used Jensen's inequality in the second inequality and Eq.~\eqref{eq:ub2_sum_collision_error_2norm} in the last inequality. This proves the result regarding expectation. To prove the result regarding variance, we note that
\begin{align}
    \Var_H \left[ \sum_{i=1}^s \mathsf{err}_i \right]  \leq  \Exp_H\left[ \left(\sum_{i=1}^{s}\mathsf{err}_i\right)^2 \right] \leq  s\Exp_H\left[ \sum_{i=1}^{s}\mathsf{err}_i^2 \right]\leq \frac{2s}{2^t}.
\end{align}
where we again used Eq.~\eqref{eq:ub2_sum_collision_error_2norm} in the last inequality.
\end{proof}

We are now ready to give the proof of Lemma~\ref{lemma:hashing_error_PC}.
\begin{proof}[Proof of Lemma~\ref{lemma:hashing_error_PC}]
From Eq.~\eqref{eq:ub_hashing_error_PC}, we have that $\mathsf{err}(\calE;H,s) \leq \sum_{i=1}^s \mathsf{err}_i$. Then,
\begin{equation}
    \Pr_H\left[ \mathsf{err}(\calE;H,s) \leq \varepsilon \right] \geq \Pr_H\left[ \sum_{i=1}^s \mathsf{err}_i \leq \varepsilon \right]
    \label{eq:ub_prob_hashing_error}
\end{equation}
Using Chebyshev's inequality along with Claim~\ref{claim:exp_var_collision_errors_PC}, we have that for any $a > 0$

\begin{align*}
    &\Pr_H\left[ \sum_{i=1}^s \mathsf{err}_i \geq \Exp_H\left[\sum_{i=1}^s \mathsf{err}_i \right] + a \sqrt{\Var_H\left[ \sum_{i=1}^s \mathsf{err}_i \right]} \right] \leq \frac{1}{a^2} \\
    \quad &\Pr_H\left[ \sum_{i=1}^s \mathsf{err}_i \geq \frac{\sqrt{2s} + a\sqrt{2s}}{\sqrt{2^t}} \right] \leq \frac{1}{a^2} \\
    \implies \quad &\Pr_H\left[ \sum_{i=1}^s \mathsf{err}_i \geq \left(1 + \frac{a \sqrt{s}}{\sqrt{s}}\right) \varepsilon \right] \leq \frac{1}{a^2}
\end{align*}
where we substituted $t = \log(2s/\varepsilon^2)$ in the third line. Finally setting $a=5$ and noting that $s \geq 1$, we can combine the above expression with Eq.~\eqref{eq:ub_prob_hashing_error} to show that
\begin{equation}
    \Pr_H\left[ \mathsf{err}(\calE;H,s) \leq 6 \varepsilon \right] \geq 0.96,
\end{equation}
which completes the proof of the desired result.
\end{proof}

\paragraph{Estimation error.} Recall that our energy estimates of the different buckets were $\{\widetilde{E}(b)\}_{b \in \FF_2^t}$. Let the energy estimates of the buckets be indexed in non-increasing order as $\widetilde{E}_1 \geq \widetilde{E}_2 \geq \cdots \geq \widetilde{E}_{2^t}$. Let us denote the estimated energy over the top $s$ buckets as
$$
\Gamma = \sum_{j=1}^s \widetilde{E}_j.
$$
We now bound the error between our estimate $\Gamma$ and $\mathsf{Energy}(\calE;s)$.
\begin{claim} \label{claim:energy_est_error_PC}
Suppose $H$ is a random subgroup of dimension $t$. If $t \geq \log(2s/\varepsilon^2)$, then with query complexity $O\left(\frac{s^2}{\varepsilon^6} \left(\log\frac{1}{\varepsilon} + \log s \right)\right)$, we have
$$
\Pr_H\left[ |\Gamma - \mathsf{Energy}(\calE;s)| \leq \varepsilon \right] \geq 0.92 \, .
$$
\end{claim}
\begin{proof}
We note that 
\begin{align*}
    |\Gamma - \mathsf{Energy}(\calE;s)| = \left| \sum_{j=1}^s (\widetilde{E}_j - p_{y_j}) \right| &\leq \left| \sum_{j=1}^s (\widetilde{E}_j - E_j) \right| + \left| \sum_{j=1}^s (E_j - p_{y_j}) \right|
\end{align*}
The second term on the right hand side is the hashing error $\mathsf{err}(\calE;H,s)$, which by Lemma~\ref{lemma:hashing_error_PC} is bounded above by $\varepsilon/2$ for $t \geq \log (2s/\varepsilon^2)$ with probability more than $0.96$. The first time on the right hand side can be bounded as 
\begin{align*}
    \left| \sum_{j=1}^s (\widetilde{E}_j - E_j) \right| \leq \sum_{j=1}^s \left|\widetilde{E}_j - E_j \right| \leq \sum_{j=1}^{2^t} \left|\widetilde{E}_j - E_j \right| \leq \varepsilon/2,
\end{align*}
where we used the fact that top $s$ buckets are a subset of all the buckets in the second inequality and concluded the final inequality by using by Claim~\ref{claim:energy_est_PC} along with a union bound for $t \geq \log (2s/\varepsilon^2)$. This consumes $O(2^{2t}/\varepsilon^2 \log(4^t)) = O\left(\frac{s^2}{\varepsilon^6} \left(\log\frac{1}{\varepsilon} + \log s \right)\right)$ queries to ensure desired error and success probability greater than $0.96$. Combining the bounds on the above two sums gives us our desired~result.
\end{proof}

\paragraph{Proof of main theorem on testing sparse Pauli channels.} We can now complete the proof of Theorem~\ref{thm:testing_sparse_PC}, and thereby show correctness of Algorithm~\ref{algo:testing_sparsity_PC}.
\begin{proof}[Proof of Theorem~\ref{thm:testing_sparse_PC}]
Using Claim~\ref{claim:energy_est_error_PC}, we can ensure that $|\Gamma - \mathsf{Energy}(\calE;s)| \leq (\varepsilon_2 - 2\varepsilon_1)/3$ with probability greater than $0.92$, using $O\left(\frac{s^2}{(\varepsilon_2-\varepsilon_1)^6} \left(\log\frac{1}{(\varepsilon_2 - \varepsilon_1)} + \log s \right)\right)$ queries. Thanks to \cref{lemma:sparse_paulis_struct}, this is enough to determine if a Pauli channels is close or far from sparse.
\end{proof}

\subsection{Testing sparse Hamiltonians via Pauli hashing}
In this section, we describe how to test the sparsity of Hamiltonians by testing the sparsity of a related Pauli channel and then using Pauli hashing. 
Given a Hamiltonian $H=\sum\lambda_x\sigma_x$ its time evolution channel at time $t$ is given by 
\begin{equation}
    \calH_t(\rho) = U(t) \rho U(t)^\dagger,
\end{equation}
where $U(t)=e^{-itH}$. 
Considering the Pauli decomposition of the unitary $U(t) = \sum_{x \in \{0,1\}^{2n}} \widehat{U}(x) \sigma_x$ and its complex conjugate $U(t) = \sum_{x \in \{0,1\}^{2n}} \overline{\widehat{U}(x)} \sigma_x$, we can write
\begin{equation}
    \calH_t(\rho) = \sum_{x,y \in \{0,1\}^{2n}} \widehat{U}(x) \overline{\widehat{U}(y)} \sigma_x \rho \sigma_y = \sum_{x,y \in \{0,1\}^{2n}} \kappa_{x y} \sigma_x \rho \sigma_y,
\end{equation}
where $\kappa_{x y}$ are coefficients corresponding to the Pauli characters $\sigma_x \rho \sigma_y$. 
Because of Taylor theorem (\cref{eq:TaylorOrder1}), the diagonal coefficients $\kappa_{x x}$ are given by
\begin{equation} \label{eq:diagonal_pauli_coeffs_UHt}
    \kappa_{x x} = \begin{cases}
        |\widehat{U}(x)|^2 = \lambda_x^2 t^2 + o(t^3) &, x \neq 0, \\
        1 - \sum_{x \in \{0,1\}^{2n} \setminus 0^{2n}} |\widehat{U}(x)|^2 = 1 - t^2 \sum_{x \in \{0,1\}^{2n}\backslash 0^{2n}} \lambda_x^2 + o(t^3) &, x = 0,
    \end{cases}
\end{equation}
We can however convert the Hamiltonian evolution channel $\calH_t$ into a diagonal Pauli channel by applying \emph{Pauli twirling}~\cite{cai2019constructing,berg2023techniques} as follows
\begin{equation}
\calH_t^{\calT}(\rho) = \mathbb E_{x}[\sigma_x\mathcal{H}_t(\sigma_x\rho\sigma_x)\sigma_x].    
\end{equation}
Particularly, $\calH_t^{\calT}$ takes the following form
\begin{equation} \label{eq:twirled_channel_H}
    \calH_t^{\calT}(\rho) = \sum_{x \in \{0,1\}^{2n}} p_t^\calT(x) \sigma_x \rho \sigma_x,
\end{equation}
where we have denoted the error rates of the channel $\calH_t^\calT$ as $\{p_t^\calT(x)\}$ and they are related to $\kappa_{x x}$ as $p(x) = \kappa_{x x}$.

We now provide intuition for how testing the sparsity of $n$-qubit Hamiltonians can be accomplished by testing the sparsity of the Pauli channel $\calH_t^{\calT}$. 
Let us order the Pauli coefficients of $U(t)$ excluding $\widehat{U}(0^{2n})$ in a non-increasing order as $|\widehat{U}_1|^2 \geq |\widehat{U}_2|^2 \geq \cdots \geq |\widehat{U}_{2^{2n}-1}|^2$. From Lemma~\ref{lem:sparsitydiscrepancy}, we noted that the following quantity was a good proxy for testing whether $H$ was close to being $s$-sparse or not
\begin{equation}
    \topenergy(t;s):= |\widehat{U}(0^{2n})|^2 + \sum_{j=1}^s |\widehat{U}_j|^2,
    \label{eq:top_energy_ham}
\end{equation}
Moreover, we observe that \footnote{This is true provided that $0^{2n}$ is not among the top $s$ error rates of $\calH_t^\calT$. We will make a more general statement later that removes this assumption.}
\begin{equation}
    \topenergy(t;s) = p_t^\calT(0^{2n}) + \mathsf{Energy}(\calH_t^\calT;s),
    \label{eq:top_energy_ham_error_rates}
\end{equation}
where $\mathsf{Energy}(\calE;s)$ is estimated in Algorithm~\ref{algo:testing_sparsity_PC} as part of testing Pauli channels $\calE$ and is equal to the sum of the top $s$ error rates of the channel.  We can thus test the sparsity of $n$-qubit Hamiltonians by testing the sparsity of the Pauli channel $\calH_t^{\calT}$, which leads to the following result which we will prove shortly.
\begin{theorem}\label{thm:sparsity_testing_pauli_hashing_ham}
Let $H$ be an $n$-qubit Hamiltonian such that $\norm{H}_{\mathrm{op}} \leq 1$. We can test if $H$ is $\varepsilon_1$-close (in normalized Frobenius norm) or $\varepsilon_2$-far from $s$-sparse Hamiltonians with probability $>0.9$ using 
$$
O\left(\frac{s^{14}}{(\eps_2^2-\eps_1^2)^{18}} \left(\log \frac{s^2}{(\eps_2^2-\eps_1^2)^3} + \log s\right) \right)
$$ queries to the evolution operator $U(t) = \exp(-iHt)$ for a choice of $t = O((\varepsilon_2^2 - \varepsilon_1^2)/s)$ and with a total evolution time of $O((s^{13} \log s)/(\eps_2^2-\eps_1^2)^{17})$.
\end{theorem}

\subsubsection{Testing algorithm}
We now proceed as we had for testing sparse Pauli channels in the earlier section. We define the relevant random coset structure and how to estimate energy across different buckets, which will be the main component of our testing algorithm. The main difference here is that we do not have direct access to the relevant Pauli channel i.e., the Pauli-twirled Hamiltonian evolution channel $\calH_t^\calT$ but we show that we can still compute energy estimates across the different buckets corresponding to $\calH_t^\calT$ by querying the Hamiltonian evolution channel $\calH_t$ itself.

\paragraph{Hashing and coset structure.} As in Section~\ref{sec:hashing_random_cosets}, we define the pairwise independent hashing process corresponding to a random subgroup $G$ of dimension $d$, generated by uniformly sampling $\{g_j\}_{j \in d}$ from $\FF_2^{2n}$ randomly. In a slight abuse of notation, we will also call the subgroup of Weyl operators corresponding to the strings in $G$ as $G$ as well and similarly for the generators. We then define the buckets corresponding to $b \in \FF_2^d$ accordingly as
\begin{equation} \label{eq:buckets_ham}    
C(b) := \{ x \in \FF_2^{2n} : [x, g_j] = b_j \, \forall j \in [d] \}.
\end{equation}
Let us denote the energy of bucket $C(b)$ for each $b\in \FF_2^d$ as the sum of the error rates of the Pauli channel $\calH_t^{\calT}$~(Eq.~\eqref{eq:twirled_channel_H}) obtained from Pauli twirling of the Hamiltonian evolution channel in the corresponding coset $a + \commutant(G)$, expressed as
\begin{equation} \label{eq:energy_bucket_ham}
E(b) = \sum_{x \in a + \commutant(G)} p_t^\calT(x).
\end{equation}

\paragraph{Energy estimation.} Recall from Eq.~\eqref{eq:energy_bucket_PC} that the enery of each bucket can be expressed as
\begin{align}
E(b) = \frac{1}{|\commutantperp(G)|} \sum_{z \in G} \lambda_t^\calT(z) (-1)^{[z,a]},
\label{eq:energy_bucket_ham_PC}
\end{align}
where $\lambda_t^\calT$ are the Pauli fidelities of the channel $\calH_t^\calT$. We now discuss how the energy of each bucket is computed efficiently using Pauli eigenstates and Pauli measurements.

\begin{claim}
Let $G$ be a random subspace of dimension $d$. Then using Algorithm~\ref{algo:energy_estimation_ham}, the energy $E(b)$ of each bucket $b \in \FF_2^d$ can be estimated to within error $\varepsilon$ with probability $1-\delta$ with $O(2^d/\varepsilon^2 \cdot \log(2^d/\delta))$ queries. This only requires preparation of Pauli eigenstates and Pauli measurements.
\label{claim:energy_est_ham}
\end{claim}
\begin{proof}
Let the eigenvalue measurement outcome obtained in line~\ref{algo_line:meas_outcome_energy_est_ham} of Algorithm~\ref{algo:energy_estimation_ham} be denoted as $(-1)^w$ with $w \in \{0,1\}$. The joint probability of obtaining $a \in \FF_2^{2n}$ (which are sampled uniformly at random) and obtaining measurement outcome $w \in \{0,1\}$ is then
\begin{align*}
    \Pr(a,w) = \frac{1}{2^{2n}} \Tr \left[ \frac{\Id + (-1)^w\sigma_z}{2} \sigma_a \calH_t\left(\sigma_a \frac{\Id_{2^n} + \sigma_z}{2^n} \sigma_a \right) \sigma_a \right]
\end{align*}
Marginalizing over $a \in \FF_2^{2n}$ gives us the probability of obtaining measurement outcome $w$
\begin{align*}
    \Pr(w) &= \Exp_{a \in \FF_2^{2n}} \Tr \left[ \frac{\Id + (-1)^w\sigma_z}{2} \sigma_a \calH_t\left(\sigma_a \frac{\Id_{2^n} + \sigma_z}{2^n} \sigma_a \right) \sigma_a \right] \\
    &= \Tr \left[ \frac{\Id + (-1)^w\sigma_z}{2} \Exp_{a \in \FF_2^{2n}} \left[\sigma_a \calH_t\left(\sigma_a \frac{\Id_{2^n} + \sigma_z}{2^n} \sigma_a \right) \sigma_a \right] \right] \\
    &= \Tr \left[ \frac{\Id + (-1)^w\sigma_z}{2} \calH_t^\calT \left(\frac{\Id_{2^n} + \sigma_z}{2^n}\right) \right] \\
    &= \frac{1}{2} \Tr \left[ \calH_t^\calT \left(\frac{\Id_{2^n} + \sigma_z}{2^n}\right) \right] + \frac{(-1)^w}{2} \Tr \left[\sigma_z \calH_t^\calT \left(\frac{\Id_{2^n} + \sigma_z}{2^n}\right) \right] \\
    &= \frac{1}{2} + \Tr \left[ \calH_t^\calT \left(\frac{\sigma_z}{2^n}\right) \right] + \frac{(-1)^w}{2} \Tr \left[\sigma_z \right] + \frac{(-1)^w}{2} \Tr \left[\sigma_z \calH_t^\calT \left(\frac{\sigma_z}{2^n}\right) \right] \\
    &= \frac{1 + (-1)^w \lambda_t^{\calT}(z)}{2}
\end{align*}
where the third equality follows from the definition of Pauli twirling and definition of $\calH_t^\calT$~ in Eq.~\eqref{eq:twirled_channel_H}, and we use the definition of Pauli fidelity in the last equality. Observing that $\Exp[(-1)^w] = \lambda_t^{\calT}(z)$, we can then obtain an estimate $\lambda_t^{\calT}(z)$ using the outcomes of $w$ and taking an empirical mean of $(-1)^w$. The estimate will be within $\varepsilon$ error with probability $1-\delta$ using $O(1/\varepsilon^2 \log(1/\delta))$ queries. We can now conclude as we had done in the proof of Claim~\ref{claim:energy_est_PC} and noting Eq.~\eqref{eq:energy_bucket_ham_PC}.
\end{proof}

\begin{algorithm}[H]
\caption{Energy estimation of buckets for Hamiltonian evolution} \label{algo:energy_estimation_ham}
\textbf{Input}: Budget $T=O(2^d/\varepsilon^2 \log(2^d/\delta))$, access to unitary evolution $U(t)=\exp(-iHt)$, evolution time $t$, subspace $G$ of dimension $d=O(\log s)$, set of buckets $B = \{b\} \subseteq \FF_2^d$
\begin{algorithmic}[1]
    \State Initialize energy estimates $\widetilde{E}(b) = 0$
    \State Initialize fidelity estimates $\widetilde{\lambda}_t^\calT(z) = 0$ and counter $m(z)$ for each $z \in G$
    \For{$z \in G$}
    \For{query $t=1,\ldots,T$}
        \State Uniformly sample at random $a \in \FF_2^{2n}$
        \State Prepare Pauli eigenstate $\rho_z$ of $\sigma_z$
        \State Apply Pauli $\sigma_a$ to $\rho_z$
        \State Apply unitary evolution $U(t)=\exp(-iHt)$
        \State Apply Pauli $\sigma_a$
        \State Measure current state with respect to the Pauli basis $\sigma_z$ to obtain eigenvalue $\gamma \in \{\pm 1\}$ \label{algo_line:meas_outcome_energy_est_ham}
        \State Update $\widetilde{\lambda}_t^\calT(z) \leftarrow \widetilde{\lambda}_t^\calT(z) + \gamma$, $m(z) \leftarrow m(z) + 1$
    \EndFor 
    \State Set $\widetilde{\lambda}_t^\calT(z) \leftarrow \frac{1}{m(z)} \widetilde{\lambda}_t^\calT(z)$
    \EndFor 
    \State Set $\widetilde{E}(b) \leftarrow \frac{1}{|G|} \sum_{z \in H} \widetilde{\lambda}_t^\calT(z) (-1)^{[z,a]}$ where $C(b)$ corresponds to coset $a + \commutant(G)$, for each $b \in B$
\end{algorithmic}
\textbf{Output}: $\{\widetilde{E}(b)\}_{b \in B}$
\end{algorithm}

We are now ready to give our tester which is showcased in Algorithm~\ref{algo:testing_sparsity_H}.
\begin{algorithm}[H]
\caption{Testing $s$-sparsity of Hamiltonians} \label{algo:testing_sparsity_H}
\textbf{Input}: Access to unitary Hamiltonian evolution $U(t) = \exp(-iHt)$, evolution time $t$, sparsity $s$, error parameters $0 < \varepsilon_1 < \varepsilon_2 \leq 1$
\begin{algorithmic}[1]
    \State Randomly sample a subgroup $H$ of dimension $d=O(\log s)$
    \State Set bucket indices $\{b\}=\FF_2^d$
    \State Set $\varepsilon = (\varepsilon_2^2 - \varepsilon_1^2)^3/6s^2$ and query budget $T= O(\poly(s)/\varepsilon^2 \log(1/\varepsilon))$
    \State Obtain energy estimates $\{\widetilde{E}(b)\}_{b \in \FF_2^d}$ using Algorithm~\ref{algo:energy_estimation_ham} upon inputs of $U(t)$, evolution time $t$, and budget $T$
    \State Compute $\Gamma \leftarrow \widetilde{E}(0^d) + \max \limits_{S \subset \FF_2^d \setminus 0^d: |S| = s} \sum_{j \in S} \widetilde{E}(j)$
    \State If $\Gamma \geq 1-\eps_1^2\frac{(\eps_2^2-\eps_1^2)^2}{s^2}-\frac{1}{2}\frac{(\eps^2_2-\eps_1^2)^3}{s^2}$, accept.
    \State If $\Gamma \leq 1-\eps_2^2\frac{(\eps_2^2-\eps_1^2)^2}{s^2}+\frac{1}{2}\frac{(\eps^2_2-\eps_1^2)^3}{s^2}$, reject.
\end{algorithmic}
\textbf{Output}: \textsf{FLAG} for accepting/rejecting $H$
\end{algorithm}

\subsubsection{Analysis}
In this section, we prove Theorem~\ref{algo:testing_sparsity_H} and thereby show correctness of Algorithm~\ref{algo:testing_sparsity_H}.
\paragraph{Hashing error.}
As in Section~\ref{sec:analysis_sparse_testing_PC}, we first analyze the error in estimating the energy of the top $s$ error rates $\mathsf{Energy}(\calH_t^\calT;s)$ via our random hashing process. We assume here that all the energy estimates across the cosets are exact and analyze the estimation error later.

Let the energies of the buckets $\{C(b)\}_{b \in \FF_2^{d} \setminus 0^d}$ be indexed in non-increasing order by energy as $E_1 \geq E_2 \geq \cdots \geq E_{2^d-1}$. We pay special attention to the bucket corresponding to the coset $C(0^d)$ which will always include the error rate of $p(0^{2n})$ due to the construction of our buckets. We will also use $\tilde{E}_j = \max_{x \in C(b)} p(x)$ to denote the energy of largest coefficient hashed into the $j$th bucket denoted by $C(b)$ for some $b \in \FF_2^{d}$. The true values of the energies are clearly the error rates $\{p_t^\calT(x)\}_{x \in \FF_2^{2n}}$ themselves. We will also order the error rates $\{p_t^\calT(x) \}_{ \in \FF_2^{2n} \setminus 0^{2n}}$(i.e., excluding $p(0^{2n})$) in a non-increasing order as $p_1^\calT \geq p_2^\calT \geq \cdots \geq p_{2^{2n}-1}^\calT$.  The goal is to obtain an estimate of $\topenergy(t;s)$~(Eq.~\eqref{eq:top_energy_ham_error_rates}) from the constructed buckets. In particular, we want to determine the value of $p(0^{2n}) + \sum_{j=1}^s p_{y_j}$ to obtain such an estimate. The \emph{hashing error} is then accordingly defined as
\begin{equation*}
\mathsf{err}(\calH_t^\calT;G,s) = \left(E(0^d) - p(0^{2n}) \right) + \sum_{j=1}^s \left( E_j - p_{y_j} \right),    
\end{equation*}
where we denoted the error by $\mathsf{err}(\cdot)$ along with noting the random subspace $G$ involved in Pauli hashing. As $p_1^\calT,\ldots,p_s^\calT$ are the $s$ largest error rates in $\FF_2^{2n} \setminus 0^{2n}$, it follows that 
\begin{equation*}
    \mathsf{err}(\calH_t^\calT;G,s) = \left(E(0^d) - p(0^{2n}) \right) + \sum_{j=1}^s \left( E_j - {E}'_j \right) + \sum_{j=1}^s \left( {E}'_j - p_{y_j} \right) \leq \left(E(0^d) - \tilde{E}(0^d) \right) + \sum_{j=1}^s \left( E_j - {E}'_j \right)
\end{equation*}
The following corollary of Lemma~\ref{lemma:hashing_error_PC} is then immediate.
\begin{corollary}\label{corr:hashing_error_ham}
Fix $\varepsilon \in (0,1/2]$. If $G$ is a random subgroup of dimension $t\geq \log (2s/\varepsilon^2)$,~then
$$
\Pr_{G}\left[\mathsf{err}(\calH_t^\calT;G,s) \leq 6\varepsilon \right] \geq 0.96 \,.
$$  
\end{corollary} 
We remark that the proof of Corollary~\ref{corr:hashing_error_ham} is very similar to that of Lemma~\ref{lemma:hashing_error_PC} and is thus not included. Note that the proof of Lemma~\ref{lemma:hashing_error_PC} bounds the error of any $s$ buckets. In particular, it could be used over the $s+1$ buckets that include error rates corresponding to $0^{2n}$ (or bucket of $\commutant(G)$) and the top $s$ buckets different from $\commutant(G)$.

\paragraph{Estimation error.} Recall that our energy estimates of the different buckets were $\{\widetilde{E}(b)\}_{b \in \FF_2^t}$. Let the energies of the buckets $\{C(b)\}_{b \in \FF_2^{d} \setminus 0^d}$ be indexed in non-increasing order as $\widetilde{E}_1 \geq \widetilde{E}_2 \geq \cdots \geq \widetilde{E}_{2^d-1}$. Let us denote the estimated energy over the bucket $C(0^{2n})$ and the top $s$ buckets (excluding $C(0^{2n})$) as
$$
\Gamma = \widetilde{E}(0^d) + \max \limits_{S \subset \FF_2^d \setminus 0^d: |S| = s} \sum_{j \in S} \widetilde{E}(j) = \widetilde{E}(0^d) + \sum_{j=1}^s \widetilde{E}_j,
$$
where the last equality follows from the definition of $\widetilde{E}_j$ above. We now bound the error between our estimate $\Gamma$ and $\mathsf{TopEnergy}_H(t)$~(Eq.~\eqref{eq:top_energy_ham}) by immediately applying Claim~\ref{claim:energy_est_error_PC} to obtain the following corollary.
\begin{corollary}\label{corr:energy_est_error_H}
Suppose $G$ is a random subgroup of dimension $d$. If $d \geq \log(2s/\varepsilon^2)$, then with query complexity $O\left({s^2}/{\varepsilon^6} \left(\log\frac{1}{\varepsilon} + \log s \right)\right)$, we have
$$
\Pr_H\left[ |\Gamma - \topenergy(t;s)| \leq \varepsilon \right] \geq 0.92 \, .
$$
\end{corollary}

\paragraph{Proof of main theorem on testing sparse Hamiltonians.} We can now complete the proof of Theorem~\ref{thm:sparsity_testing_pauli_hashing_ham}, and thereby show correctness of Algorithm~\ref{algo:testing_sparsity_H}.
\begin{proof}[Proof of Theorem~\ref{thm:sparsity_testing_pauli_hashing_ham}]
Let $t=O((\eps_2^2-\eps_1^2)/s)$. By \cref{lem:sparsitydiscrepancy}, we have that if $H$ is $\eps_1$-close to being sparse, then 
\begin{equation*}
    \topenergy(t;s)\geq 1-\eps_1^2\frac{(\eps_2^2-\eps_1^2)^2}{s^2}-\frac{1}{3}\frac{(\eps^2_2-\eps_1^2)^3}{s^2},
\end{equation*}
while if $H$ is $\eps_2$-far from $s$-sparse, then 
\begin{equation*}
    \topenergy(t;s)\leq 1-\eps_2^2\frac{(\eps_2^2-\eps_1^2)^2}{s^2}+\frac{1}{3}\frac{(\eps^2_2-\eps_1^2)^3}{s^2}.
\end{equation*}
For our tester, it then suffices to estimate $\mathsf{TopEnergy}_H(t)$ up to error 
\begin{equation*}
    \eps=\frac{(\eps_2^2-\eps_1^2)^3}{6s^2}.
\end{equation*}
Using Corollary~\ref{corr:energy_est_error_H} and for the specified $t$, we can ensure that $|\Gamma - \mathsf{TopEnergy}_H(t)| \leq \eps$ with probability greater than $0.92$, using $O\left(\frac{s^{14}}{(\varepsilon_2^2-\varepsilon_1^2)^{18}} \left(\log\frac{s^2}{(\varepsilon_2^2 - \varepsilon_1^2)^3} + \log s \right)\right)$ queries.
The decision rules of $\Gamma$ in Algorithm~\ref{algo:testing_sparsity_H} for accepting/rejecting are then evident.
\end{proof}

\section{Lower bounds for learning}
Below we prove lower bounds on learning sparse and local Hamiltonians given query access to the time-evolution operator. In comparison to prior work of~\cite[Theorem~5.1]{bluhm2024hamiltonianv2} the class of Hamiltonians witnessing our lower bound also applies to algorithms which are diagonal and encode just Boolean functions; in contrast the hard instance in~\cite{bluhm2024hamiltonianv2} applies to Hamiltonians involve application of Haar random unitaries. 
In order to prove our lower bounds, we use the following generic lemma, which will allow us to reduce the problem of proving lower bounds for Hamiltonian testing/learning problems to a question about Boolean function analysis. Below we will be talking of the standard oracle model in query complexity, i.e.,
$$
O_f:\ket{x,0}\rightarrow \ket{x,f(x)}.
$$
\begin{theorem}[{\cite[Theorem~14]{gilyen2019optimizing}}]
\label{thm:phaseConv}
		Let $p: X \rightarrow [0,1]$, and suppose $U_p:\mathcal{H}\otimes\mathcal{H}_{\text{aux.}}\rightarrow\mathcal{H}\otimes\mathcal{H}_{\text{aux.}}$ is a probability oracle with an $n$-qubit auxiliary Hilbert space $\mathcal{H}_{\text{aux.}}=\mathbb{C}^{2^n}$.
		Let $\eps\in(0,1/3)$, then we can implement an $\eps$-approximate phase oracle $O$ such that for any phase oracle $\mathrm{O}_p$ defined as
  $$
  O_p:\ket{0} \mapsto \sum_x \sqrt{p(x)}\ket{x},
  $$ 
  and for all $\ket{\psi}\in \mathcal{H}$ 
		$$
  \|{O\ket{\psi}\ket{0}^{\!\otimes (n+a)}-\mathrm{O}_p\ket{\psi}\ket{0}^{\!\otimes (n+a)}}\|_{\mathrm{op}}\leq \eps,
  $$ 
		using $O(\log(1/\varepsilon))$ applications of $U_p$ and~$U_p^\dagger$, with $a=O(\log\log(1/\eps))$, where
  $$
  U_p:\ket{x,0}\rightarrow \ket{x}\otimes \big(\sqrt{p(x)}\ket{\psi_g(x)}\ket{0}+\sqrt{1-p(x)}\ket{\psi_b(x)}\ket{1}\big),
  $$
  where $\ket{\psi_g(x)},\ket{\psi_b(x)}$ are arbitrary orthogonal states.
	\end{theorem}

\begin{lemma}
\label{cor:hamiltoniantofunction1}
	Let $t\geq 0$.	Let $f:\01^n\rightarrow \pmset{}$ be a Boolean function. Given quantum query access to an oracle $O_f$, then we can simulate 
		$$ 
		H_f=e^{it\sum_{S\in \01^n}\widehat{f}(S)\ketbra{S}{S}}
		$$
		for time $t\in \R$ with precision $\eps$ making $\widetilde{O}(t\log(1/\varepsilon))$ queries to $U_f$.
	\end{lemma} 
	\begin{proof}
 We first observe that, using the phase kickback trick, with access to $O_f$, one can also implement the phase oracle. 		Applying $O_f$ on a uniform superposition, we get $\frac{1}{\sqrt{2^n}}\sum_x \ket{x,f(x)}$. Next, one can obtain $\frac{1}{\sqrt{2^n}}\sum_x f(x) \ket{x}$ with probability~$1/2$: replace $f(x)\in\pmset{}$ by $(1-f(x))/2\in\01$ unitarily, apply the Hadamard transform to the last qubit and measure it. With probability $1/2$ we obtain the outcome~0, in which case our procedure rejects. Otherwise the remaining state is $\frac{1}{\sqrt{2^n}}\sum_x f(x) \ket{x}$. So from here onwards, we can assume that we have access to
 $$
 O_f:\ket{0}\rightarrow \frac{1}{\sqrt{2^n}}\sum_x f(x)\ket{x}.
 $$
	Below we show that one use the oracle $U$ to construct a probability oracle $U_p$: 
		\begin{align*}
		U_f:\ket{x}\ket{0}\ket{0} \kern5mm
		&\overset{Had}{\mapsto}\kern12mm 
		\ket{x}\ket{0}\ket{+} \\
		&\overset{O_f}{\mapsto}\kern12mm 
		\ket{x}\frac{1}{\sqrt{2^n}}\sum_{y\in X} {f}(y)\ket{y}\ket{+}\\
  		&\overset{\kern-12mm Had^n\kern-12mm}{\mapsto} \kern12mm
		\ket{x}\sum_{y\in X} \widehat{f}(y)\ket{y}\ket{+}\\
		&\overset{\kern-12mm H_f\kern-12mm}{\mapsto} \kern12mm
		\ket{x}\left(\widehat{f}(x)\ket{x}\ket{-}+ \sum_{y\neq x} \widehat{f}(y)\ket{y}\ket{+}\right)\\
		&\overset{Had}{\mapsto} \kern12mm 
		\ket{x}\left(\widehat{f}(x)\ket{x}\ket{1} + \sum_{y\neq x} \widehat{f}(y)\ket{y}\ket{0}\right)\\
		&\overset{\kern-12mm Swap\kern-12mm}{\mapsto} \kern12mm 
		\ket{x}\left(\widehat{f}(x)\ket{x}\ket{1} + \sum_{y\neq x} \widehat{f}(y)\ket{y}\ket{0}\right).	
		\end{align*}
		As Theorem~\ref{thm:phaseConv} shows we can simulate a fractional phase query $\mathrm{O}^r_f$ where $r:=t/\lceil|t|\rceil$ with precision $\eps/\lceil|t|\rceil$ making $O(\log(t/\eps))$ queries to $U_f$. Observe that $\lceil|t|\rceil$ consecutive applications of $\mathrm{O}^r_f$ give $\mathrm{O}^t_f$, which is exactly the Hamiltonian simulation unitary that we wanted to implement.
	\end{proof}

The same proof as the lemma above implies the following corollary as well.
\begin{corollary}
\label{cor:hamiltoniantofunction}
	Let $t\geq 0$.	Let $p:\01^n\rightarrow [0,1]$ be a distribution. Given quantum query access to an oracle $O_p$
	 we can simulate a Hamiltonian corresponding to the probability distribution
		$$ 
		H_p=e^{it\sum_{S\in \01^n}p(S)\ketbra{S}{S}}
		$$
		for time $t\in \R$ with precision $\eps$ making $\widetilde{O}(t\log(1/\varepsilon))$ queries to $U_p$.
	\end{corollary}

\subsection{Adaptive coherent memoryless learning lower bound}
To prove our main theorem, we use the following facts, starting with a well known bound on the size of an $\varepsilon$-net of the $n$-dimensional sphere.
\begin{fact}[{\label{fac:epsnet}\cite[Exercise 2.3.1]{tao2012topics}}]
For every $d \ge 1$ and any $0<\varepsilon<1/2$ there exists an $\varepsilon$-net of the sphere $S^{d-1}$ of cardinality at least
$t=(c/\varepsilon)^d$, i.e., there exists $\{v_1,\ldots,v_t\}\subseteq S^{d-1}$ such that $\|v_i-v_j\|_2\geq \varepsilon$ for all $i\neq j$.
\end{fact}

\begin{theorem}
    There exists a class of Hamiltonians with $\|H\|_2 \leq 1$  such that learning $s$-sparse $n$-qubit Hamiltonians without quantum memory  upto error $\varepsilon$ using time step $t$, need to make $\Omega\big(s (\log 1/\varepsilon)/(t\log s)\big)$ adaptive quantum queries.
\end{theorem}

\begin{proof}
   Let's assume that the Hamiltonian is supported on the \emph{first} $(\log s)$ qubits and only has support on the $\sigma_0,\sigma_3$ (i.e., it is diagonal), so the total sparsity equals $2^{(\log s)}=s$ (note that in the original learning algorithm, the support is \emph{unknown} to the learner so we are proving a lower bound on a simpler problem here). Let the corresponding Pauli coefficients be $\{\lambda_x:x\in \{0,3\}^{(\log s)}\}$.  Now consider an $\varepsilon$-net on these qubits to be the set of vectors $\{h^1,\ldots,h^t\}\subseteq \{0,3\}^{\log s}$ where $t=(1/\varepsilon)^{s}$. Now, let us consider the class of Hamiltonians $\mathcal{H}$, given by
    $$
    H_i=\sum_{x\in \{0,3\}^{(\log s)}} h^i(x)\sigma_x.
    $$
Observe that the unitary evolution corresponding to this Hamiltonian is given by
$$
U_p=e^{it\sum_{x\in \01^n}p(x)\ketbra{x}{x}}
$$
By Corollary~\ref{cor:hamiltoniantofunction},
every  learning algorithms that made queries to $U_p$ at time $t$, can be converted (with a factor $t$-overhead) to an algorithm that only is given \emph{quantum queries to $p$}, i.e., given access to the standard oracle model $O_p:\ket{0}\rightarrow \sum_x \sqrt{p(x)} \ket{x}$. 
 Clearly a learning algorithm for learning $H_i$ by making queries to $U_p$ upto error-$\varepsilon$  in the $\ell_2$-distance implies that the learning algorithm can \emph{identify} $i$ since the coefficient vectors $h^i$ form an $\varepsilon$-net. 

Now we use Holevo's bound to conclude the proof (we will use a version that appeared in Nayak's work~\cite{nayak1999optimal} and also recently in a work of Chen et al.~\cite{chen2023testing}). Since each quantum query contains $(\log s)$-qubits of information (since we assumed that the Hamiltonian was identity on the remaining $n-\log s$ qubits), if there is a $k$-query algorithm that \emph{identifies} $i$ as above,~then
$$
k\cdot (\log s)\geq \log |\mathcal{H}|,
$$
and now using the lower bound from Fact~\ref{fac:epsnet}, that implies a lower bound of 
$$
k\geq s/(\log s)\cdot \log (1/\varepsilon).
$$
Overall, this implies an $\Omega(s/(t\log s)\cdot \log (1/\varepsilon))$ lower bound on the learning problem.
\end{proof}
One inherent weakness in the proof technique above is, the $(\log s)$-factor in the denominator cannot be removed since Holevo's bound is a generic statement about transmitting arbitrary quantum states. In the next section, we are able to use a more careful analysis and show that one can get rid of the dependence on $(\log s)$-albeit in a weaker model of learning.

\subsection{Non-adaptive incoherent memoryless learning lower bound}

\begin{theorem}
    Learning $s$-sparse $n$-qubit Hamiltonians without quantum memory upto error $1/2s$ using time step $t$,~requires 
    $$
    \Omega\big((s\log s)/t\big)
    $$ non-adaptive quantum queries.
\end{theorem}

In order to prove this theorem, we will need the following 

\begin{proof}
In order to prove this theorem, we first construct our hard instance of Hamiltonians based on Boolean functions. To this end, we embed Boolean functions as Hamiltonians in a natural way: for every $f:\01^n\rightarrow \pmset{}$, let 
$$
H_f=\sum_{S\subseteq [n]}\widehat{f}(S) \prod_{i\in S}Z_i,
$$
in which case $H_f\ket{x}=f(x)\ket{x}$ for every basis state $x\in \01^n$.  Such an embedding was considered and shown to be useful in~\cite{hadfield2021representation}. Observe that the unitary evolution corresponding to this Hamiltonian is given by
$$
U_f=e^{it\sum_{S\in \01^n}\widehat{f}(S)\ketbra{S}{S}}
$$
By Lemma~\ref{cor:hamiltoniantofunction1},
every  learning algorithms that made queries to $U_f$ at time $t$, can be converted to an algorithm that only is given \emph{quantum queries to $f$}, i.e., given access to the standard oracle model $O_f:\ket{x,0}\rightarrow \ket{x,f(x)}$.\footnote{Technically, we need the phase oracle, but the bit-oracle and phase oracle are equivalent up to a constant overhead in query complexity when allowed controlled operations}. 

So from here onwards, we prove a lower bound on the number of quantum queries required for our learning task, and that complexity divided by $t$ will be our eventual Hamiltonian learning lower bound. Our next step is to construct our hard set of Boolean functions. To this end, for every $(\log s)$-dimensional subspace $V$, let $f_V(x)=[x\in V^\perp]$, in which case it is well-known that 
$$
\widehat{f}(S)=[x\in V]/s,
$$
see~\cite{o2014analysis} for a proof. In particular, note that the number of $T$ for which $\widehat{f}(T)\neq 0$ equals $s$. In particular,  the corresponding Hamiltonian $H_f$ whose Pauli coefficients are precisely $\widehat{f}(T)$, has sparsity $s$.  

Our third step now is to consider the class of Boolean functions
$$
\Cc=\{f_V:\01^n\rightarrow \01\vert f_V(x)=[x\in V^\perp] \text{ s.t.} V \text{ is a }(\log s)-\text{dimensional subspace} \}.
$$
We next show that every non-adaptive quantum learning algorithm for learning the unknown $f$ (given quantum query access $O_f$), needs to make $\Omega(k)$ queries.  The proof of this is similar to the information-theoretic proof in~\cite{arunachalam2018optimal}. 		We prove the lower bound for $\Cc$ using a three-step information-theoretic technique. Let $\mathbf{A}$ be a random variable that is uniformly distributed over $\Cc$.  Suppose $\mathbf{A}=f_V$, and let $\mathbf{B}=\mathbf{B}_1\ldots\mathbf{B}_T$ be $T$ quantum queries 
 $$
 \ket{\psi^i_{f_{V}}}= \sum_{x\in \01^{n}}\sqrt{\alpha_i(x)} \ket{x,f_V(x)},
 $$ 
 for $f_V\in \Cc$, where the amplitudes could potentially depend on the $i$th query, but independent of $i-1$ different measurement outcomes.  The random variable $\mathbf{B}$ is a function of the random variable~$\mathbf{A}$. 
		The following upper and lower bounds on $I(\mathbf{A}:\mathbf{B})$ are similar to~\cite[Theorem~12]{arunachalam2018optimal} and we omit the details of the first two steps here.
		\begin{enumerate}
			\item $I(\mathbf{A}:\mathbf{B})\geq \Omega(\log|\Cc|)$ because $\mathbf{B}$ allows one to recover $\mathbf{A}$ with high probability.
			\item $I(\mathbf{A}:\mathbf{B})\leq T\cdot I(\mathbf{A}:\mathbf{B}_1)$ using a chain rule for mutual information.
			
			\item $I(\mathbf{A}:\mathbf{B}_1)\leq O(n\cdot \eta_a)$.\\[1mm]
			\emph{Proof (of 3).} Since $\mathbf{A}\mathbf{B}$ is a classical-quantum state, we have 
			$$
			I(\mathbf{A}:\mathbf{B}_1)= S(\mathbf{A})+S(\mathbf{B}_1)-S(\mathbf{A}\mathbf{B}_1)=S(\mathbf{B}_1),
			$$ 
			where the first equality is by definition and the second equality uses $S(\mathbf{A})=\log |\Cc|$ since $\mathbf{A}$ is uniformly distributed over~$\Cc$, and $S(\mathbf{A}\mathbf{B}_1)=\log |\Cc|$ since the matrix 
			$$
			\sigma=\frac{1}{|\Cc|} \sum_{f_V \in \Cc} \ketbra{c}{c}\otimes \ketbra{\psi_{f_V}}{\psi_{f_V}}
			$$ is block-diagonal with $|\Cc|$ rank-1 blocks on the diagonal. It thus suffices to bound the entropy of the (vector of singular values of the) reduced state of $\mathbf{B}_1$, which~is
			$$
			\rho=\frac{1}{|\Cc|}\sum_{f_V \in \Cc}\ketbra{\psi_{f_V}}{\psi_{f_V}}.
			$$
For notational convenience, let $\eta_{\mathsf{a}}=\mathop{\Exp}_{c,c'\in \Cc}\Pr_x[c(x)\neq c'(x)]$. Also, let $\sigma_0\geq \sigma_1\geq\cdots\geq \sigma_{2^{n+1}-1}\geq 0$ be the singular values of $\rho$. Since~$\rho$ is a density matrix, these form a probability distribution. Now observe that $\sigma_0\geq 1-\eta_a$: consider the vector $u=\frac{1}{|\Cc|}\sum_{c'\in \Cc}\ket{\psi_{c'}}$  and observe that
   \begin{align*}
  u^\top \rho u &=\frac{1}{|\Cc|^3}\sum_{V,V',V''\in \Cc}\langle \psi_{f_V}|\psi_{f_{V'}}\rangle\langle \psi_{f_{V}}|\psi_{f_{V''}}\rangle\\
  &=\Exp_{V} \Big[\Exp_{V'}[\langle\psi_{f_V}|\psi_{f_{V'}}\rangle]\Big]\cdot\Big[\Exp_{V''}[\langle \psi_{f_V}|\psi_{f_{V''}}\rangle]\Big]\\
  &\geq \Big(\mathop{\Exp}_{V,V'}[\langle\psi_{f_V}|\psi_{f_{V'}}\rangle]\Big)\cdot \Big(\mathop{\Exp}_{V,V''}[\langle \psi_{f_{V}}|\psi_{f_{V''}}\rangle]\Big)=\big(\mathop{\Exp}_{{f_V},{f_{V'}}\in \Cc}\Pr_x[{f_V}(x)={f_{V'}}(x)]\big)^2\geq 1-2\eta_a,
   \end{align*}
where the first inequality is by Chebyshev's sum inequality (since all the inner products are non-negative) and the second  inequality  follows from the definition of $\eta_a$. Hence we have that $\sigma_0=\max_{u}\{u^\top \rho u / u^\top u\} \geq 1-2\eta_a$ (where we used that $\|u\|_2\leq 1$). 

   Let $\mathbf{N}\in\{0,1,\ldots,2^{n+1}-1\}$ be a random variable with probabilities $\sigma_0,\sigma_1,\ldots,\sigma_{2^{n+1}-1}$, and $\mathbf{Z}$ an indicator for the event ``$\mathbf{N}\neq 0$.'' Note that $\mathbf{Z}=0$ with probability $\sigma_0\geq 1-2\eta_a$, and $H(\mathbf{N}\mid \mathbf{Z}=0)=0$. By a similar argument as in~\cite[Theorem~15]{arunachalam2018optimal}, we~have 
			\begin{align*}
			S(\rho) & =H(\mathbf{N})=H(\mathbf{N},\mathbf{Z})=H(\mathbf{Z})+H(\mathbf{N}\mid\mathbf{Z})\\
			& =H(\sigma_0)+\sigma_0\cdot H(\mathbf{N}\mid \mathbf{Z}=0) + (1-\sigma_0)\cdot H(\mathbf{N}\mid \mathbf{Z}=1) \\
			& \leq H(\eta_a) + \eta_a(n+1)\\
   &\leq O(\eta_a(n+\log (1/\eta_a)) 
			\end{align*}
			using $H(\alpha)\leq O(\alpha\log (1/\alpha))$.
		\end{enumerate}
		Combining these three steps implies $T=\Omega(\log |\Cc| / (n\eta_a))$.
  
  It now remains to bound $|\Cc|,\eta_a$. To this end, we prove bounds on both quantities below.
  \begin{claim}
			The number of distinct $(\log s)$-dimensional subspaces of $\F_2^n$ is at least $ s^{\Omega(n-\log s)}$.
		\end{claim}
		
		\begin{proof}
			For simplicity below, let $d=\log s$. We can specify a $d$-dimensional subspace by giving $d$ linearly independent vectors in it. The number of distinct sequences of $d$ linearly independent vectors is exactly $(2^n-1)(2^n-2)(2^n-4)\cdots (2^n-2^{d-1})$, because once we have the first $t$ linearly independent vectors, with span $\Se_t$, then there are $2^n-2^t$ vectors that do not lie in $\Se_t$. 
			
			However, we are double-counting certain subspaces in the argument above, since there will be multiple sequences of vectors yielding the same subspace. The number of sequences yielding a fixed $d$-dimensional subspace can be counted in a similar manner as above and we get
			$(2^{d}-1)(2^{d}-2)(2^{d}-4)\cdots (2^{d}-2^{d-1})$.
			So the total number of subspaces is
			$$
			\frac{(2^n-1)(2^n-2)\cdots (2^n-2^{d-1})}{(2^{d}-1)(2^{d}-2)\cdots (2^{d}-2^{d-1})}\geq \frac{(2^n-2^{d-1})^{d}}{(2^{d}-1)^{d}} \geq 2^{\Omega((n-d)d)}= s^{\Omega(n-\log s)},
			$$ 
   where we used $d=\log s$
		\end{proof}
  Next, it remains to \emph{upper} bound $\eta_a$. To this end, first observe that
  \begin{align*}
  \eta_{\mathsf{a}}\leq \max_{c,c'\in \Cc}\Pr_x[c(x)\neq c'(x)]&=\max_{V,V'}\Pr_x[c_V(x)+ c_{V'}(x)\neq 0]\\
  &=\max_{V}\Pr_x[c_V(x)=1]=\frac{1}{2^n}\cdot |V^\perp|=1/s.
  \end{align*}
  Putting everything together, we have shown that, \emph{ exact} learning the concept class $\Cc$ requires 
  $$
T=\Omega(\frac{\log |\Cc| }{ n\eta_a})\geq \Omega(\frac{(\log s)\cdot (n-\log s)}{ n\cdot 1/s})={\Omega}(s\log s).
  $$
  In order to conclude the proof of the theorem, note that every algorithm that satisfies\footnote{Note that technically the learning algorithm, need not output a Boolean function $f':\01^n\rightarrow \pmset{}$, but in our setting since it \emph{knows} that the unknown Hamiltonian is characterized by a Boolean function, even if the algorithm outputs a real-valued function, then rounding it to a bit would also be a good approximation.}
  $$
  \varepsilon\geq \sum_x |\lambda_x-\lambda'_x|^2=\sum_S (\widehat{f}(S)-\widehat{f'}(S))^2=\Exp_x[|f(x)-f'(x)|^2]=\frac{1}{2^n}\sum_x [f(x)\neq f'(x)]=d(f,f')
  $$
We finally conclude by using~\cite[Claim~2.2]{haviv2016list}
  $$
  d(f_V,f_{V'})=\Pr_x[f_V(x)\neq f_{V'}(x)] \geq 1/(2s),
  $$
  so learning this function to error $<1/(2s)$ implies exact learnability.
\end{proof}

\bibliographystyle{alphaurl}
\bibliography{references}
\end{document}